\documentclass[12pt,reqno]{article}

\usepackage{preamble_draft}

\begin{document}

\title{Providing Certainty\footnote{For valuable comments and suggestions, we thank Tilman B\"orgers, Benjamin Brooks, Christoph Carnehl, Yeon-Koo Che, Modibo Camara, Andr\'es Carvajal, Takuma Habu, Emir Kamenica, Jinwoo Kim, Andrew Ferdowsian, Amanda Friedenberg, Athanasios Geromichalos, David Miller, Roi Orzach, Marco Ottaviani, Minseon Park, Alessandro Pavan, Doron Ravid, Arseniy Samsonov, Mu Zhang, and various seminar participants.}
}

\author{
Andrew B. Choi 
\and
Christoph Schlom 
\and
Chengyang Zhu\footnote{Choi: University of Michigan, \href{mailto:choiandy@umich.edu}{choiandy@umich.edu}. Schlom: UC Davis, \href{mailto:cschlom@ucdavis.edu}{cschlom@ucdavis.edu}. Zhu: Boston University, \href{mailto:zhuc@bu.edu}{zhuc@bu.edu}.}
}
\date{June 2026 \\
 \href{https://drive.google.com/file/d/1rQDN9W3Cy9HSq8_MennlfRtJqgQgj-Rj/view?usp=drive_link}{Click here for the latest draft.}}
\maketitle

\begin{abstract}

We introduce a moral hazard model in which public information about a payoff-relevant state arrives over time, an agent decides when to make an irreversible investment, and a principal commits to a state-contingent policy to incentivize investment. To discourage the agent from waiting for more information, the principal's optimal policy \textit{provides certainty}, reducing the degree to which the agent's payoff depends on the state. This is inefficient -- both players would be better off with less certainty. We study when the agent receives positive rent, and when moral hazard delays investment. Our results apply to environmental subsidies and R\&D incentives.



\vspace{2mm}

\end{abstract}

\section{Introduction}\label{sec:motivating-example}

Policymakers often seek to reduce uncertainty about future policies. For example, lawmakers introduce subsidies that remain in place for a fixed number of years, regardless of subsequent fiscal shocks. Governments make advance market commitments, promising to purchase a pre-specified quantity of a vaccine once it is developed, irrespective of realized demand.

We argue that such commitments are solutions to a novel moral hazard problem in which a principal wants an agent to invest, but the agent would rather wait for more information. To encourage investment, it is optimal for the principal to \textit{provide certainty} about her future policy. This is distortionary -- both the principal and the agent would be better off with less certainty.

Our results differ from those of standard moral hazard models (\`a la \citealp{holmstrom1979moral}) in two crucial ways. First, the principal provides certainty, not to insure a risk-averse agent, but because certainty discourages the agent from waiting for information. Second, the agent can obtain positive rent 
without limited liability or private information.

A more detailed overview of the paper will be provided in \cref{subsec:overview}. Here, we begin with a motivating example. Suppose a principal (she) wants an agent (he) to invest in a project. There is an uncertain state of the world $c\in \{1,2\}$ with a common prior probability $p=1/2$ that $c=1$. There are two periods. The principal chooses a policy $y\geq 0$, and the agent chooses when, if ever, to invest. The timeline of the game is as follows:
\begin{itemize}[leftmargin=0cm]\itemsep-.1em
    \item[] \textbf{Period 0:} Principal commits to a state-contingent policy rule $y(c)$.
    \item[] \hphantom{\textbf{Period 0:}} Agent decides whether to invest.
    \item[] \textbf{Period 1:} State $c$ is publicly realized.
    \item[] \hphantom{\textbf{Period 1:}} If agent has not invested, he decides whether to invest. Policy is implemented.
\end{itemize}
The cost of the policy to the principal is $cy$. If the agent invests, he incurs a cost $I=6$, obtains a flow benefit $b=1$ in each period, and also benefits from the principal's policy $y$. That is,
\begin{align*}
        \text{Agent's payoff} &\,= \begin{cases}
            2b - I + y \caseif \text{he invests in period 0} \\
            b - I + y  \caseif \text{he invests in period 1} \\
            0 \caseif \text{he never invests.}
        \end{cases}
\end{align*}
The policy may be interpreted either as a direct subsidy, or as public expenditure in infrastructure that enhances the value of the agent's investment. The uncertainty over the principal's marginal cost $c$ may represent uncertainties in the future fiscal environment.

The principal's objective is to minimize her expected cost from policy, subject to inducing the agent to invest in period 0.\footnote{This need not be because the principal intrinsically prefers early investment. The principal may simply wish to induce investment under both states; then, the principal \textit{must} induce investment in period 0, since the agent will not wait until period 1, forgoing the period-0 flow benefit, only to invest under both states. See \cref{OA:motivating-example} for an analysis of the optimal investment timing in the example.} Note that if the principal sets $y=0$, the agent will never invest.




The agent faces a trade-off between obtaining more flow benefits and obtaining more information about policy. Investing in period 0 allows him to enjoy the flow benefit twice, but if $y$ turns out to be low, he will regret having invested. On the other hand, waiting until period 1 gives him more information: he will be able to observe $c$, infer $y(c)$, and invest only if his payoff from doing so is positive.



\paragraph*{First Best (Contractible Investment).}
Consider a first-best benchmark in which the principal can commit to a policy rule that depends not only on the state $c$, but also on whether and when the agent invested. In this case, it is clearly without loss for the principal to commit to set the lowest possible policy of $y=0$ unless the agent invests in period 0. Then, the agent has no reason to wait until period 1 -- he should either invest in period 0 or never invest. Thus, the principal's first-best problem boils down to choosing a subsidy rule $y(c)$ (which specifies the subsidy to be paid if the agent invests in period 0 and the realized state is $c$) in order to minimize the expected policy cost, subject to the (ex-ante) \textit{individual rationality constraint} (IR) that the agent prefer investing in period 0 to never investing:
\begin{align*}
    \min_{y} &\, \E[cy(c)] \\
    \text{subject to } &\,  2b - I + \E[y(c)] \geq 0. \tag{IR}
\end{align*}
The first-best solution is $(y^{FB}(1),y^{FB}(2)) = (8,0)$. Intuitively, the cheapest way to satisfy IR is to provide zero policy incentives in the high-cost state, and set the policy in the low-cost state just high enough that IR holds with equality.


\paragraph*{Second Best (Non-contractible Investment).} Let us now suppose the principal can only condition her policy on the state, and not on the agent's investment decision. Then, the first-best outcome cannot be implemented. Indeed, suppose the principal commits to set the policy to $y^{FB}(c)$ in state $c$, regardless of the agent's actions. As before, the agent obtains zero expected payoff from investing in period 0, but he now has a profitable deviation: wait until period 1, observe the state, and then invest if and only if $c=1$, to obtain an expected payoff of $p[y^{FB}(1)+b-I] = 3/2$. 
Intuitively, the first-best policy rule loads all incentives onto the low-cost state, so it gives the agent too much option value from waiting for information about the state.

To prevent the agent from waiting, a feasible policy rule must additionally satisfy the \textit{no waiting constraint} (NW):
\begin{align*}
     2b - I + \E[y(c)] \geq p(b - I + y(1)), \tag{NW}
\end{align*}
where the LHS is the agent's expected payoff from investing in period 0, and the RHS is his expected payoff from waiting and then investing in period 1 only in the low-cost state. Hence, the principal's second-best problem is\footnote{In principle, the agent may also consider two other deviations:  he could wait and invest in both states, or wait and invest only if $c=2$. It is easy to check that both of these potential deviations can be ignored (i.e., the corresponding non-deviation constraints do not bind). This observation is a special case of our \cref{thm:aux=original}.}
\begin{align*}
    \min_{y} &\, \E[cy(c)] \\
    \text{subject to } &\, (\text{IR}),\,(\text{NW}).
\end{align*}
One can verify that both constraints should bind, and that the solution is $(y^{SB}(1),y^{SB}(2)) = \left(I-b, I - b(2-p)/(1-p)\right) = (5, 3)$.

Note that in the high-cost state, the second-best policy ($y^{SB}(2) = 3$) is higher than the first-best policy ($y^{FB}(2) = 0$). Intuitively, waiting for information gives the agent the option of avoiding investment when the policy is low. To discourage waiting, the principal must reduce the value of this option, and she achieves this by raising the policy in the high-cost state.\footnote{In the current example, (NW) \textit{only} depends on $y(2)$, since the terms involving $y(1)$ cancel out. As we will see later, this is a feature that comes from the agent's payoff being additively separable in $y$.}
A higher $y(2)$ then relaxes the IR constraint, allowing the principal to lower $y(1)$.
As a result, the second-best policies are sandwiched between the first-best policies: $y^{FB}(2) < y^{SB}(2) < y^{SB}(1) < y^{FB}(1)$. Moreover, the expected policies are equal, $\E[y^{FB}(c)] = \E[y^{SB}(c)] = I - 2b = 4$, so the second-best policy rule is a mean-preserving contraction of the first-best policy rule. That is, moral hazard makes the principal \textit{provide certainty} about her future policy.

Providing certainty is distortionary. $y^{SB}$ is clearly more costly than $y^{FB}$ for the principal, since $y^{SB}$ requires her to set a relatively higher policy when it is more costly to do so. On the other hand, conditional on investing in period 0, the agent has zero value for a less-dispersed policy rule, since his payoff is linear in policy; indeed, the agent's expected payoff is 0 under both $y^{FB}$ and $y^{SB}$. Hence, the second-best policy is Pareto-dominated by the first-best policy.

Note that, to discourage waiting, the principal is providing certainty to an agent who is risk-neutral (with respect to the policy). This stands in contrast to standard moral hazard models \citep{holmstrom1979moral}, where the agent is insured because he is risk-averse. In our main model, we will allow the agent to be risk-averse, and show that the principal optimally reduces the dispersion of her policy beyond what efficient risk-sharing would dictate.


In the example, the provision of certainty does not benefit the agent -- his expected payoff under $y^{SB}$ is 0. The reason behind this, as we will argue using our main model, is that the agent's benefit from policy is independent of the timing of his investment -- the policy $y$ gets added to the agent's payoff as long as the agent invests, regardless of whether he invests in period 0 or 1. On the other hand, in settings where investing earlier allows the agent to benefit even more from a given policy, it will be shown that the agent often benefits from certainty. This will again be in contrast to standard models of moral hazard, where the agent's rent typically comes from having limited liability or private information.



In studying the example, we have only considered the cost-minimizing policy rule, assuming that the principal must induce the agent to invest in period 0. In \cref{OA:motivating-example}, we augment the example by explicitly introducing the principal's benefit from investment, and solve for the optimal investment timing. It turns out that the optimal investment timing is weakly later under second-best than under first-best; that is, moral hazard can delay investment.


\subsection{Overview of the Paper}\label{subsec:overview}

In our main model (\cref{sec:setup}), the principal and the agent have richer payoff structures and interact over an arbitrary number of periods. Information about a binary state is publicly revealed over time via a breakdown process. The agent solves a stopping time problem of choosing when (if ever) to invest. The principal sets a policy in the final period, and commits ex ante to a policy rule which is contingent on the full history of public signals. The agent wishes to wait for information before investing, while the principal wishes to encourage early investment.
This general setup can be applied to the provision of environmental subsidies (\cref{ex:renewing-subsidies}) or the procurement of goods that require research and development (\cref{ex:procurement}).

\cref{sec:optimal-policy-rule} solves for the principal's optimal policy rule. It would be natural to consider the classical two-step approach \citep{grossman1983ananalysis}: first solve for the cost-minimizing policy rule that induces the agent to adopt a given stopping time, and then maximize over stopping times. A challenge is that, even though we assume a breakdown information process, the agent can still potentially choose from a multi-dimensional set of stopping times. We therefore combine the two-step approach with a guess-and-verify approach: we first solve for the cost-minimizing policy rule that makes an agent prefer a given \textit{simple stopping time} over all ``later'' simple stopping times (\cref{thm:inner}), and then verify that maximizing across these cost-minimization problems delivers the optimal policy rule in the original problem (\cref{thm:inner=aux,thm:aux=original}).

The main model allows the agent to be risk-averse, so the principal may reduce the dispersion of her policy purely for risk-sharing purposes. Nevertheless, the insight from the motivating example generalizes: the optimal policy rule provides certainty to the agent, beyond what efficient risk-sharing would imply (\cref{thm:inner}).



We study whether the principal's policy incentives benefit the agent  (\cref{sec:rent}). The answer hinges on whether investment timing and policy are \textit{complements} for the agent. In settings where the agent's benefit from policy does not depend on how early he invested, as in \cref{ex:renewing-subsidies}, the agent always receives zero rent. In contrast, when there is complementarity -- the earlier the agent invests, the more he benefits from higher policy -- the agent obtains positive rent under relatively mild conditions, as in \cref{ex:procurement}.

Then, focusing on settings where the agent's benefit from policy is independent of investment timing, we show that the structure of the optimal policy rule depends crucially on the agent's \textit{waiting premium}, which is the policy-independent portion of the agent's gain from waiting one more period before stopping. 
When the waiting premium is increasing over time, the principal can implement the optimal policy rule by announcing a \textit{state-measurable} policy rule at a time of her choice (\cref{sec:ota}).
On the other hand, when the waiting premium is decreasing over time, the optimal policy rule has a \textit{nested} structure (\cref{sec:nested}). The nested structure in turn implies that moral hazard delays investment.
We show via an example that, without nestedness, moral hazard may hasten investment (\cref{OA:hastening}).


As we discuss in the literature section (\cref{sec:lit}), a growing empirical literature aims to quantify the benefits of policy certainty in various economic settings. Our theoretical framework complements this effort by studying how a policymaker should \textit{endogenously provide} certainty, trading off the benefit (incentivizing investment) against the cost (not adapting policy to state).

\section{Setup}\label{sec:setup}

A principal (she) and an agent (he) interact over a finite number of periods $t\in \{0,1,\ldots,T\}$. We will abuse notation and sometimes let $T$ denote the set $\{0,1,\ldots,T\}$. There is a payoff-relevant state of the world $\t \in \{0, 1\}$ with an interior common prior. We interpret $\t=1$ as the good state and $\t=0$ as the bad state. The principal implements a policy $y\in \R_+$ in period $T$. The agent decides when, if ever, to make a one-time investment -- that is, he solves a stopping time problem.



\paragraph*{Breakdown Signals.} 

In each period, the principal and the agent symmetrically update their information about $\t$, by observing a public Markov signal $x_t$. Each $x_t$ takes a value of 1 (good) or 0 (bad). We let $\t = x_T$, meaning that the players fully learn the state in period $T$. The transition probabilities are $\Probability(x_{t}=0|x_{t-1}=0) = 1$ and $\Probability(x_{t}=1|x_{t-1}=1) = p_t \in (0,1)$. We say that a \textit{breakdown} occurred in period $t$ to mean that the signal transitioned from $x_{t-1}=1$ to $x_t=0$. 
The initial signal is $x_0 = 1$.

\paragraph*{Histories.} A (terminal) \textit{history} is a sequence of signals of length $T$, representing the signals that arrived in each period. We define
\[\bfa:=(\overbrace{1,\dots,1}^{T \text{ entries}}),\qquad \bfb_{t}:=(\overbrace{1,\dots,1}^{t-1 \text{ entries}},\overbrace{0,\dots,0}^{T-t+1\text{ entries}})\]
to represent the two possible varieties of histories. History $\bfa$ indicates that there has been no breakdown. History $\bfb_t$ indicates that a breakdown occurred in period $t$. Let 
$\ms{H}:=\{\bfa,\bfb_1,\ldots,\bfb_T\}$ denote the set of all possible histories, and let $\mathbf{h}$ denote a generic element of $\ms H$. Figure \ref{fig:signal-hist} summarizes the signal structure and histories.

\begin{figure}
    \centering
    \includegraphics[width=0.55\linewidth]{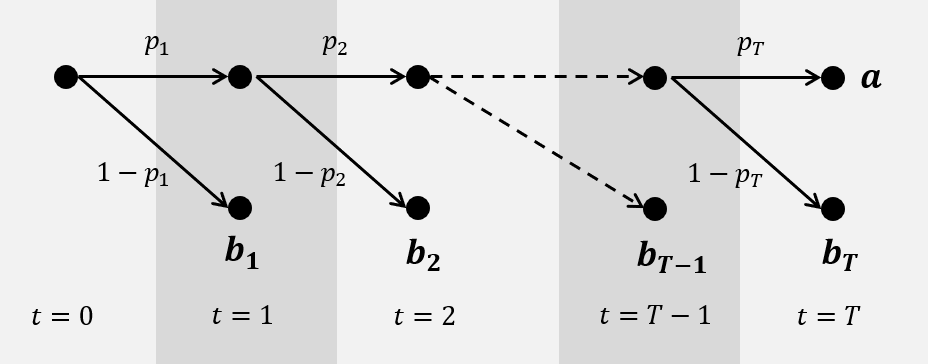}
    \caption{Signals and histories.}
    \label{fig:signal-hist}
\end{figure}

\paragraph*{Policy Rule.} 
The principal is able to commit ex ante to a \textit{policy rule}, $y:\ms{H}\to \R_+.$ This entails two economic assumptions regarding her commitment power:
\begin{enumerate}
    \item The policy rule \textit{can} depend on both the state and the signals. Such contingent policies are common. For example, the Clean Electricity Production Credit in the US is set to phase out when the US greenhouse gas emissions from electricity, which may be viewed as an exogenous state from an individual firm's point of view, drop below a predetermined level.\footnote{Similarly, in Operation Warp Speed, the US government procured 100 million COVID vaccines from Pfizer, with an option to order 500 million more \citep{pfizer20}. Since that option would presumably be exercised if and only if demand were sufficiently high, it can be viewed as a state-contingent policy.}
    \item The policy rule \textit{cannot} depend on the agent's choice of whether and when to invest. It is often difficult for a policymaker to observe investment decisions by individual firms or households, especially if the investment is intangible (such as building human capital). Moreover, one interpretation of our model is that there is a unit mass of homogeneous agents, and the principal chooses a single policy that applies to all agents in the economy (such as a central bank choosing the interest rate). Then, even if individual agents' investments are observed, it may not be feasible to condition policy on them.\footnote{Furthermore, it would not help to condition policy on the aggregate investment level, since no individual agent can affect this aggregate level.} 
\end{enumerate} 
In addition, we are restricting attention to deterministic policy rules which do not randomize over policies conditional on history $\bfh$.\footnote{The restriction is without loss in the special case where the principal's payoff is non-increasing in policy. This is because, for each $\bfh$, a randomized policy can be replaced by its certainty equivalent for the agent. An analogous argument was first made by \cite{grossman1983ananalysis}.}

\begin{rem}\label{rem:monitoring}
    We have assumed that the principal can condition her policy only on the public history -- that is, not on the agent's investment decision. Alternatively, we could allow the principal to condition policy on \textit{whether, but not when} the agent invested. Formally, such a policy rule would be described by $\tty(\bfh,i)$, where $i\in\{0,1\}$ is a binary variable indicating whether the agent invested by time $T.$

In \cref{OA:monitoring} we show that the stronger monitoring structure would not strictly benefit the principal, and so any optimal solution to our problem would remain optimal under the stronger monitoring structure.\footnote{I.e., if $y$ is an optimal policy rule in our problem, then the policy rule $\tty(\bfh,i)=y(\bfh)$ for all $\bfh,i$ is optimal under the stronger monitoring structure.} 

\cref{OA:monitoring} also shows that, if we strengthen the monitoring structure in this way, then we can weaken \cref{ass:A-payoff-main}(c) to:
\[
    \text{Assumption 1(c$'$):  For all $\theta$ and $y$,  $u(\theta,y,\infty) \geq u(\theta,0,\infty)=0$,}
\]
while maintaining the optimality of the policy rule that we derive. Assumption 1(c$'$) allows the agent to benefit from the principal's policy, even if he has not invested. In many settings, Assumption 1(c$'$) is more realistic than \cref{ass:A-payoff-main}(c). For example, it seems possible for the principal to give a cash transfer to an agent who has \textit{not} invested in a project; the agent, presumably, would benefit from this. Our results in \cref{OA:monitoring} show that our model also captures such situations, provided that the principal can verify whether investment has occurred by time $T$.
\end{rem}

\paragraph*{Timing.}

Within each period $t \in T$, the timing of events is as follows:
\begin{enumerate}[1.]
    \item Signal $x_t$ is publicly realized.
    \item If the agent has not already invested, he decides whether to invest or not.
    \item If $t=T$, the principal implements policy $y$.
\end{enumerate}

\paragraph*{Payoffs.} Payoffs are realized in period $T$. There are three payoff-relevant variables: the state $\t$, the principal's policy $y$, and the agent's investment timing $t\in T \cup \{\infty\}$, where $t=\infty$ means that the agent never invests. The agent's payoff is $u(\theta,y,t)$, and the principal's payoff is $v(\theta,y,t).$ Assumptions about $u$ and $v$ will be introduced shortly.

\paragraph*{Agent's Problem.} 

The agent solves a stopping time problem of choosing when to invest. We will use ``stop" and ``invest" interchangeably to refer to the agent's action.
Let $\ms T$ denote the set of all stopping times $\tau:\ms H \to \Tui$ adapted to the signal process $\{x_t\}_{t\in T}$. Taking as given the policy rule $y$, the agent chooses $\tau \in \ms T$ to maximize his expected payoff:\footnote{It is without loss to ignore mixed stopping times. If the agent mixes between two stopping times, then he must be indifferent between them, so either stopping time can be induced, and the principal can induce whichever she prefers.}
\begin{align*}
    \max_{\tau \in \ms T} \E[u(\t,y(\bfh),\tau)].
\end{align*}


\paragraph*{Principal's Problem.}
The principal chooses a policy rule $y$ and a stopping time $\tau$ to maximize her expected payoff, subject to the agent's incentive constraint:
\begin{align}
    [\ms P_0] \hspace{10em} \max_{y,\,\tau \in \ms T} &\, \mathbb{E}[v(\t,y(\bfh),\tau)]  \hspace{20em} \notag\\
    \text{subject to } &\, \E[u(\t,y(\bfh),\tau)] \geq \E[u(\t,y(\bfh),\tau')] \qquad \forall \tau' \in \ms T. \notag
\end{align}

\paragraph*{Payoff Assumptions.} 

We start with a definition.

\begin{defn}
    A payoff function $f(\t,y,t)$ satisfies \textit{aggregate positive monotonicity (APM)} if, whenever
    \begin{equation*}
    \label{eq:apm-premise}
        \sum_i \alpha_i f(\theta_i,y_i,t)\geq 0
    \end{equation*}
    for $t\in T$ and a finite sequence $(\alpha_i,\theta_i,y_i)_i$ with $\alpha_i>0$, we have:
    \begin{align*}
        \sum_i \alpha_i f(\theta_i,y_i,t') \geq \sum_i \alpha_i f(\theta_i,y_i,t) \qquad \forall t' < t. \label{eq:wAPM}
    \end{align*}
\end{defn}
APM is a weak discounting notion. It says that, provided the agent is receiving a lottery with a non-negative expected payoff by stopping tomorrow, he would prefer to receive that same lottery and stop today.\footnote{The agent's payoff from not investing will be normalized to 0 (\cref{ass:A-payoff-main}(c)), so a non-negative expected payoff means that the agent weakly prefers the lottery to never investing.} APM is implied by $f$ being non-increasing in $t$, but is weaker than that condition, because it only applies if the lottery has a non-negative expected utility.

We make the following assumptions on the agent's payoffs.

\begin{assumption}\label{ass:A-payoff-main}
    The agent's payoff function, $u(\t,y,t)$, can be written as:
        $$u(\t,y,t) = \phi(t)w(\t,y) + g(\t,t),$$
        where $\phi,w,$ and $g$ are real-valued functions\footnote{Specifically, $\phi:\Tui\to\R,$ $w:\{0,1\}\times \R_+\to\R$ and $g:\{0,1\}\times (\Tui)\to \R.$} satisfying the following conditions:
    \begin{enumerate}
        \item $w(\t,y)$ is twice continuously differentiable in $y$,  with $w_y>0$ and $w_{yy}\leq 0$.\footnote{We use $w_y$ and $w_{yy}$ to denote the first and second partial derivatives of $w$ with respect to $y$.} 
        \item For all $t\in T$, $\phi(t)$ is strictly positive and non-increasing.
        \item For all $\t$ and $y$,  $u(\theta,y,\infty)=0$. 
        \item $u$ satisfies APM.
        \item For all $y\geq 0$ and $t \in T$, $u(1,y,t) \geq u(0,y,t)$ and $u_y(1,y,t)\geq u_y(0,y,t)$.
        \item For all $t\in T$, $u(0,0,t)<0$.
    \end{enumerate}
\end{assumption}
The agent's payoff $u$ consists of a policy-dependent term, $\phi(t) w(\t,y)$, and a policy-independent term, $g(\t,t)$.
Note that the investment timing $t$ affects the agent's payoff through two channels -- it scales the impact of policy $y$ via the multiplicative component $\phi$, and it also enters additively via $g$.
That $\phi$ is non-increasing (condition (b)) means that investment timing and policy are weak complements for the agent -- if he invests earlier, he benefits weakly more from higher policy. The complementarity is strict if $\phi(t)$ is strictly decreasing. If $\phi$ is constant, then there is no complementarity, and the agent's payoff is additively separable in policy $y$ and investment timing $t$ -- as long as the agent invests in some period, he receives the same benefit from the policy.
Condition (c) says that if the agent never invests, then he cares about neither the policy nor the state, obtaining a constant payoff normalized to 0.
A sufficient condition for $u$ to satisfy APM (condition (d)) is that $g/\phi$ is non-increasing in $t$ (see \cref{OA:bgd}).
Condition (e) says that the agent's utility, and his marginal utility of policy, are both higher if the state is good. Condition (f) makes our problem non-trivial; it says that the agent prefers not to invest if the state is bad and the policy is 0, the lowest possible value.

Intuitively, as long as the agent is obtaining a positive payoff by stopping, he benefits from stopping earlier. However, he must trade off this benefit against the value of waiting for more information, since depending on the state (which also determines the policy), his payoff from stopping may be negative.

We make the following assumptions on the principal's payoffs.
\begin{assumption}\label{ass:P-payoff-main}
    The principal's payoff function, $v(\t,y,t)$, satisfies the following:
    \begin{enumerate}
        \item $v(1,0,\infty) = v(0,0,\infty) = 0.$
        \item $v$ is differentiable in $y$. For all $t\in T$, $v_{yy}<0$.
        \item For all $\t$ and $y$, $v_y(\t,y,\infty) \leq 0$. For all $\t$ and $t\in T$, $\lim_{y\to\infty} v(\t,y,t) = -\infty$.
        \item For all $\t$ and $t \in \Tui$, $v(\t,0,t)$ is non-increasing in $t$.
        \item $v$ satisfies APM.
        \item For all $y$ and $t\in T$, $v(1,0,t) \geq v(0,0,t)$ and $v_y(1,y,t)\geq v_y(0,y,t)$.
    \end{enumerate}
\end{assumption}
Condition (a) is a normalization.
Condition (b) says that the principal's payoff is strictly concave in policy.\footnote{In the motivating example (\cref{sec:motivating-example}), for expositional simplicity, we assumed $v$ to be affine in $y$.} Note that we do not require $v$ to be monotone in policy. 
Condition (c) says that if the agent never invests, the principal does not benefit from setting a higher policy; and the principal's payoff decreases without bound as the policy grows large. Condition (d) says that, at least when the policy is 0, the principal prefers that the agent invests and that he does so earlier.
Condition (f) says that, if the policy is 0 and the agent invests, the principal's payoff is higher in the good state; moreover, if the agent invests, the principal's marginal utility of policy is higher in the good state.






Conditions (c), (d) and (e) represent a weak version of the assumption that the principal wishes to induce the agent to invest early, but finds it costly to provide policy incentives. Weakening this assumption broadens the scope of our model, and in particular, allows us to capture Example 2 below.\footnote{In many settings, the policymaker herself may prefer to provide a positive level of policy, or prefer that the agent does not invest too early with little information. We would not be able to capture these settings if we simply assumed that $v$ is decreasing in $y$ and $t$.}



%


\paragraph*{Examples.}
We end this section by presenting two examples of our model. Throughout the paper, we will return to these examples to apply our results.

\begin{example}[Renewing Subsidies]\label{ex:renewing-subsidies}

The agent is a developer considering building a wind farm, which produces one unit of renewable electricity every period after it is built. Past legislation has determined and funded a subsidy of $s$ per unit, but this subsidy expires after the end of period $T-1$.\footnote{Empirically, expiration of environmental subsidies is a major source of policy uncertainty. See, for example, \cite{chen2024dynamic}. } The principal is the government and commits to a new subsidy $y \geq 0$ which will be paid from period $T$ onward. (This example has an infinite horizon, but we assume that investment must be made before time $T$.) The state $\t \in \{0,1\}$ represents exogenous technological progress; if $\t=1$, producing wind energy becomes more profitable for the agent. 

The agent's payoff from not investing is 0, and his payoff from investing in period $t\leq T$ is
\begin{align*}
    u(\t,y,t) =  \sum_{i=t}^{T-1}\d^i s - \d^t I + \sum_{i=T}^{\infty}\d^i (y + \t),
\end{align*}
where $\d \in (0,1)$ is the discount rate, and $I>0$ is the up-front cost of investment. We assume $s\geq I(1-\delta)$,
which means that the current subsidy more than covers the agent's rental cost on capital $I$. One can verify that $u$ satisfies \cref{ass:A-payoff-main} with $\phi(t)=1$, $w(\t,y) = (y+\t)\d^{T}/(1-\d)$, and $g(\t,t) = s(\d^t - \d^T)/(1-\d) - \d^t I$. Notice that $s\geq I(1-\delta)$ implies APM.

If the agent does not invest, the principal's payoff is 0. If the agent invests in period $t \leq T$, the principal's payoff is
\begin{align*}
    v(\t,y,t) = \sum_{i=t}^{\infty} \d^i V - \frac{\d^T}{1-\d} c(y),
\end{align*}
where $V>0$ is the social surplus generated per unit produced, and $c(y)$ is the cost of subsidies and satisfies $c(0) = 0$, $c'(y) \geq 0$, and $c''(y)>0$. One can verify that $v$ satisfies \cref{ass:P-payoff-main}; notice that APM is immediate, since $v$ is non-increasing in $t.$\footnote{We assume for simplicity that the current subsidy $s$ has already been funded and does not enter $v$. This assumption can easily be relaxed.}

This example is similar to the motivating example in \cref{sec:motivating-example}, but here, the agent cares about the state, while the principal has state-independent preferences and is risk-averse.

\end{example}

\begin{example}[Procurement]\label{ex:procurement}

The principal is the government and seeks to procure a good (say, a vaccine) from the agent, a producer, at time $T.$ Developing vaccines requires a costly and uncertain R\&D process, which is more likely to bear fruit the earlier it is started.

Specifically, to start R\&D, the agent must sink a cost of $I$. If he does so in period $t$, then success arrives in each period $t,t+1,\dots,T$ with probability $q \in (0,1)$. Thus, if the agent invests in period $t$, the probability that development succeeds by period $T$ is $1-(1-q)^{T-t+1}$. If the agent succeeds by period $T$, he is able to produce an arbitrary amount of vaccines at constant marginal cost (normalized to 0). If the agent does not succeed by period $T$, he cannot produce any vaccines.

Both players share a discount rate $\delta \in (0,1].$ We assume that the ratio $\frac{\delta^t}{1-(1-q)^{T-t+1}}$ is non-decreasing in $t$, for $t\in T.$
This means that as the agent waits, the probability of success decreases faster than fiscal discounting.
It is satisfied if $\delta$ is sufficiently close to 1 and $q$ is sufficiently close to 0.

The payoff-relevant state, $\theta\in\{0,1\}$, represents the principal's value for vaccines (which is ex ante uncertain and evolves with, for example, public news about the disease's severity). The policy rule $y$ represents an \textit{advance market commitment} \citep{kremer2004strong,kremer2020advance} by the principal to buy quantity $y$ of vaccines at a fixed price (normalized to $1/\delta^T$) if the vaccine has been developed. 
If the principal buys quantity $y$ in state $\theta$, the agent obtains revenue $y$, and the principal's payoff is:
\[(k\t+1) h(y) - y,\]
where $k>0$ is a scale parameter, and we assume $h'(y)>0$, $h''(y)<0$, $h(0)=0$, $h'(0)>1$, and $\lim_{y\to\infty} [(k+1)h(y) - y] = -\infty$. An example is $h(y) = \sqrt{y}$. Both the principal and the agent maximize expected utilities.

The agent's expected payoff (where the expectation is taken over the outcome of the R\&D process) from investing at time $t \in T$ and receiving policy $y$ in state $\theta$ is:
\begin{align*} 
    u(\theta,y,t)=\left(1-(1-q)^{T-t+1}\right) y - \d^t I.
\end{align*}
Taking $\phi(t)=(1-(1-q)^{T-t+1})$, $w(\theta,y)=y$ and $g(\t,t)= - \delta^t I$, we have that $u$ is of the functional form given in \cref{ass:A-payoff-main}. One can verify that all conditions in \cref{ass:A-payoff-main} are satisfied. In particular, APM is implied by the monotonicity of $\frac{\delta^t}{1-(1-q)^{T-t+1}}$.

If the agent invests at time $t$, the state is $\theta$, and the policy is $y$, the principal's expected net payoff is:

\begin{equation*}
    v(\theta,y,t)=\left(1-(1-q)^{T-t+1}\right)\left[(k\t+1) h(y) - y\right].
\end{equation*}
\noindent One can verify that all conditions in \cref{ass:P-payoff-main} are satisfied. In particular, to show APM, we may observe that $v$ takes the form given for the agent's payoff in \cref{ass:A-payoff-main}, with $\phi(t)=(1-(1-q)^{T-t+1})$, $w(\theta,y)=(k\t+1)h(y) - y$ and $g(\theta,t)=0,$ and then 
use the same argument as in \cref{OA:bgd}. 
\end{example}

\section{Optimal Policy Rule}\label{sec:optimal-policy-rule}


To solve the principal's problem $\original$, it would be natural to consider the classical two-step approach \citep{grossman1983ananalysis}, first solving for the optimal policy rule that induces the agent to play a given stopping time, and then maximizing the principal's expected payoff across stopping times. 
A potential difficulty is that our agent chooses a stopping time -- a mapping from histories to stopping decisions -- rather than a one-dimensional effort level, which complicates the agent's incentive constraints.

To address this issue, we combine the two-step approach with a guess-and-verify approach. We start by introducing the \textit{simple stopping times}, a naturally ordered subset of all stopping times (\cref{subsec:simple-stopping-time}).
After considering the first-best benchmark (\cref{subsec:first-best}), we solve for an optimal  (second-best) policy rule that makes an agent prefer a given simple stopping time over all other simple stopping times (\cref{subsec:solving-inner}). We then verify that maximizing the principal's payoff across simple stopping times leads us to the solution to the original problem $\original$ (\cref{subsec:solving-original}).


\subsection{Simple Stopping Times}\label{subsec:simple-stopping-time}

For $t \in T$, the \textit{simple stopping time} $\tau_t \in \ms T$ is the random variable defined by 
$\tau_t = t$ if $x_t = 1$, and $\tau_t = \infty$ if $x_t = 0$. We let $\tau_\infty \:= \infty$.
That is, an agent who adopts $\tau_t$ always waits until period $t$; then, if the period-$t$ signal is good, he stops immediately, and otherwise he never stops. A stopping time that is not simple is one where  the agent sometimes stops after a breakdown. Though such behavior is rather unintuitive, it cannot be ruled out a priori, since it can be a best response against some policy rules, and the policy rule is chosen endogenously by the principal. 

Given a policy rule $y$, the agent's expected payoff from a simple stopping time $\tau_t$ is
\begin{align*}
    U_t(y) = 
    \begin{dcases}
         \E [u(\t,y(\bfh),t) \mid x_t = 1] \prod_{i=1}^t p_i \caseif t < \infty\\
        0 \caseif t=\infty,
    \end{dcases}
\end{align*}
where the expectation is taken with respect to the state $\t$ and the terminal history $\mathbf{h} \in \mathcal{H}$. 
When unambiguous, we will omit the argument $y$ and simply write $U_t$.


The principal's \textit{auxiliary problem} is
\begin{align}
    [\ms P] \hspace{12.5em} \max_{y,\,\tau_t} &\, \mathbb{E} [v(\t,y(\bfh),\tau_t)]  \hspace{20em} \notag\\
    \text{subject to } &\, U_t \geq U_{t'} \qquad \forall t' \in \{0,1,\ldots,T,\infty\}. \notag
\end{align}
The auxiliary problem $[\ms P]$ differs from $[\ms P_0]$ in two aspects: the principal cannot induce the agent to adopt a non-simple stopping time, and the agent cannot deviate to a non-simple stopping time. We will verify in \cref{thm:aux=original} that the former difference is without loss of optimality, while the latter is without gain of optimality.


For $t \in \Tui$, we also define the \textit{inner (auxiliary) problem}:
\begin{align}
     [\ms P(t)] \hspace{11.5em}\max_{y} &\, \mathbb{E}[v(\t,y(\bfh),\tau_t)] \hspace{20em}\notag\\
    \text{subject to } &\, U_t \geq U_{t'} \qquad \forall t' \in \{t,t+1,\ldots,T,\infty\}. \notag
\end{align}
The problem $[\ms P(t)]$ differs from $[\ms P]$ in two aspects. First, $[\ms P(t)]$ must induce the simple stopping time $\tau_t$. Second, $[\ms P(t)]$ does not allow the agent to deviate to any ``earlier'' simple stopping time $\tau_{t'}$ with $t' < t$. These modifications ensure that $\inner$ is convex.\footnote{\cite{grossman1983ananalysis} transforms the principal's cost-minimization problem into a convex program by treating the agent's payoff from his wage as the principal's control variable. Their approach relies on a separability assumption on the players' payoffs, which would be too restrictive in our setting.}

\subsection{First-Best Benchmark}\label{subsec:first-best}

Before proceeding with the analysis of the optimal policy rule, it will be useful to consider a \textit{first-best} scenario where the agent's action is contractible, so that the principal can commit to a policy rule $y(\bfh,t)$ contingent on both the history $\bfh$ and the agent's investment timing $t$. In this benchmark case, whenever the agent deviates, the principal can maximally punish the agent by setting $y(\bfh,t)=0$.
Hence, as we show in \cref{OA:1st-best}, the principal's first-best problem boils down to
\begin{align}
    [\ms P^{FB}] \hspace{13em}\max_{y,\,\tau_t} &\, \mathbb{E} [v(\t,y(\bfh),\tau_t)] \hspace{15em}  \notag\\
    \text{subject to } &\, U_t(y) \geq \max_s U_s(\mathbf{0}), \notag
\end{align}
where $y(\bfh)$ is the policy implemented if the history is $\bfh$ and the agent stops at $\tau_t(\bfh)$ (that is, he does not deviate), and $\mathbf 0$ is the policy rule that sets $\mathbf{0}(\bfh) = 0$ for all histories $\bfh$. 
The following proposition describes the first-best policy rule.

\begin{prop}\label{prop:benchmark-solution}
    Suppose $(y,\tau_t)$ with $t < T$ solves $[\ms P^{FB}]$, and suppose $y(\bfh) > 0$ for $\bfh\in\{\bfb_{t+1},\ldots,\bfb_T,\bfa\}$. Then it must be that
    \begin{align} \label{eq:prop-2-1}
        y(\bfb_{t+1})=y(\bfb_{t+2})=\cdots = y(\bfb_{T}) \qquad \text{and} \qquad \frac{v_y\left(0,y(\bfb_T),t\right)}{u_y\left(0,y(\bfb_T),t\right)}=\frac{v_y\left(1,y(\bfa),t\right)}{u_y\left(1,y(\bfa),t\right)}.
    \end{align}
\end{prop}
\begin{proof}
    See \Cref{proof:benchmark-solution}.
\end{proof}
Condition \eqref{eq:prop-2-1} says that conditional on the payoff-relevant state being bad, the first-best policy does not depend on when the breakdown occurred (as long as it occurred after the agent's equilibrium investment time).  Moreover, the Borch rule holds -- the ratio of the players' marginal utilities of policy is equalized across states. We will see in \cref{subsec:solving-inner} that \eqref{eq:prop-2-1} is in general violated under second-best.

\subsection{Solving the Inner Problem: Provision of Certainty}\label{subsec:solving-inner}

Let us now return to studying the optimal (second-best) policy rule.
We first show that a solution to the inner problem $\inner$ must satisfy two properties: the principal provides certainty to the agent, and certainty is ``front-loaded''.
This analysis serves two purposes. First, the properties are economically meaningful because they will also apply to any solution to the original problem $\original$, as we will verify in \cref{subsec:solving-original}. Second, when verifying the connection between $\inner$ and $\original$, our arguments will make use of the properties of $\inner$ derived here.

In stating the result, we will say that a constraint is \textit{slack (binding)} to mean that the corresponding Lagrangian multiplier must be zero (nonzero).


\begin{thm}[Provision of Certainty] \label{thm:inner}
    Suppose $y$ solves $\inner$ for some $t < T$.
    The following statements are true:
    \begin{enumerate}\itemsep1em
        \item It must be that 
        \begin{align}
            \dfrac{v_y(0,y(\bfb_T),t)}{u_y(0,y(\bfb_T),t)}\leq \dfrac{v_y(1,y(\bfa),t)}{u_y(1,y(\bfa),t)}. \label{eq:thm1-a}
        \end{align}
        Furthermore, suppose $y(\bfb_T)>0$ and $y(\bfa)>0$. Then, the constraint $U_t \geq U_T$ is slack (binding) if \eqref{eq:thm1-a} holds as an equality (a strict inequality).
        \item It must be that
        \begin{align}
            y(\bfb_{t+1}) \geq y(\bfb_{t +2}) \geq \cdots \geq y(\bfb_T). \label{eq:thm1-b}
        \end{align}
          If $y(\bfb_{t+s}) = y(\bfb_{t+s+1}) >0$, then the constraint $U_t \geq U_{t+s}$ is slack. If $y(\bfb_{t+s}) > y(\bfb_{t+s+1})$, then $U_t \geq U_{t+s}$ is binding.
          \item It must be that $u(1,y(\bfa),t) \geq u(0,y(\bfb_{t+1}),t)$. 
    \end{enumerate}
\end{thm}
\begin{proof}
    See \cref{proof:thm:inner}.
\end{proof}

Let us explain why the three statements of \cref{thm:inner} together imply that the principal provides certainty. 
First, part (a) says that when the constraint $U_t \geq U_T$ is binding, Borch rule is violated -- the policy $y(\bfa)$ in the good state is inefficiently low, and the policy $y(\bfb_T)$ in the bad state is inefficiently high.
Intuitively, a higher $y(\bfb_T)$ increases the agent's expected payoff if he stops at $t$, but not if he waits until $T$. Hence, raising $y(\bfb_T)$ helps satisfy $U_t \geq U_T$.

Part (b) says that policy incentives are \textit{front-loaded}: conditional on the state being bad, the principal sets an even higher policy if the breakdown occurred earlier (see \Cref{fig:subs-theorem1}). The intuition is similar: 
raising policies $y(\bfb_{t+1})$ through $y(\bfb_{t+s})$ increases the agent's expected payoff if he stops at $t$, but not if he waits until $t+s$, so it helps satisfy the constraint $U_t \geq U_{t+s}$.


Parts (a) and (b) show that the optimal policy rule is distortionary. Except in the trivial case where the constraint $U_t\geq U_{t+s}$ is slack for every $s=1,\ldots,T-t$, at least one of the inequalities in \eqref{eq:thm1-a} or \eqref{eq:thm1-b} must be strict. Then, both the principal and the agent would be made better off if they could agree that the principal would raise $y(\bfa)$ and lower some $y(\bfb_{t+s})$, and the agent would continue to adopt the stopping time $\tau_t$. Such a Pareto improvement is not feasible because it would lead the agent to profitably deviate to a later stopping time.




We interpret the above-mentioned distortion -- $y(\bfa)$ being too low and $y(\bfb_{t+s})$ being too high -- as the principal \textit{providing certainty} to the agent. 
Indeed, (b) and (c) imply that the agent's ex post payoff $u(\t,y(\bfh),t)$ is higher under $\bfa$ than $\bfb_{t+s}$, and $u$ is increasing in $y$, so the distortion reduces uncertainty in the agent's ex post utility levels.\footnote{To be precise, consider the distribution of the agent's ex post payoff $u(\t,y(\bfh),\tau_t)$  under the solution $(y,\tau_t)$. It is possible to raise $y(\bfa)$ and lower $y(\bfb_T)$ in such a way that this distribution becomes more dispersed in the convex order (which implies the agent is indifferent ex ante), while the principal is ex ante strictly better off.} 
Likewise, our interpretation of (b) is that the principal front-loads certainty by promising higher utility levels for earlier breakdowns. In the special case where the agent's payoff is state-independent (as in the motivating example or \cref{ex:procurement}), (b) and (c) imply that $y(\bfa) \geq y(\bfb_{t+s})$, so the principal provides certainty about not only the agent's utility level, but also the policy itself.

\begin{figure}[t]
    \centering
    \begin{subfigure}[t]{0.48\linewidth}
    \centering
        \includegraphics[width=1\linewidth]{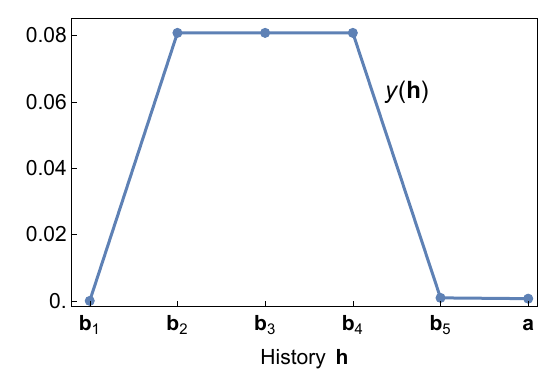}
        \caption{Solution to $[\ms P(1)]$}
    \end{subfigure}
    \hfill
    \begin{subfigure}[t]{0.48\linewidth}
    \centering
        \includegraphics[width=1\linewidth]{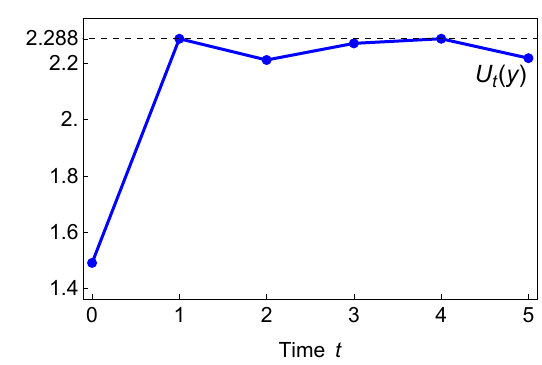}
        \caption{$U_t(y)$}
    \end{subfigure}
    \caption{Solving an inner problem (\Cref{ex:renewing-subsidies})}
    \medskip
    \small
    \justifying
    \noindent The figure is based on \cref{ex:renewing-subsidies} with $T=5$.
    Panel (a) depicts a solution $y$ to the inner problem $[\ms P(1)]$.
    Since the agent does not stop if the history is $\bfb_1$, the policy $y(\bfb_1)$ is irrelevant and may be set to 0. For $t\geq 2$, the principal provides certainty by raising the policy under the bad state ($y(\bfb_t) > y(\bfa)$), and certainty is front-loaded ($y(\bfb_t)$ is nonincreasing in $t$).
    
    Panel (b) plots the agent's expected payoff $U_t(y)$ from adopting $\tau_t$ under policy $y$, for each $t$.
    Comparing panels (a) and (b), we see that the policies at two adjacent histories are equal to each other if and only if the corresponding constraint is slack. That is, we have $y(\bfb_4)>y(\bfb_5)$ because $U_1 = U_4$, and we have $y(\bfb_2)=y(\bfb_3)=y(\bfb_4)$ because $U_1 > U_2$ and $U_1 > U_3$.
    (Parameters: $(p_t)_{t=1}^5=(0.45,0.8,0.7,0.8,0.98)$, $\delta=0.95, I=5.5, s=0.8, V=0.002$ and $c(y)=4y^3$.)
    \label{fig:subs-theorem1}
\end{figure}

\cref{thm:inner} implies that the second-best policy features distortions that do not arise under first-best. However, the second-best policy is not in general Pareto-dominated by the first-best policy. This is because the agent's expected payoff may be higher under second-best, as we will see in \cref{sec:rent}.




\subsection{Solving the Original Problem}\label{subsec:solving-original}

Suppose one solves the inner problem $\inner$ for each $t$, and then selects $t$ which maximizes the principal's payoff. We now argue that this two-step approach is essentially equivalent to solving the original problem $\original$. We will first show that taking the two-step approach is equivalent to solving the auxiliary problem $\auxiliary$, and then argue that $\auxiliary$ is equivalent to $\original$.

We start with a definition.
\begin{defn}\label{def:standard-sol}
        A solution $y$ to $\inner$ is \textit{standard} if
        \begin{align}\label{eq:standard-sol}
           y(\bfb_s)=0 \quad \forall s \leq t;
\qquad
y(\bfa)=0 \ \text{if } t=\infty.
        \end{align}
        A solution $(y,\tau_t)$ to $\original$ or $[\ms P]$ is \textit{standard} if it satisfies \eqref{eq:standard-sol}.
\end{defn}
A standard solution $(y,\tau_t)$ has two important properties. First, $\tau_t$ must be a simple stopping time (notice that stopping times which solve $[\ms P_0]$ are not guaranteed to be simple). Second, if a breakdown occurs before the agent is supposed to consider stopping, the principal must provide the zero policy. This is natural: the zero policy provides the maximum possible incentives for the agent not to stop, and, by \cref{ass:P-payoff-main}(c), is the principal's favorite policy conditional on the agent not stopping.


We have the following lemma.


\begin{lem}\label{lem:standard-solutions-exist}
    The following statements are true:
    \begin{enumerate}
        \item For $t\in \Tui$, if $\inner$ has a nonempty feasible set, then a standard solution to $[\ms P(t)]$ exists.
        \item A standard solution to $[\ms P]$ exists.
    \end{enumerate}
\end{lem}
\begin{proof}
    See \cref{proof:lem:standard-solutions-exist}.
\end{proof}

Let $W(t)$ denote the value of $[\ms P(t)]$, with $W(t) \:= -\infty$ if the feasible set of $\inner$ is empty. The following theorem establishes the relationship between $\inner$ and $\auxiliary$.

\begin{thm}\label{thm:inner=aux}
    The following statements are true:
    \begin{enumerate}
        \item If $(y,\tau_{t})$ solves $[\ms P]$, then $t\in \argmax_{s} W(s)$, and $y$ solves $[\ms P(t)]$.
        \item Suppose $t \in \argmax_s W(s)$, and $y$ is a standard solution to $\inner$. Then, there exists $t' \in \argmax_s W(s)$ such that $t' \leq t$, $y$ solves $[\ms P(t')]$, and the pair $(y,\tau_{t'})$ solves $[\ms P]$.
    \end{enumerate} 
\end{thm}
\begin{proof}
    See \cref{proof:thm:inner=aux}.
\end{proof}
Recall that the inner problem $\inner$ ignores the agent's deviations to earlier stopping times. \cref{thm:inner=aux} says that this is without loss. Specifically, part (a) of the theorem says that any solution to $[\ms P]$ can be obtained by the two-step procedure of first solving $[\ms P(t)]$ at each $t$ and then maximizing over $t$. This implies that the necessary properties of a solution to $\inner$, derived in \cref{thm:inner}, must also be necessary for a solution to $\auxiliary$.
Part (b) establishes a partial converse: if one solves $[\ms P(s)]$ for each $s$ and then maximizes over $s$, the policy rule obtained in this manner solves $[\ms P]$, with the caveat that the agent might adopt an earlier stopping time ($\tau_{t'}$ rather than $\tau_t$). Part (b) will prove useful for obtaining solutions to various examples.

Intuitively, we may ignore the agent's deviations to earlier stopping times for the following reason.
Under a standard solution to $\inner$, the principal sets the policy to be zero if a breakdown occurs before $t$. However, as long as the policy is zero, the principal prefers that the agent stops, even if the state is bad (\cref{ass:P-payoff-main}(d)). Moreover, conditional on the policy being nonzero, APM implies that the principal prefers earlier stopping. Hence, if the agent wishes to deviate to an earlier stopping time $\tau_{t'}$, the principal should actually let the agent deviate.

We now establish the relationship between $\auxiliary$ and $\original$.


\begin{thm}[Equivalence between ${[\ms P]}$ and ${[\ms P_0]}$]
\label{thm:aux=original}
The following statements are true:
\begin{enumerate}
    \item If $(y,\tau_t)$ is a standard solution to $[\ms P]$, then it solves $[\ms P_0].$ 
    \item If $(y,\tau_t)$ solves $[\ms P_0]$, then $(y,\tau_t)$ solves $[\ms P].$ 
\end{enumerate}
\end{thm}
    
\begin{proof}
    See \cref{app:thm-2-proof}.
\end{proof}

\cref{thm:aux=original}, together with \cref{thm:inner=aux}, verifies that it is sufficient to analyze the inner problem $\inner$ to solve the original problem $[\ms P_0]$, in the following sense. First, \cref{thm:inner=aux}(b) and \cref{thm:aux=original}(a) imply that, to find a standard solution to $\original$, one can use the two-step procedure: solve each $\inner$, and maximize the principal's payoff across $t$. Second, \cref{thm:inner=aux}(a) and \cref{thm:aux=original}(b) imply that any standard solution to $\original$ must satisfy all the necessary properties of a solution to $\inner$, including those stated in \cref{thm:inner}. In particular, under any standard solution to $\original$, the principal provides certainty to the agent.\footnote{In \cref{thm:aux=original}, it is not possible to discard the assumption that $(y,\tau_t)$ is a standard solution in (a), or that $\tau_t$ is a simple stopping time in (b) (see \cref{app:non-inclusion}).}


Let us sketch the proof of \cref{thm:aux=original}. \cref{thm:aux=original} comprises two main ideas, contained in two supporting lemmas. First, we would like to show that the agent cannot exploit any standard solution to $[\ms P],$ by using a non-simple stopping time. Thus, any standard solution to $[\ms P]$ remains feasible for $[\ms P_0].$ This is proved in \cref{lem:p_sol_is_feas}. Second, we would like to show that the principal would not \textit{want} to induce the agent to use a non-simple stopping time; this is proved in \cref{lem:principal_wants_simple}.

\begin{lem}
\label{lem:p_sol_is_feas}
    Let $(y,\tau_{t^*})$ be a standard solution to $[\ms P].$ Then, $\tau_{t^*}$ solves the agent's problem given policy rule $y$:
    \[\tau_{t^*} \in \argmax_{\tau\in\ms T} \E[u(\t,y(\bfh),\tau)].\]
\end{lem}

We prove \cref{lem:p_sol_is_feas} by showing that, under any standard solution to $[\ms P]$, the agent cannot receive strictly positive payoff from stopping following a breakdown. The idea of the argument is to exploit APM, in conjunction with the optimality properties of $y$ given by \cref{thm:inner}. Suppose that $(y,\tau_{t^*})$ is a standard solution to $[\ms P]$, and suppose that the agent could profitably stop after some breakdown; let the latest breakdown for which this is the case occur at $s$. Then, let $t$ be the earliest period where the agent's incentive constraint binds under $y.$ In the most interesting case, we have $t^*<s<t.$

By the complementary slackness result in \cref{thm:inner}, and by \cref{thm:inner=aux}(a) which implies that this result also applies to $\auxiliary$, we must have $y(\bfb_{t^*+1})=\ldots = y(\bfb_t)$. Thus, by playing $\tau_{t^*}$, the agent guarantees himself $y(\bfb_s)$ in the event of any breakdown between $t^*$ and $t$; since we assumed that stopping after a breakdown at $s$ gave him positive utility, by APM this event must also give him positive utility. Thus, $U_{t^*}$ must be strictly larger than $U_t$, contradicting the assertion that the incentive constraint at $t$ binds.

Next, \cref{lem:principal_wants_simple} shows that it is optimal for the principal to induce a simple stopping time in $[\ms P_0].$

\begin{lem}
\label{lem:principal_wants_simple}
There exists a policy $y^*$ and a simple stopping time $\tau_{t^*}$ such that $(y^*,\tau_{t^*})$ solves $[\ms P_0]$.
\end{lem}
\noindent To prove \cref{lem:principal_wants_simple}, suppose that an optimal solution $(y,\tau)$ has a non-simple $\tau$. Under $\tau$, let $t$ be the period in which the agent stops following no breakdown, and let $s\leq t$ be the principal's \textit{favorite} period among those in which the agent stops following a breakdown. That is, $s\in\argmax_{\hat{s}} v(0,y(\bfb_{\hat{s}}),\hat{s})$, where $\hat{s}$ ranges over all the periods where the agent stops following a breakdown. 

Then, the principal can (weakly) improve by deviating to $(\tilde{y},\tau_0)$, where $\tilde{y}(\bfb_i):=y(\bfb_s)$ for $i\leq t$, and $\tilde{y}=y$ otherwise. Against $\tilde{y},$ one can show that the agent is willing to stop at 0. (A key observation is that, by doing so, he is guaranteed to receive $u(0,y(\bfb_s),0)$ in the case of any breakdown occurring before $t$; by APM this is at least as good as $u(0,y(\bfb_s),s)$, and so must be positive, since he stopped after a breakdown at $s$ under $y$.) As a result, the principal improves: in the case of any breakdown occurring before $t,$ she receives $v(0,y(\bfb_s),0);$ by APM, and because $s$ was her favorite instance of the agent stopping after a breakdown, this must be at least as good as what she received under $y$. 

\cref{thm:aux=original} follows as a corollary to \cref{lem:p_sol_is_feas,lem:principal_wants_simple}. 
\begin{customex}{1}[Renewing Subsidies]
We use this example, with the number of periods equal to $T=2$, to illustrate how our theorems can be used to solve for the optimal policy rule. First, we solve the inner problem $\inner$ for each $t$ (\cref{fig:subs-theorem2-a}). For each $\inner$, let $y^{(t)}$ denote the solution and $W(t)$ its value. One can verify that $t=1$ uniquely maximizes $W(t)$  (\cref{fig:subs-theorem2-b}). Therefore, by \cref{thm:inner=aux}(b) and 
\cref{thm:aux=original}(a), $(y^{(1)},\tau_1)$ solves the original problem $\original$.


\begin{figure}[!htb]
    \centering
    \begin{subfigure}[t]{0.48\linewidth}
    \centering
        \includegraphics[width=1\linewidth]{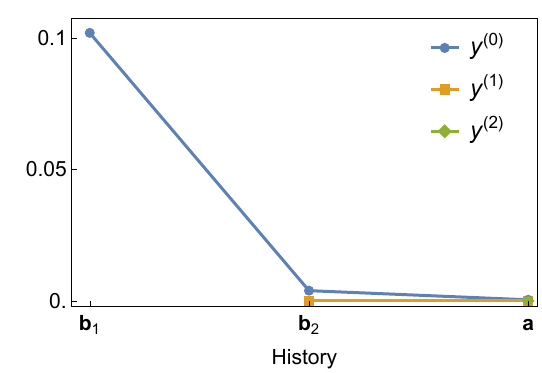}
        \caption{Solutions to inner problems}
        \label{fig:subs-theorem2-a}
    \end{subfigure}
    \hfill
    \begin{subfigure}[t]{0.48\linewidth}
    \centering
        \includegraphics[width=1\linewidth]{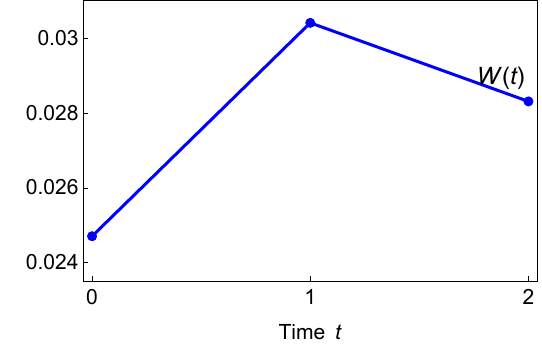}
        \caption{Values of inner problems}
        \label{fig:subs-theorem2-b}
    \end{subfigure}
    \caption{Solving \cref{ex:renewing-subsidies} ($T=2$)}
    \medskip
    \small
    \justifying
    \noindent Panel (a) shows the solution $y^{(t)}$ to the inner problem $\inner$, for $t = 0,1,2$. Panel (b) shows the value of each inner problem.
    (Parameters: $(p_1, p_2)=(0.8,0.98)$, $\delta=0.95, I=5.5, s=0.8, V=0.002$ and $c(y)=4y^3$.)
\end{figure}

\end{customex}

\section{Agent's Rent}\label{sec:rent}


In this section, we ask how the optimal policy affects the agent's welfare. We establish a lower bound on the agent's equilibrium payoff (\cref{prop:null-is-lower-bound}), and study when the agent is held to this lower bound (\cref{prop:no-rent}) and when he obtains rent above it (\cref{prop:rent}). We apply these results to \cref{ex:procurement,ex:renewing-subsidies}, and compare the results to those in standard moral hazard models.

For ease of exposition, in this section, we work under the following assumption:
\begin{assumption}\label{ass:rent}
    For all $\t\in \{0,1\}$, $y\geq 0$ and $t,t' \in T$, we have $\argmax_y v(\t,y,t) = \argmax_y v(\t,y,t') =:\uy (\t)$ and $v(0,\uy(0),t) > 0$.
\end{assumption}
\noindent \cref{ass:rent} says that the principal's ex post optimal policy does not depend on the agent's investment timing (as long as the investment occurs at some point), and this policy always gives the principal a strictly positive payoff. Note that the assumption is satisfied by \cref{ex:renewing-subsidies,ex:procurement}. The results in this section extend straightforwardly to the case without \cref{ass:rent}; see \cref{proof:rent-section}.

We begin with some definitions. The \textit{null policy rule}, denoted $\ynull$, is the policy rule defined by
\begin{align*}
    \ynull (\bfh) \:= \begin{cases}
       \uy(0) \caseif \bfh \in \{\bfb_{1},\bfb_{2},\ldots,\bfb_T\} \\
        \uy(1) \caseif \bfh = \bfa.
    \end{cases} 
\end{align*}
The null policy rule maximizes the principal's ex post payoff under each history, taking as given that the agent invests.\footnote{By \cref{ass:rent}, the principal's ex post optimal policy does not depend on the timing of investment.} We define the agent's \textit{null payoff} as 
$$\Unull \:= \max_{s} U_s (\ynull).$$
This is the agent's ex ante payoff when he best responds to the null policy rule.\footnote{Since $\uy(1) \geq \uy(0)$, a simple stopping time is optimal for the agent under the null policy rule.}

Since the agent's payoff is increasing in policy, his null payoff $\Unull$ is in general weakly greater than his maxmin payoff, $\max_s U_s(\mathbf{0})$. In the special case where the principal's payoff $v$ is decreasing in policy $y$, as in \cref{ex:renewing-subsidies} or the motivating example, the null policy rule is $\ynull \equiv 0$, and the agent's null payoff equals his maxmin payoff, $\Unull = \max_s U_s(\mathbf{0})$. 


Because the principal optimally incentivizes the agent by raising the policy beyond her ex post optimal level, the null payoff serves as a lower bound on the agent's equilibrium payoff, as stated by the following proposition.
\begin{prop}\label{prop:null-is-lower-bound}
    Suppose \cref{ass:rent} holds. If $(y,\tau_t)$ solves $\original$, then $U_t(y) \geq \Unull$.
\end{prop}
\begin{proof}
    See \cref{proof:rent-section}.
\end{proof}

Recall that we assume  $\phi$ is weakly decreasing, meaning investment timing and policy are weak complements for the agent.
The next result says that when $\phi$ is constant, so that there is no complementarity, the lower bound in \cref{prop:null-is-lower-bound} must be tight.

\begin{prop}[No Rent Without Complementarity]
\label{prop:no-rent}
    Suppose \cref{ass:rent} holds, and $\phi$ is constant. If $(y,\tau_t)$ solves $[\ms P_0]$, then $U_t(y) = \Unull$.
\end{prop}
\begin{proof}
See \Cref{proof:rent-section}.
\end{proof}


To gain intuition for \cref{prop:no-rent}, suppose that $(y,\tau_0)$ solves $\original$, and suppose that $\tau_2$ is the agent's best response under the null policy rule. First, we claim that $y(\bfh) = \ynull(\bfh)$ for $\bfh \in \{\bfb_3,\ldots,\bfb_T,\bfa\}$. Since the agent is already willing to stop in period 2 under the null policy, the principal does not need to provide policy incentives beyond the null level to prevent deviations to $\tau_t$ for $t>2$. It may be that incentives must be provided to prevent deviations to $\tau_1$ or $\tau_2$. However, because $\phi$ is constant, raising $y(\bfh)$ for $\bfh \in \{\bfb_3,\ldots,\bfb_T,\bfa\}$ would increase $U_0(y)$, $U_1(y)$, and $U_2(y)$ by the same amount, and so would not affect the agent's incentives to deviate to $\tau_1$ or $\tau_2$. Hence, we must have $y(\bfh) = \ynull(\bfh)$ for $\bfh \in \{\bfb_3,\ldots,\bfb_T,\bfa\}$. In turn, this implies $U_2(y) = U_2(\ynull) = \Unull$.

Now, suppose towards a contradiction that $U_0(y)$ is strictly higher than $U_2(y) = \Unull$. It must be that $U_0(y) = U_1(y)$, as otherwise, no incentive constraint binds, implying $U_0(y) = \Unull$. Then, the principal can marginally lower $y(\bfb_2)$ while maintaining $U_0(y) > U_2(y)$. This increases the principal's expected payoff because her marginal utility of policy must be negative in equilibrium. Moreover, since $\phi$ is constant, lowering $y(\bfb_2)$ does not affect the incentive constraint $U_0(y) \geq U_1(y)$. So, the perturbation is profitable for the principal, contradicting the optimality of $y$. Therefore, it must be that $U_0(y) = \Unull$.


The arguments in the above two paragraphs clearly rely on $\phi$ being constant. In particular, when $\phi$ is strictly decreasing, it may be optimal to set $y(\bfh) > \ynull(\bfh)$ for $\bfh \in \{\bfb_3,\ldots,\bfb_T,\bfa\}$ because increasing these policies increases $U_0(y)$ more than it increases $U_1(y)$, helping to satisfy the incentive constraint $U_0(y) \geq U_1(y)$. This provides intuition for the following proposition, which is a partial converse of \cref{prop:no-rent} and states that a strict complementarity between investment timing and policy often leads to rent.
\begin{prop}[Rent from Complementarity]\label{prop:rent}
    Suppose \cref{ass:rent} is satisfied. Let $(y,\tau_t)$ be a standard solution to  $\original$. Suppose the following statements hold:
    \begin{enumerate}[(i)]
        \item  $\phi$ is strictly decreasing.
        \item  If $\hat t \in \argmax_s U_s (\ynull)$, then $\hat t > t$.
        \item There exists $\hat t \in \argmax_s U_s(\ynull)$ such that $\hat t \leq T$, and for this $\hat t$, there exists $\bfh \in \{\bfb_{\hat t+1},\bfb_{\hat t+2},\ldots,\bfb_T,\bfa\}$ such that $y(\bfh)>0$.
    \end{enumerate}
    Then, $U_t(y) > \Unull$. 
\end{prop}
\begin{proof}
    See \cref{proof:rent-section}.
\end{proof}
Condition (ii) says that the optimal policy rule makes the agent stop earlier than he would under the null policy rule. This is necessary for the agent to obtain positive rent, as otherwise the null policy rule will be optimal. To understand the key role played by decreasing $\phi$ (condition (i)), suppose that, under the null policy rule $\ynull$, the agent adopts the stopping time $\tau_{\hat t}$ for some $\hat t > t$. 
Since $\phi$ is strictly decreasing, raising $y(\bfh)$ above $\ynull(\bfh)$ for $\bfh \in \{\bfb_{\hat t+1},\bfb_{\hat t+2},\ldots,\bfb_T,\bfa\}$ leads to a first-order gain in incentivizing the agent to adopt $\tau_t$ rather than $\tau_{\hat t}$.
On the other hand, if $y(\bfh)$ is interior (condition (iii)), and we have $y(\bfh) = \ynull(\bfh)$, then the cost to the principal of marginally increasing $y(\bfh)$ is of second order. So, the optimal $y$ must satisfy $y(\bfh) > \ynull(\bfh)$. 
Since the agent's payoff is increasing in policy, this implies $U_{\hat t}(y) > U_{\hat t}(\ynull) = \Unull$. However, $\tau_t$ is the agent's best response against $y$, so we conclude $U_t(y) \geq U_{\hat t}(y) > \Unull$.


Condition (iii) is not necessary for rent. For instance, it may be that $\hat t = \infty$, so we have $U_{\hat t}(y) = \Unull = 0$, and yet $U_t(y) > 0$. We provide such an example in \cref{OA:rent}.

\cref{prop:no-rent,prop:rent} can be used to compare the agent's first- and second-best payoffs. Since the first-best problem $\auxiliaryfb$ is a concave maximization program, the agent's first-best equilibrium payoff is equal to either $\max_s U_s(\mathbf{0})$ (if the incentive constraint binds) or $U_t(\ynull)$ for $t$ such that $\tau_t$ solves $\auxiliaryfb$ (if the incentive constraint is slack).\footnote{When the incentive constraint is slack, the first-best policy must coincide with $\ynull$ on histories $\bfh \in \{\bfb_{t+1},\ldots,\bfa\}$. Moreover, the policies under histories $\bfh \in \{\bfb_1,\ldots,\bfb_t\}$ do not affect the agent's payoff from $\tau_t$. So, the agent's payoff is equal to $U_t(\ynull)$.}
Hence, whenever \cref{prop:rent} applies, the agent's equilibrium payoff is strictly higher under second-best than under first-best.
This explains why the first-best solution does not in general Pareto-dominate the second-best solution. On the other hand, suppose that the first-best incentive constraint binds, that $\Unull = \max_s U_s(\mathbf{0})$ (for instance, because $v$ is decreasing in $y$), and that $\phi$ is constant. Then, \cref{prop:no-rent} implies that the agent's first- and second-best payoffs are equal.

Next, we apply \cref{prop:rent,prop:no-rent} to our two examples.

\begin{customex}{1}[Renewing Subsidies]
    In \cref{ex:renewing-subsidies}, $\phi$ is constant because the agent's benefit from subsidies does not depend on when the agent invests. Hence, \cref{prop:no-rent} applies. In particular, the principal's utility $v$ is strictly decreasing in subsidy $y$, so the null policy rule is simply $\ynull \equiv 0$, and the agent is held to his maxmin payoff. That is, if $(y^*,\tau_t)$ solves the principal's problem, we must have
    $U_t(y^*) = \Unull = \max_s U_s(\mathbf{0})$; see \Cref{fig:subs-nl+el}. \qed
\end{customex}

\begin{figure}[!hbt]
    \centering
    \includegraphics[width=0.8\linewidth]{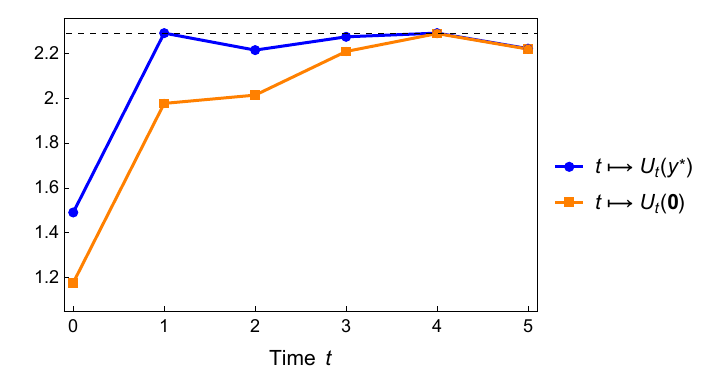}
    \caption{Equilibrium payoff in \cref{ex:renewing-subsidies} (constant $\phi(t)$, $T=5$)}
    \medskip
    \small
    \justifying
    \noindent Under the optimal policy $y^*$, the agent chooses $\tau_1$ and obtains $U_1(y^*)$. As depicted by the dashed horizontal line, this payoff is equal to the agent's maxmin payoff $U_4(\mathbf{0})$, which is what he gets if the principal sets $y\equiv 0$, and the agent best responds by choosing $\tau_4$.
    (Parameters: $(p_t)_{t=1}^5=(0.45,0.8,0.7,0.8,0.98)$, $\delta=0.95, I=5.5, s=0.8,  V=0.002$, $c(y)=4y^3$.)
    \label{fig:subs-nl+el}
\end{figure}

\begin{customex}{2}[Procurement]
In \cref{ex:procurement}, the function $\phi(t)=1-(1-q)^{T-t+1}$ is strictly decreasing -- the earlier the agent invests, the more likely he is to succeed in R\&D by period $T$. Under $\ynull$, the principal always purchases the quantity that maximizes her ex post payoff $(k\t + 1)h(y) - y$. Suppose that, even under the null policy rule, the agent eventually becomes willing to invest, conditional on no breakdown; let $\hat t$ denote the earliest time at which the agent becomes willing to invest. Moreover, suppose the principal's optimal solution hastens investment, inducing $\tau_t$ with $t < \hat t$. Then, conditions (i) to (iii) in \cref{prop:rent} are satisfied, so the agent obtains positive rent. In particular, the interiority condition $y(\bfh)>0$ holds automatically because the principal's ex post optimal quantity is always positive, due to the assumption that $h'(0)>1$. See \cref{fig:proc-nl+el} for an illustration. \qed

\end{customex}
    

\begin{figure}[!htb]
    \centering
    \includegraphics[width=0.8\linewidth]{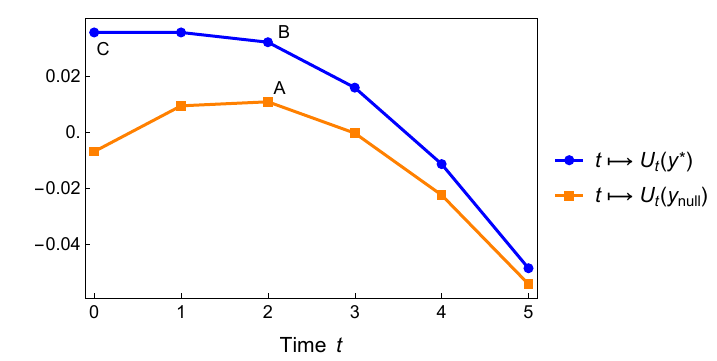}
    \caption{Equilibrium payoff in  \Cref{ex:procurement} (decreasing $\phi(t)$, $T=5$)}
    \medskip
    \small
    \justifying
    \noindent 
    The agent waits until period 2 under $\ynull$ (point A), but the optimal policy $y^*$ makes him invest earlier, in period 0 (point C). Hence, the agent obtains rent (C is above A). 
    The intuition is that, in order to prevent the agent's deviation away from C, it is optimal for the principal to increase $y(\bfa)$, which makes $U_2(y^*)$ (point B) strictly higher than $U_2(\ynull)$ (point A). However, by the agent's incentive constraints, C must be weakly above B.
    (Parameters: $p_t=0.9$ for all $t$, $\delta=0.89, I=0.36, q=0.09, k=1.2$ and $h(y)=\sqrt{y}$.)
    \label{fig:proc-nl+el}
\end{figure}

While we impose a lower bound of zero on the principal's policy, our result that the agent can obtain rent does not rely on the agent's limited liability. In fact, even if $v$ is decreasing in $y$, so that the principal's best response to  any given strategy of the agent is to set $y=0$, it can be that $y>0$ under every history, and the agent's rent is positive. We provide such an example in \cref{OA:rent}.\footnote{Limited liability also does not bind in the parametric version of \cref{ex:procurement} in \cref{fig:proc-nl+el}, since the optimal policy satisfies $y^*(\bfh)>0$ for all $\bfh$. However, since $v$ is increasing in $y$ for low values of $y$, the principal does not even prefer to lower the policy to $y=0$, so limited liability becomes slack in a rather mechanical manner. Hence, \cref{fig:proc-nl+el} may not be entirely satisfactory for illustrating our point that limited liability is not required for rent, and we provide a more compelling example in \cref{OA:rent}.} This contrasts with standard moral hazard models following \cite{holmstrom1979moral}. There, an agent without limited liability must get zero rent because the principal can lower the agent's wage across all output levels without affecting the agent's incentive constraints. An analogous argument does not apply to our model when $\phi$ is strictly decreasing, since for any history $\bfh$ under which the agent would have stopped, lowering $y(\bfh)$ necessarily tightens the agent's incentive constraint.

\section{State-Measurable Policy Rule}
\label{sec:ota}

In \cref{sec:nested,sec:ota}, we focus on the case of constant $\phi$ and argue that a quantity we call the agent's \textit{waiting premium} plays a key role in determining the structure of the optimal policy rule.
When the waiting premium is increasing in time, the principal can implement the optimal policy rule by announcing a state-measurable policy rule at a time of her choice (\cref{sec:ota}). On the other hand, when the waiting premium is decreasing in time, the optimal policy rule often has a \textit{nested} structure, which in turn implies that moral hazard always delays investment (\cref{sec:nested}).

\medskip

Assuming $\phi$ is constant, we may normalize it to $\phi\equiv 1$. Then, for $t \in \{1,2,\ldots,T\}$, the difference in the agent's expected payoff from adopting the stopping time $\tau_{t-1}$ rather than $\tau_t$ can be written as
\begin{align}
U_{t-1}-U_t 
= &\, \left(\prod_{i=1}^{t-1} p_i\right) \left[(1-p_t)w(0,y(\bfb_{t}))-\tilde A(t)\right], \label{eq:payoff-dif}
\end{align}
where $\tilde A(t)\:= - g(0,t-1) + p_t g(0,t) - (g(1,t-1) - g(0,t-1) + g(0,t) - g(1,t)) \prod_{i=t}^T p_i.$\footnote{In settings where $g$ is independent of $\t$, as in \cref{ex:renewing-subsidies}, we may write $g(0,t) = g(1,t) = g(t)$ and $\tilde A(t) = p_t g(t) - g(t-1).$}
To interpret \eqref{eq:payoff-dif}, first note that $\tau_{t-1}$ and $\tau_t$ only lead to different ex post payoffs for the agent if a breakdown does not arrive by period $t-1$, which happens with probability $\prod_{i=1}^{t-1} p_i$. Then, if a breakdown occurs in period $t$, the agent receives $w(0,y(\bfb_t))$ under $\tau_{t-1}$ but not under $\tau_t$. The term $\tilde A(t)$ captures the portion of the payoff difference between $\tau_{t-1}$ and $\tau_t$ that does not depend on the policy $y$.

Importantly, the payoff difference \eqref{eq:payoff-dif} depends on the policy rule only via $y(\bfb_t)$. This is because when $\phi$ is constant, the agent's payoff from the policy only depends on whether he stops, and not on when he stops; and because $\bfb_t$ is the only history under which the agent stops under $\tau_{t-1}$ but not under $\tau_t$.


It will be useful to define the agent's \textit{waiting premium} to be
$$ A(t) \:= \frac{\tilde A(t)}{1-p_t}.$$
Intuitively, the waiting premium $A(t)$ captures the policy-independent benefit that the agent gets by waiting one more period at time $t-1$.
One can see from \eqref{eq:payoff-dif} that the agent prefers $\tau_{t-1}$ to $\tau_t$ if and only if the extra exposure to policy, $w(0,y(\bfb_{t}))$, is greater than the waiting premium $A(t)$. The waiting premium may be positive or negative.

We have the following result.
\begin{prop}[State-Measurability]\label{prop:state-msrble}
    Suppose $\phi$ is constant, and $A$ is strictly increasing. Fix some $t \leq T-2$. If $y$ solves $\inner$, then $U_t>\max\{U_{t+1},\ldots, U_{T-1}\}$, and thus $y(\bfb_{t+1}) = y(\bfb_{t+2}) = \cdots = y(\bfb_{T})$.
\end{prop}
\begin{proof}
See \Cref{proof: sufficient condition for OTA}.
\end{proof}
\cref{prop:state-msrble} says that, in $[\ms P(t)]$, when the waiting premium is increasing, the optimal policy rule is measurable in the payoff-relevant state $\t$, conditional on a breakdown not occurring by time $t$.\footnote{A standard solution has $y(\bfb_s) = 0$ for $s\leq t$, so state-measurability only holds conditional on no breakdown until time $t$.} 
Intuitively, when the waiting premium increases over time, if the agent finds it optimal to wait at $t$, he should wait at least until the last period $T$. Hence, the agent's incentive constraints do not bind except possibly for the last two constraints, $U_t \geq U_T$ and $U_t \geq 0$.
Complementary slackness (\cref{thm:inner}(b)) then implies that the principal does not front-load certainty -- as long as a breakdown occurs after the agent stops, the policy does not depend on the timing of the breakdown. 
Note that by \cref{thm:aux=original}(b) and \cref{thm:inner=aux}(a), \cref{prop:state-msrble} also applies to any $(y,\tau_t)$ which solves $\original$.

\cref{prop:state-msrble} has a stark implication for how the optimal policy rule can be implemented.
Suppose we weaken the principal's commitment power as follows: the principal can only announce a state-measurable policy rule $y(\t)$, but she can make this announcement at a time of her choice. That is, in each period $t$, after the public signal $x_t$ is realized but before the agent acts, she can commit irreversibly to some rule $y(\t)$, if she has not done so already.\footnote{In the final period $T$, if the principal has not committed yet, she must choose a policy $y$.} 
This defines a dynamic game between the principal and the agent.
It can be verified that, as long as the conclusion of \cref{prop:state-msrble} holds, the solution to $\original$ can be implemented as a PBE of this dynamic game.\footnote{Suppose $(y,\tau_t)$ solves $\original$ and satisfies $y(\bfb_{t+1}) = \cdots = y(\bfb_T)$. Then, the dynamic game must have a PBE in which the principal announces the zero policy if a breakdown occurs by time $t$, and otherwise announces, in period $t$, the state-measurable rule $\hat y$ given by $\hat y(0) = y(\bfb_{t+1})$ and $\hat y(1) = y(\bfa)$. Clearly, the agent has no reason to deviate. If the principal has a profitable deviation, then the principal could have also profitably deviated in a similar manner from $(y,\tau_t)$ under $\original$, contradicting the optimality of $(y,\tau_t)$.
} 
Hence, when the waiting premium is increasing, the principal can implement the optimal policy rule even if she can only announce state-measurable policy rules.

\begin{customex}{1}[Renewing Subsidies]
Since $\phi$ is constant, as long as the waiting premium $A(t)$ is strictly increasing, \cref{prop:state-msrble} implies that it is enough for the principal to announce, at some point in time, a state-measurable subsidy rule (see \cref{fig:subs-ota}). 
To obtain conditions for the waiting premium to be increasing, write
\begin{align}\label{eq:A(t)-ex1}
    A(t+1) - A(t) =  \d^{t-1}(s-(1-\d)I)\left[ \frac{1}{1-p_t} - \frac{\d p_{t+1}}{1-p_{t+1}} \right].
\end{align}
Given our maintained assumption that $s \geq (1-\d)I$, we see that $A(t)$ is strictly increasing if and only if $s > (1-\d) I$, and
\begin{align}\label{eq:A(t)-ex1-suff}
    \frac{1-p_{t+1}}{1-p_t} > \d p_{t+1} \qquad \forall t\in \{1,2,\ldots,T-1\}.    
\end{align}
Inequality \eqref{eq:A(t)-ex1-suff} says that the ratio of successive breakdown probabilities (the LHS) is greater than the ``effective discount rate'' of $\d p_{t+1}$, which takes into account both the innate time discount $\d$ and the probability of no breakdown, $p_{t+1}$. A sufficient condition for \eqref{eq:A(t)-ex1-suff} is that $p_t$ is non-increasing. Intuitively, when the breakdown probability $1-p_t$ increases over time, the waiting premium also increases over time. See \cref{fig:subs-ota} for an illustration.
\begin{figure}[!h]
    \centering
    \includegraphics[width=0.6\linewidth]{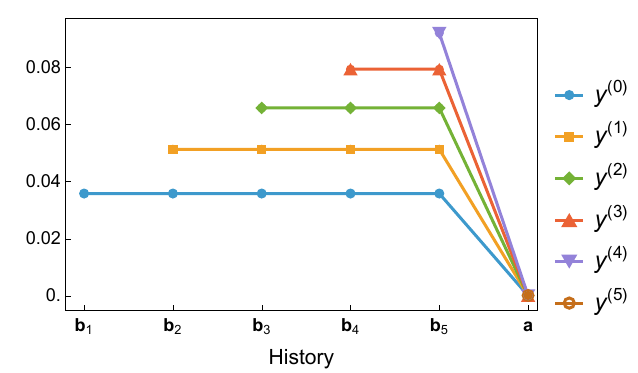}
    \caption{State-measurable policy rule (\cref{ex:renewing-subsidies} with increasing $A(t)$, $T=5$)}
    \medskip
    \small
    \justifying
    \noindent 
    The principal can implement the optimal policy rule by waiting until some period and then announcing a state-measurable policy rule. 
    (Parameters: $p_t=0.9$ for all $t$, $\delta=0.95, I=5, s=0.55, V=2$ and $c(y)=y^2$. Since $s > (1-\d)I$ and $p_t$ is constant, \cref{prop:state-msrble} applies.)
    
    \label{fig:subs-ota}
\end{figure}

To further interpret \eqref{eq:A(t)-ex1-suff}, suppose the transition probabilities $\{p_t\}$ arise from a discrete version of the Poisson bad news process. That is, suppose there is an underlying true state of the world $\omega \in \{0,1\}$ with prior $\pi = \Probability(\om=1) \in (0,1)$; conclusive bad news arrives in each period with probability $\lambda \in (0,1)$ conditional on $\omega = 0$; the state $\t=1$ represents the event in which no bad news has arrived by time $T$; and the ex post payoffs $u(1,y,t)$ and $v(1,y,t)$ are reinterpreted as expected payoffs across $\om$, conditional on no bad news up to time $T$. Then, the transition probability is given by 
$p_t = \frac{\pi + (1-\pi)(1-\l)^t}{\pi + (1-\pi)(1-\l)^{t-1}},$
and \eqref{eq:A(t)-ex1-suff} is equivalent to
\begin{align}\label{eq:A(t)-ex1-suff-Poisson}
    1-\l > \d p_t p_{t+1}.
\end{align}
In particular, \eqref{eq:A(t)-ex1-suff} and \eqref{eq:A(t)-ex1-suff-Poisson} hold whenever $1-\l \geq \d$. Intuitively, the waiting premium is increasing when learning happens slowly (low $\l$), so that the informational benefit from waiting remains high as time passes, and when discounting is severe (low $\d$), so that the net flow benefit in each period, which is foregone when delaying investment, depreciates quickly over time. \qed
\end{customex}

\section{Nested Policies and Investment Timing}\label{sec:nested}

In this section, we continue to assume a constant $\phi$ and show that when the waiting premium $A(t)$ is decreasing, the optimal solutions to inner problems have a nested structure (\cref{subsec:nested}). In turn, the nested structure will imply that moral hazard delays investment (\cref{subsec:investment-timing}). 

\subsection{Nested Policies} \label{subsec:nested}

We start with a definition.
\begin{defn}[Nestedness]
    Fix $\original$, and let $\{y^{(t)}\}_{t=0}^T$ be a sequence of policy rules such that each $y^{(t)}$ is a standard solution to $\inner$. The sequence $\{y^{(t)}\}_{t=0}^T$ is \textit{nested} if, for each $t \in \{0,1,\ldots,T-1\}$, we have $y^{(t)}(\bfh)=y^{(t+1)}(\bfh)$ for all $\bfh \in \{\bfb_{t+2}, \bfb_{t+3}, \ldots, \bfb_T, \bfa\}$.
\end{defn}
The solutions to inner problems are nested if, for each $t$, the solution to $[\ms P(t)]$ is identical to the solution to $[\ms P(t+1)]$, except possibly for $y(\bfb_{t+1})$ which equals 0 under any standard solution to $[\ms P(t+1)]$.
When solutions are nested, it is possible to solve the sequence of inner problems inductively -- first solve $[\ms P(T)]$ to obtain $y(\bfa)$, then solve $[\ms P(T-1)]$ while taking $y(\bfa)$ as given to obtain $y(\bfb_{T})$, and so on.
Our motivation for defining nestedness is mainly instrumental; we will show in \cref{subsec:investment-timing} that nestedness leads moral hazard to delay investment.

The following proposition provides a sufficient condition for solutions to be nested.

\begin{prop}[Nestedness]\label{prop: sufficient condition for nestedness}
    Suppose the following conditions hold:
    \begin{enumerate}[(i)]
        \item $\phi$ is constant.
        \item $A$ is strictly decreasing.
        \item For all $\t \in \{0,1\}$ and $t,t' \in T$,
        we have $\argmax_y v(\t,y,t) = \argmax_y v(\t,y,t') =: \uy(\t).$
        \item $u(1,\uy(1),T)>0$.
        \item $A(t) \neq w(0,\underline y(0))$  for every $t \in\{ 1, \ldots, T\}$.
    \end{enumerate}
    Then, the following statements are true:
    \begin{enumerate}
        \item The mapping $s\mapsto U_s(\ynull)$ 
        is strictly quasi-concave and has a unique maximizer $s = \hat t < \infty$.
        \item If $t\leq \hat t$, and $y$ solves $\inner$, then $U_t(y)=U_{t+1}(y)=\cdots =U_{\hat t}(y)$.
        \item $\inner$ has a unique standard solution $y^{(t)}$ for each $t \in T$, and $\{y^{(t)}\}_{t=0}^T$ is nested.
    \end{enumerate}
\end{prop}
\begin{proof}
See \Cref{proof: sufficient condition for nestedness}.
\end{proof}

Whereas \cref{prop:state-msrble} says that \textit{none} of the incentive constraints bind until the very last period if the waiting premium is increasing,  \cref{prop: sufficient condition for nestedness} says that \textit{every} incentive constraint binds until the time at which the agent would have stopped anyway (part (b)), and the solutions to inner problems are nested (part (c)), if the waiting premium is decreasing and additional conditions (iii) and (iv) hold.
Condition (iii) coincides with the first part of \cref{ass:rent} and is satisfied, for instance, if $v$ is decreasing in $y$ as in \cref{ex:renewing-subsidies}.
Condition (iv) says that even if the principal commits to the ex post optimal policy $\uy$, when the agent reaches the final period and observes that the state is good, he prefers to invest. Condition (v) holds generically and rules out the possibility that $U_s(\ynull)$, despite being strictly quasi-concave, is maximized at two adjacent values of $s$.

Part (a) follows from equation \eqref{eq:payoff-dif} --  $\ynull(\bfb_t)$ is by definition constant in $t$, and the waiting premium $A$ is decreasing in time, so the agent's expected payoff under $\ynull$ must be quasi-concave in $t$. 
To see why part (b) holds, suppose towards a contradiction that under a policy rule $y$ that solves $\inner$, the agent's expected utilities from simple stopping times have a ``dip''; for example, let us say $U_t(y) >U_{t+1}(y)$ and $U_{t+1}(y) < U_{t+2}(y)$. By complementary slackness (\cref{thm:inner}(b)), it must be that $y(\bfb_{t+1}) = y(\bfb_{t+2})$. Then, since $A(t+1) > A(t+2)$ by assumption, equation \eqref{eq:payoff-dif} implies that $U_t(y) - U_{t+1}(y) <0$ whenever $U_{t+1}(y) - U_{t+2}(y) < 0$, a contradiction. Part (c) follows because part (b) and equation \eqref{eq:payoff-dif} together imply $w(0,y^{(t)}(\bfb_s)) = A(s)$ for all $s \in \{t+1, t+2, \ldots, \hat t\}$, which uniquely pins down  $y^{(t)}(\bfb_s)$, independently of $t$.

\begin{customex}{1}[Renewing Subsidies] 
Among the conditions in \cref{prop: sufficient condition for nestedness}, (i) and (iii) hold by assumption. 
Condition (iv) holds if $1 > (1-\d)I$. We can see from equation \eqref{eq:A(t)-ex1} that (ii) holds if and only if $s > (1-\d)I$ and $\frac{1-p_{t+1}}{1-p_t}<\delta p_{t+1}$. The latter inequality says that, over time, the breakdown probability $1-p_t$ decreases quickly enough, so that the agent's waiting premium also decreases over time. Under these conditions, \cref{prop: sufficient condition for nestedness} tells us that the solutions to inner problems must be nested (see \cref{fig:subs-nest}). \qed



\begin{figure}[!htb]
    \centering
    \includegraphics[width=0.6\linewidth]{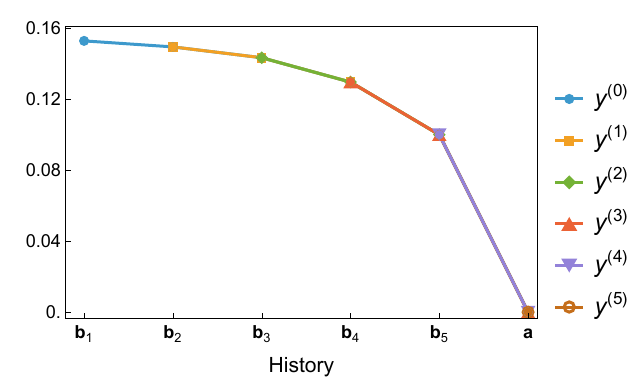}
    \caption{Solutions to inner problems (\cref{ex:renewing-subsidies}, decreasing $A(t)$, $T=5$)}
    \medskip
    \small
    \justifying
    \noindent 
    The solutions to inner problems are nested. For example, $y^{(2)}$ coincides with  $y^{(3)}$ at every history except $\bfb_{3}$. 
    (Parameters: $(p_t)_{t=1}^5=(\frac15,\frac35,\frac34,\frac56,\frac89)$, $\delta=0.9, I=3, s=0.5, V=0.1$ and $c(y)=y^2$. The waiting premium is decreasing: $A(1) = 0.90,\; A(2) = 0.88,\; A(3) = 0.85,\; A(4) = 0.77,$ and $A(5) = 0.59$.)
    \label{fig:subs-nest}
\end{figure}
\end{customex}

\subsection{Investment Timing}\label{subsec:investment-timing}

We now provide a sufficient condition for moral hazard to delay investment, meaning the second-best investment timing is (weakly) later than the first-best investment timing. 

\begin{prop}[Delayed Investment]\label{prop: nestedness implies delays}
    Suppose the following conditions hold:
    \begin{enumerate}[(i)]
        \item $\phi$ is constant.
        \item The principal's utility function can be written as $v(\t,y,t) = c(\t,y) + f(\t,t)$ for some functions $c$ and $f$.\footnote{This condition is satisfied by \cref{ex:renewing-subsidies}.}
        \item There exists a nested sequence $\{y^{(t)}\}_{t=0}^T$ of second-best standard solutions to inner problems. 
    \end{enumerate} 
    Then, if $(y^{SB},\tau_{t^{SB}})$ solves $\original$, there exists a solution $(y^{FB},\tau_{t^{FB}})$ to $[\ms P^{FB}]$ with $t^{FB} \leq t^{SB}$.
\end{prop}
\begin{proof}
See \Cref{proof: nestedness implies delays}.
\end{proof}
Intuitively, when the second-best solutions are nested, 
the marginal cost to the principal of inducing the agent to invest one period earlier is always higher under second-best than under first-best, so investment must be delayed under second-best. The delay is sometimes strict, as illustrated for the motivating example in \cref{OA:motivating-example}.

Perhaps surprisingly, when second-best solutions are not nested, it is possible for moral hazard to \textit{hasten}, rather than delay, investment. We provide such an example in \cref{OA:hastening}.

\section{Related Literature}
\label{sec:lit}


There is a long-standing literature on investment under uncertainty \citep{bernanke1983irreversibility, dixit1994investment, friedman1968role,rodrik1991policy}.\footnote{The \textit{bad news principle} by \cite{bernanke1983irreversibility} states that, in a single-person decision problem of whether to invest early, what matters is the severity of bad future outcomes. Our model studies the strategic implications of this principle: we ask how a principal should optimally control an agent's future outcomes in order to influence the agent's investment timing.}
A growing empirical literature has made progress in quantifying the \textit{benefits} of policy certainty, showing that greater certainty about fiscal, trade, or environmental policies encourages investment \citep{fernandez2015fiscal,gulen2016policy, handley2017policy, gowrisankaran2024policy, chen2024dynamic}.\footnote{See also \cite{baker2016measuring}, which develops a measure of policy uncertainty.} However, much less attention has been paid to the \textit{cost} of certainty: policy becomes less responsive to future states. The current paper complements this literature by developing a theoretical framework to understand how a policymaker should endogenously provide certainty, trading off its benefits against its costs.

A literature on dynamic commitment under uncertainty \citep{athey2005optimal,halac2014fiscal,halac2022instrument, choi2024i'll} has  studied settings where future information arrives privately to the principal. In our language, 
\cite{choi2024i'll} considers a principal who delays an agent's stopping decision by promising today to provide certainty tomorrow. This literature has focused on how the principal incentivizes her future self to be truthful after receiving private information. 
In contrast, we study an environment where the principal and the agent observe the same public signal. This is highly relevant in many settings, such as in designing environmental or macroeconomic policies, but remains understudied. 
Enriching the agent's moral hazard constraint while dropping the principal's incentive constraints allows us to tractably analyze the optimal dynamic policy rule.

The literature has studied how to incentivize an agent to acquire -- or refrain from acquiring -- private information before signing a contract \citep{cremer1992gathering,cremer_contracts_1998}, exerting effort \citep{lewis_information_1997}, or buying a security \citep{boot_security_1993,fulghieri2001information,inderst2006informed,yang2019financing,yang_optimality_2020,li2023security}. In contrast, we consider a real options model, where information is public and arrives exogenously over time, and the agent makes an irreversible investment decision. This allows us to abstract away from adverse selection associated with private learning and identify the intrinsic trade-off that the principal faces, between adapting her policy to uncertain future states and providing certainty to discourage the agent from waiting. Also, the dynamic nature of our model allows us to characterize how the optimal policy responds to information arrival and how the agent's rent depends on the way delay affects his payoff.\footnote{In static models of learning, the cost of information acquisition is typically assumed to be additively separable from the agent's payoff from the underlying decision. In contrast, in our model, the cost of delaying investment interacts with policy and state through $\phi(t)$.}

More broadly, various strands of literature have studied how option value makes an agent wait for information. \cite{mcclellan2024dynamic} shows that, in a dynamic negotiation problem, the principal should incentivize the agent to continue negotiating by giving the agent more time to explore his time-varying outside option. In models of social learning, agents wait until others act in order to learn from their actions \citep{chamley_information_1994, murto_learning_2011,frick_innovation_2024,fonseca_staggered_nodate, bobtcheff_information_2025, laiho_gradual_2025}.
A sender can disclose information about an exogenous state of the world to influence a receiver's stopping decision \citep{ely2017beeps,ely2020moving,orlov2020persuading,saeedi2024getting,koh2024persuasion}.\footnote{There are also principal-agent models where the agent chooses a stopping time, but this choice is observed by the principal \citep{grenadier2005investment,board2007selling,kruse2015optimal}. In these papers, a key source of the agency problem is that the agents have hidden information.}

\cite{kremer2004strong} and \cite{kremer2020advance} have argued that a government can incentivize vaccine development by making an advance market commitment to purchase successfully developed vaccines. The application of our model to the procurement problem speaks to the optimal design of the advance market commitment when the government faces uncertainty about its future demand for vaccines.




\section{Conclusion}

This paper identifies a fundamental agency friction: to encourage an agent to invest, a principal must sacrifice her ability to adapt to future information. We believe this trade-off is relevant in many economic environments and constitutes a useful building block for future theoretical and empirical work.

Whether real-world policymakers provide too much or too little certainty remains an important empirical question. In recent years, many commentators have raised concerns that uncertainty about future policy is hindering investment. A report commissioned by the International Chamber of Commerce argues that policy uncertainty reduced real business investment by \$202 billion in 2025.\footnote{See \cite{icc2026policy}.}  Nevertheless, highly uncertain policies may be optimal when there are sufficiently large benefits to flexibly adapting policy to future states. This paper suggests how empirical evidence on the value of policy flexibility could inform this debate. To the extent that real-world policies do not provide certainty optimally, our paper offers normative guidance for their design.

\bibliographystyle{economet}
\bibliography{reference}

\appendix

\section{Appendix}\label{sec:appendix}

\subsection{Additional Notation} 
For any $t\in T$, and any $s\leq t$, define the \textit{partial histories}
\[\mathbf{a}_t:=(\overbrace{1,\dots,1}^{t \text{ entries}}),\qquad \mathbf{b}_{s,t}:=(\overbrace{1,\dots,1}^{s-1 \text{ entries}},\overbrace{0,\dots,0}^{t-s+1\text{ entries}}).\]
The partial history $\mathbf{a}_t$ indicates that it is period $t$ and there has been no breakdown. The partial history $\mathbf{b}_{s,t}$ indicates that it is period $t$ and there was a breakdown at period $s$. 
The partial history $\mathbf{a}_0$ indicates that it is period 0. 
Let 
$\ms{G}:=\{\mathbf{a}_0\}\cup(\cup_t \{\mathbf{a}_t\})\cup (\cup_{t,s\leq t}\{\mathbf{b}_{s,t}\})$
denote the set of possible partial histories. Notice that a history is simply a partial history of length $T$, and that, accordingly, $\ms{H}\subset \ms{G}.$

Say that a partial history $\mathbf{g}'$ of length $t$ is a \textit{sub-history} of a strictly longer partial history $\mathbf{g}$, denoted $\mathbf{g}'\prec_{sh}\mathbf{g}$, if $\mathbf{g}'$ coincides with the first $t$ entries of $\mathbf{g}$.
The partial history $\mathbf{a}_0$ is defined to be a sub-history of any other partial history. 

A map $d: \ms G \to \{\text{\qcr{wait}},\text{\qcr{stop}},\varnothing\}$ is a \textit{stopping rule} if
\[d(\bfg)=\varnothing \iff \exists\, \bfg'\prec_{sh}\bfg \text{ with }d(\bfg')=\text{\qcr{stop}}.\]
Given a signal process, there is a one-to-one relationship between stopping rules $d$ and stopping times $\tau$. We will use $d_{\tau}$ to denote the stopping rule induced by stopping time $\tau$, and $\tau_d$ to denote the stopping time induced by stopping rule $d$.

For a partial history $\bfg$, define
\[p(\bfg):=\begin{cases}
\left(\prod_{i=1}^{s-1}p_i\right)(1-p_s)  \caseif \bfg=\bfb_{s,t}\\
\prod_{i=1}^{t}p_i  \caseif \bfg=\bfa_t
\end{cases}\]
to be the probability that the partial history $\bfg$ will be reached at some point. Notice that $p$ is also defined for histories, via the relationships $\bfa=\bfa_T$ and $\bfb_t=\bfb_{t,T}$. 

Next, define
$U_{\tau}(y)\:=\E[u(\theta,y(\bfh),\tau)]$ and $V_\tau(y)\:=\E[v(\theta,y(\bfh),\tau)]$
to be the players' expected utilities if the agent adopts $\tau$ and the policy rule is $y$.
When it is clear from context, we will omit the argument $y$ and simply write $U_{\tau}$.

Let us present a useful formula for $U_{\tau}$ and $V_\tau$. In general, a stopping time $\tau$ will specify, first, that the agent stops at at most one partial history at which there has been no breakdown. Call this partial history $\bfa_t$ (with $t=\infty$ if there is no such partial history). Second, it will specify some collection of partial histories $\ms R:=\{(r,s):d(\bfb_{r,s})=\ttstop\}$ where the agent stops, following a breakdown. We then have
\begin{align}
    U_{\tau}(y)&\,=\sum_{(r,s)\in\ms R} p(\bfb_{r}) u(0,y(\bfb_r),s)+U_t(y),\label{eq:u_formula} \\
    V_\tau(y)&\,=\sum_{(r,s)\in\ms R} p(\bfb_{r}) v(0,y(\bfb_{r}),s)+V_t(y) + \sum_{r:\, (r,s) \notin \ms R \ \forall s\geq r} p(\bfb_r)v(0,y(\bfb_r),\infty),\label{eq:v_formula}
\end{align}
where $V_t(y):=\sum_{i=t+1}^T p(\bfb_i)v(0,y(\bfb_i),t)+p(\bfa)v(1,y(\bfa),t)$. If the policy $y$ satisfies $y(\bfh) = 0$ for $\bfh \in \{\bfb_1,\ldots,\bfb_t\}$, then $V_t(y)$ represents the principal's expected utility under $(y,\tau_t)$. Note that the second summation in \eqref{eq:v_formula} is non-positive by assumption.

It will also be useful to define, for $s\in T$ and $t \in \Tui$ with $s \leq t$,
\begin{align}
    U_{s|t}(y)\:= \begin{dcases}
        \sum_{i=t+1}^{T}p(\bfb_i)u(0,y(\bfb_i),s)+p(\bfa)u(1,y(\bfa),s) \caseif t<\infty \\
        p(\bfa)u(1,y(\bfa),s) \caseif t=\infty,
        \end{dcases} \label{eq:A_payoff_contrib}\\
    V_{s|t}(y)\:= \begin{dcases}
        \sum_{i=t+1}^{T}p(\bfb_i)v(0,y(\bfb_i),s)+p(\bfa)v(1,y(\bfa),s) \caseif t<\infty \\
        p(\bfa)v(1,y(\bfa),s) \caseif t=\infty.\label{eq:P_payoff_contrib}
    \end{dcases}
\end{align}
That is, $U_{s|t}(y)$ ($V_{s|t}(y)$) represents the contribution to the agent's (principal's) utility from breakdowns later than $t$ if the agent adopts $\tau_s$ and the policy is $y$. Again, we sometimes omit the argument $y$.

\subsection{\texorpdfstring{Proof of \cref{thm:inner}}{Proof of Theorem 1}}\label{proof:thm:inner}

Since $U_t - U_{t'}$ is concave in $y$ for all $t' \geq t$, and since we may without loss set $y(\bfb_s) = 0$ for all $s\leq t$, $\inner$ is a concave maximization program with a finite number of constraints, satisfying strong duality. 
For $s \in \{1,2,\ldots , T-t\}$, let $\lambda_s$ denote the Lagrange multiplier corresponding to the constraint $U_t\geq U_{t+s}$, and let $\l_\infty$ be the multiplier for $U_t \geq U_\infty$. The first-order condition of the Lagrangian $\ms L$ with respect to $y(\bfa)$ is
    \begin{equation}
    \label{eq_foc_ya}
        \frac{\partial \ms L} { \partial  y(\bfa)} = v_y(1,y(\bfa),t)+ \sum_{r=1}^{T-t} \lambda_r\left[\phi(t)w_y\left(1,y(\bfa)\right) - \phi(t+r)w_y\left(1,y(\bfa)\right) \right] + \lambda_\infty \phi(t)w_y\left(1,y(\bfa)\right) \leq 0,
    \end{equation}
    where the inequality is strict only if $y(\bfa)$ is equal to its lower bound of 0.
    The first-order condition with respect to $y(\bfb_{t+s})$ is
    \begin{align}
    \begin{split}
        \label{eq_foc_yb}
        \frac{\partial \ms L} { \partial  y(\bfb_{t+s})} = &\, v_y(0,y(\bfb_{t+s}),t)+\sum_{r=1}^{s-1} \lambda_r\left[\phi(t)w_y\left(0,y(\bfb_{t+s})\right) - \phi(t+r)w_y\left(0,y(\bfb_{t+s})\right)\right] \cr
        &\, + \sum_{r=s}^{T-t} \lambda_r\phi(t)w_y\left(0,y(\bfb_{t+s})\right) +\lambda_\infty \phi(t)w_y\left(0,y(\bfb_{t+s})\right) 
        \leq 0,
        \end{split}
    \end{align}
where the inequality is strict only if $y(\bfb_{t+s})=0$.

We will prove part (b) first, and then prove parts (a) and (c).

\paragraph*{Proof of (b).} 
To show \eqref{eq:thm1-b}, we will argue that $y(\bfb_{t+s}) \geq y(\bfb_{t+s+1})$ for each $s$. This holds trivially if $y(\bfb_{t+s+1}) = 0$, so let us assume that $y(\bfb_{t+s+1}) > 0$, which implies $\partial \ms L / \partial y(\bfb_{t+s+1}) = 0.$
Suppose towards a contradiction that $y(\bfb_{t+s})<y(\bfb_{t+s+1})$. Then, by using the assumptions that $v$ is strictly concave, $u$ is strictly increasing and concave, and $\phi$ is positive and non-increasing, one can verify that $\partial \ms L / \partial y(\bfb_{t+s}) > \partial \ms L / \partial y(\bfb_{t+s+1}) $. Since $\partial \ms L / \partial y(\bfb_{t+s}) \leq 0$, we must have $\partial \ms L / \partial y(\bfb_{t+s+1}) < 0$, a contradiction.

To prove the remainder of (b), first suppose $y(\bfb_{t+s}) = y(\bfb_{t+s+1}) > 0$. We then have 
$\partial \ms L / \partial y(\bfb_{t+s}) = \partial \ms L / \partial y(\bfb_{t+s+1}) = 0.$
After canceling out all terms except the ones involving $\l_s$, the equality $\partial \ms L / \partial y(\bfb_{t+s}) = \partial \ms L / \partial y(\bfb_{t+s+1})$ reduces to
$ 
\lambda_s \phi(t) = \lambda_s (\phi(t)-\phi(t+s)).
$
It must then be that $\lambda_{s}=0$, i.e., the constraint $U_t\geq U_{t+s}$ is slack. 

Next, suppose $y(\bfb_{t+s}) > y(\bfb_{t+s+1})$. We must then have $y(\bfb_{t+s})>0$ and $\partial \ms L / \partial y(\bfb_{t+s}) = 0$, implying
$\partial \ms L / \partial y(\bfb_{t+s}) \geq \partial \ms L / \partial y(\bfb_{t+s+1})$. 
Since $v$ is strictly concave, $u$ is strictly increasing and concave, and $\phi$ is positive and non-increasing, the inequality $\partial \ms L / \partial y(\bfb_{t+s}) \geq \partial \ms L / \partial y(\bfb_{t+s+1})$ holds only if
$\lambda_s \phi(t) > \lambda_s(\phi(t)-\phi(t+s)).$ 
This inequality fails if $\l_s=0$, so it must be that $\lambda_s>0$. That is, $U_t\geq U_{t+s}$ must bind. 

\paragraph*{Proof of (a).}


To prove \eqref{eq:thm1-a}, first consider the case $\partial \ms L / \partial y(\bfa) < 0$, implying $y(\bfa)=0$. Note that we have
\begin{align*}
    \frac{u_y(0,y(\bfb_T),t)}{u_y(1,0,t)} \leq \frac{u_y(0,0,t)}{u_y(1,0,t)} \leq 1 \leq \frac{v_y(0,0,t)}{v_y(1,0,t)} \leq \frac{v_y(0,y(\bfb_T),t)}{v_y(1,0,t)},
\end{align*}
where the first inequality follows from the concavity of $u$, the second inequality from $u_y > 0$ and \cref{ass:A-payoff-main}(e), the third inequality from $v_y(1,0,t) < 0$ (which is necessary for $\partial \ms L / \partial y(\bfa) < 0$) and \cref{ass:P-payoff-main}(f), and the last inequality from the concavity of $v$. Comparing the first and the last fractions and rearranging gives
\begin{align*}
    \frac{v_y(0,y(\bfb_T),t)}{u_y(0,y(\bfb_T),t)} \leq \frac{v_y(1,0,t)}{u_y(1,0,t)}  = \frac{v_y(1,y(\bfa),t)}{u_y(1,y(\bfa),t)},
\end{align*}
giving us \eqref{eq:thm1-a}.
Next, consider the case $\partial \ms L / \partial y(\bfa) = 0$.
The first-order conditions with respect to $y(\bfb_T)$ and $y(\bfa)$ can be rearranged to
    \begin{align}
    \frac{v_y(0,y(\bfb_T),t)}{u_y(0,y(\bfb_T),t)} &\,\leq
    -\sum_{r=1}^{T-t-1}\lambda_r\left(\frac{\phi(t)-\phi(t+r)}{\phi(t)}\right)
    -\lambda_{T-t}-\lambda_\infty \label{eq:thm1-proof-4}\\
    \frac{v_y(1,y(\bfa),t)}{u_y(1,y(\bfa),t)} &\,=
    -\sum_{r=1}^{T-t}\lambda_r\left(\frac{\phi(t)-\phi(t+r)}{\phi(t)}\right)-\lambda_\infty. \label{eq:thm1-proof-5}
    \end{align}
Suppose towards a contradiction that
\eqref{eq:thm1-a} fails. Conditions \eqref{eq:thm1-proof-4} and \eqref{eq:thm1-proof-5} then imply $\lambda_{T-t} \phi(t)<\lambda_{T-t}\left(\phi(t)-\phi(T)\right).$
This is a contradiction because $\lambda_{T-t}\geq 0$ and $\phi(T) >0$. We have thus proved \eqref{eq:thm1-a}.

To prove the remainder of (a), first suppose $y(\bfa)>0$, $y(\bfb_T)>0$, and that \eqref{eq:thm1-a} holds as an equality. Since \eqref{eq:thm1-proof-4} and \eqref{eq:thm1-proof-5} hold with equalities, equating their right-hand sides gives
$\lambda_{T-t} \phi(t) = \lambda_{T-t}\left(\phi(t)-\phi(T)\right),$
which implies $\lambda_{T-t}=0$.
Next, suppose $y(\bfa)>0$, $y(\bfb_T)>0$, and \eqref{eq:thm1-a} holds as a stricty inequality.
We then have
$\lambda_{T-t} \phi(t) > \lambda_{T-t}\left(\phi(t)-\phi(T)\right),$
implying $\lambda_{T-t}>0$.

\paragraph*{Proof of (c).}
We consider two cases.
First, suppose the constraint $U_t\geq U_{t+s}$ holds as an equality for some $s\in \{1,2,\ldots,T-t\}$. Without loss, let $s$ be the smallest such number.
By part (b), we must have $y(\bfb_{t+1})=\cdots = y(\bfb_{t+s})$.

Using \eqref{eq:A_payoff_contrib}, we may write
\begin{align*}
    U_{t+s} = U_t= u(0,y(\bfb_{t+1}),t)\sum_{i=1}^{s} p(\bfb_{t+i}) 
    + U_{t|t+s}.
\end{align*}
We have $U_{t+s}\geq 0$, so by APM, we have $U_{t|t+s} \geq U_{t+s}$. Hence, it must be that  $u(0,y(\bfb_{t+1}),t)\leq 0$. By part (b), we then have $u(0,y(\bfb_{t+r}),t)\leq 0$ for all $r\in \{1,\ldots,T-t\}$. But since $U_t\geq 0$, we must have $u(1,y(\bfa),t)\geq 0$. We thus have $u(1,y(\bfa),t)\geq u(0,y(\bfb_{t+1}),t)$, as desired.

Second, suppose that the constraint $U_t\geq U_{t+s}$ holds as a strict inequality for every $s\in \{1,2,\ldots,T-t\}$. Then, the last statement in part (b) implies $y(\bfb_{t+1})=\cdots = y(\bfb_{T})=:y(\bfb)$. Suppose towards a contradiction that $u(0,y(\bfb),t)>u(1,y(\bfa),t)$, which implies $y(\bfb)>y(\bfa)$ by \Cref{ass:A-payoff-main}(e). Then, condition \eqref{eq_foc_ya}, \Cref{ass:P-payoff-main}(b), and \Cref{ass:P-payoff-main}(f)  imply
$0 \geq v_y(1,y(\bfa),t) > v_y(0,y(\bfb),t)$,
while \cref{ass:A-payoff-main}(a) and \cref{ass:A-payoff-main}(e) imply $u_y(1,y(\bfa),t) \geq u_y(0,y(\bfb),t) > 0$. We therefore have
\begin{align}\label{eq:thm1-proof-8}
            \dfrac{v_y(0,y(\bfb_T),t)}{u_y(0,y(\bfb_T),t)}< \dfrac{v_y(1,y(\bfa),t)}{u_y(1,y(\bfa),t)}. 
\end{align}
On the other hand, since $U_t > U_{t+s}$ for every $s$, by complementary slackness, we have $\l_1 = \ldots = \l_{T-t} = 0$. So, the first-order conditions \eqref{eq_foc_ya} and \eqref{eq_foc_yb} can be rearranged to
\begin{align*}
    \dfrac{v_y(0,y(\bfb_T),t)}{u_y(0,y(\bfb_T),t)} = -\l_{\infty} \geq \dfrac{v_y(1,y(\bfa),t)}{u_y(1,y(\bfa),t)},
\end{align*}
contradicting \eqref{eq:thm1-proof-8}. This proves the result. \qed



\subsection{\texorpdfstring{Proof of \cref{lem:standard-solutions-exist}}{Proof of Lemma \ref{lem:standard-solutions-exist}}}\label{proof:lem:standard-solutions-exist}

\paragraph*{Proof of (a).}


Suppose the feasible set of $\inner$ is nonempty. We claim that it is without loss of optimality to impose the constraint $y \leq \bar y$ for some large enough $\bar y >0$. To see this, first consider a history $\bfh \in \{\bfb_{t+1},\ldots,\bfb_T,\bfa\}$, for which we have $\tau_t(\bfh)  = t < \infty$. It cannot be optimal for the principal to set $y(\bfh)$ arbitrarily high; this is because a feasible solution exists, $v(\t,y,t) \to -\infty$ as $y \to \infty$ (\cref{ass:P-payoff-main}(c)), and $\bfh$ occurs with positive probability. Next, consider a history $\bfh \in \{\bfb_1,\ldots,\bfb_t\}$, for which it must be that $\tau_t(\bfh) = \infty$. It is without loss to set $y(\bfh) = 0$, since this weakly increases the principal's payoff (\cref{ass:P-payoff-main}(c)) and does not affect the agent's incentive constraints. Hence, we may impose $y \leq \bar y$. Then, by the Weierstrass theorem, there exists an optimal solution to $\inner$.


Now, take any optimal solution $y$. Consider the policy rule $\tilde{y}$ that satisfies $\tilde{y}(\bfb_s)=0$ for all $s \leq t$ and is otherwise identical to $y$. Clearly, $\tilde y$ is a standard solution to $[\ms P(t)]$. \qed

The proof of (b) is analogous to the proof of (a) and is omitted.

\subsection{\texorpdfstring{Proof of \cref{thm:inner=aux}}{Proof of Theorem \ref{thm:inner=aux}}}\label{proof:thm:inner=aux}

We prove part (b) first, and then part (a).

\paragraph*{Proof of (b).} 

It will be enough to argue that, whenever $(y, \tau_t)$ itself does not solve $[\ms P]$, one can find $t' < t$ such that $t'\in \argmax_s W(s)$, $y$ solves $[\ms P(t')]$, and $(y, \tau_{t'})$ solves $[\ms P]$.

Suppose $(y, \tau_t)$ does not solve $[\ms P]$. Then, $(y,\tau_t)$ must be infeasible under $[\ms P]$ because the agent wants to deviate to an earlier stopping time. Take any $t' \in \argmax_{s \in \{0,\ldots,T,\infty\}} U_s (y)$. Clearly, we have $t' < t$.
Fixing the policy rule at $y$, we will argue that $V_{t'}(y) \geq V_t(y)$; that is, the principal's expected payoff is higher if the agent plays $\tau_{t'}$ than $\tau_{t}$. If $t=\infty$, then $y$, being a standard solution, must satisfy $y\equiv 0$. Then, \cref{ass:P-payoff-main}(d) implies $V_{t'}(y) \geq V_t(y)$. Hence, let us suppose $t<\infty$.

Using \eqref{eq:P_payoff_contrib} and the fact that $y$ is a standard solution to $\inner$, we can write
$$ V_{t'}(y) = \sum_{i=t'+1}^t p(\bfb_i) v(0,0,t') + V_{t'|t}(y). $$
First, \cref{ass:P-payoff-main}(a) and \cref{ass:P-payoff-main}(d) imply $v(0,0,t') \geq 0$. Next, due to \cref{ass:P-payoff-main}(a) and (c), $[\ms P(\infty)]$ has value $W(\infty) = 0$. Since $t$ maximizes $W(t)$, and $y$ solves $\inner$, we must have $V_t(y) \geq 0$. Then, APM implies $V_{t'|t}(y) \geq V_t(y)$. We have thus shown that $V_{t'}(y) \geq V_t(y)$.

We must then have $W(t') \geq W(t)$, since $y$ is feasible for $[\ms P(t')]$. However, we assumed $t \in \argmax_t W(t)$, so it  must be that $W(t') = W(t)$. Therefore, $y$ solves $[\ms P(t')]$, and $t'\in \argmax_s W(s)$. Since the agent's best response to $y$ among all simple stopping times is $\tau_{t'}$, the pair $(y, \tau_{t'})$ is feasible under, and therefore solves, $[\ms P]$. This proves the claim.

\paragraph*{Proof of (a).}

Suppose $(y,\tau_{t})$ solves $[\ms P]$. Then, $y$ must be feasible under $[\ms P(t)]$. Moreover, \cref{lem:standard-solutions-exist}(a) and \cref{thm:inner=aux}(b) imply that the value of $[\ms P]$ is equal to $\max_s W(s)$, which cannot be less than $W(t)$. 
Therefore, $y$ solves $[\ms P(t)]$, and $t \in \argmax_{s} W(s)$. \qed

\subsection{\texorpdfstring{Proof of \cref{thm:aux=original}}{Proof of Theorem \ref{thm:aux=original}}}\label{app:thm-2-proof}

We first prove 
\cref{lem:p_sol_is_feas}, which says that any standard solution to $[\ms P]$ is feasible for $[\ms P_0]$ (\ref{proof:lem:p_sol_is_feas}). We then prove \cref{lem:principal_wants_simple}, which says that there exists a simple stopping time that solves $[\ms P_0]$ (\ref{proof:lem:principal_wants_simple}). Finally, we combine \cref{lem:p_sol_is_feas,lem:principal_wants_simple} to establish an equivalence between solutions to $[\ms P]$ and to $[\ms P_0],$ proving \cref{thm:aux=original} (\ref{proof:thm:aux=original}).

\subsubsection{\texorpdfstring{Proof of \cref{lem:p_sol_is_feas}}{Proof of Lemma \ref{lem:p_sol_is_feas}}}\label{proof:lem:p_sol_is_feas}

The following useful property of $u$ is an immediate consequence of APM.
\begin{lem}[Positive monotonicity]
\label{lem:pos-mon}
Let $f$ be a payoff function satisfying APM. For any $\theta\in\{0,1\}$, $y\in \R_+$ and $t\in T$ such that $f(\theta,y,t)\geq 0,$ and any $s<t$, we have:
\[f(\theta,y,s)\geq f(\theta,y,t).\]
\end{lem}
\begin{proof}
    Use the single-element sequence $(1,\theta,y)$ in the definition of APM.
\end{proof}

We will first prove \cref{lem:2_part_a}, which represents the special case of \cref{lem:p_sol_is_feas} when at least one of the agent's incentive constraints binds.

\begin{lem}
\label{lem:2_part_a}
    Fix $t^*\in T$, and let $y$ be a standard solution to $[\ms P(t^*)]$, such that the constraint $U_{t^*}\geq U_{t}$ binds for some $t\in \{t^*+1,\dots,\infty\}$. Then, we must have $u(0,y(\bfb_{\hat{s}}),\hat{t})\leq 0,$ for all $\hat{t}\in T$ and all $\hat{s}\leq \hat{t}$. 
\end{lem}
\begin{proof}

    Suppose not -- i.e., there exists some such $\hat{s},\hat{t}$ with $u(0,y(\bfb_{\hat{s}}),\hat{t})> 0$. Then, by positive monotonicity, it must be that $u(0,y(\bfb_{\hat{s}}),\hat{s})> 0.$ Let $s$ be the largest value of $\hat{s}$ satisfying $u(0,y(\bfb_{\hat{s}}),\hat{s})> 0.$ Since $(y,\tau_{t^*})$ is a standard solution, by \cref{ass:A-payoff-main}(f), we have $s>t^*$.
    \medskip

    \noindent \textit{Claim 1.} For any $s'\in(t^*,s],$ we have $u(0,y(\bfb_{s'}),s')>0.$ 
    
    \noindent\textit{Proof.} By \cref{thm:inner}, we have that $y(\bfb_{s'})\geq y(\bfb_{s}).$ Then, we have:
\begin{align*}u(0,y(\bfb_{s'}),s')&\geq u(0,y(\bfb_s),s') &[\text{Monotonicity in }y]\\
&\geq u(0,y(\bfb_s),s)>0.&[\text{Positive monotonicity}]
\end{align*}
\qed

Without loss, take $t$ to be the earliest period at which $U_{t^*}\geq U_t$ binds. We will argue case-wise that such a $t$ cannot exist.
\medskip

\underline{Case 1:} $t\leq s.$

Observe that $U_t\geq 0$ (since $U_t=U_{t^*}\geq U_{\infty}=0$), and that $u(0,y(\bfb_t),t)> 0$ (by \textit{Claim 1}). Then, write:
\begin{align*}U_{t-1}-U_{t}= &\, p(\bfb_{t})u(0,y(\bfb_{t}),t-1)\\
    &\, +p(\bfa)u(1,y(\bfa),t-1)+\sum_{i=t+1}^T p(\bfb_i)u(0,y(\bfb_{i}),t-1)\\
    &\, -p(\bfa)u(1,y(\bfa),t)-\sum_{i=t+1}^T p(\bfb_i)u(0,y(\bfb_{i}),t),
    \end{align*}
    where the third line is $- U_{t}.$ Then, APM implies that the sum of the second and third lines is non-negative, while positive monotonicity implies that the first line is strictly positive. Thus, $U_{t-1}>U_{t}$. But, since $y$ is feasible for $[\ms P(t^*)]$, we must have $U_{t^*}\geq U_{t-1},$ so that $U_t=U_{t^*}$ cannot hold.
\medskip

\underline{Case 2:} $s<t<\infty.$

Since $t$ is the earliest binding period, the contrapositive of the second statement in \cref{thm:inner}(b) implies that
\begin{equation}
y(\bfb_{t^*+1})=\dots =y(\bfb_t). \label{eq_comp_slack_lem_5}
\end{equation}
Now, let $m:=u(0,y(\bfb_{s}),t^*).$ By positive monotonicity, we have $m\geq u(0,y(\bfb_{s}),s)> 0$. By  \eqref{eq_comp_slack_lem_5}, if the agent stops at $t^*,$ then he is guaranteed at least $m$ from any breakdown occurring between $t^*$ and $t$. We therefore have:
\[U_{t^*}\geq P_{(t^*,t]}m+U_t,\]
where $P_{(t^*,t]} \:= \sum_{i=t^*+1}^t p(\bfb_i)$ denotes the ex ante probability of a breakdown occurring in the interval $(t^*,t]$. This contradicts the assertion that $U_t=U_{t^*},$ since $P_{(t^*,t]}$ and $m$ are both strictly positive.
\medskip

\underline{Case 3:} $t=\infty.$

Since $t$ is the earliest binding period, the contrapositive of the second statement in \cref{thm:inner}(b) implies that
\begin{equation*}
y(\bfb_{t^*+1})=\dots =y(\bfb_T) := y(\bfb).
\end{equation*}
Defining $m\:= u(0,y(\bfb),t^*)$, we have $m \geq u(0,y(\bfb),s) > 0$. Then,
$$U_{t^*} = P_{(t^*,T]}m + p(\bfa)u(1,y(\bfa),t^*), $$
where $P_{(t^*,T]} \:= \sum_{i=t^*+1}^T p(\bfb_i)$. However, by \cref{thm:inner}(c), we have $u(1,y(\bfa),t^*) \geq u(0,y(\bfb),t^*) >0$, implying $U_{t^*} > 0 = U_\infty$, a contradiction.
\end{proof}

Next, \cref{lem:p_sol_is_feas} uses \cref{lem:2_part_a} to conclude that the agent cannot profitably deviate to some non-simple stopping strategy against the principal's solution to $[\ms P].$
\medskip

\noindent \textbf{\cref{lem:p_sol_is_feas}}. \textit{Let $(y,\tau_{t^*})$ be a standard solution to $[\ms P].$ Then, $\tau_{t^*}$ solves the agent's problem given policy rule $y$:
    \[\tau_{t^*} \in \argmax_{\tau\in\ms T} \E[u(\t,y(\bfh),\tau)].\]}
\begin{proof}
    By \cref{thm:inner=aux}, we have that $y$ is a standard solution to $[\ms P(t^*)].$ We will consider separately the case where all constraints in $[\ms P(t^*)]$ are slack and where at least one binds.
    \medskip

    \underline{Case 1:} $U_{t^*}\geq U_{t'}$ binds, for some $t'\in \{t^*+1,\dots,\infty\}.$

    The result follows by \cref{lem:2_part_a}. Consider any stopping rule, $d$. Let $t$ be the (unique, if it exists) period at which $d(\bfa_t)=\ttstop,$ or $t=\infty$ if there is no such period. Let $\ms R:=\{(r,s):d(\bfb_{r,s})=\ttstop\}$ denote the set of other partial histories at which $d$ stops. Then, using formula \eqref{eq:u_formula}, we have:
    \[U_{\tau_d}=\sum_{(r,s)\in\ms R} p(\bfb_{r})u(0,y(\bfb_{r}),s)+U_t. \]
    By \cref{lem:2_part_a} we have that the first term is weakly negative, and therefore $U_{\tau_d}\leq U_t.$ But since $y$ is feasible for $[\ms P],$ we have $U_t\leq U_{t^*}.$ Thus, $U_{\tau_d}\leq U_{t^*},$ and so $\tau_{t^*}$ is optimal for the agent.
    \medskip

    \underline{Case 2:} $U_{t^*}\geq U_{t'}$ does not bind, for any $t'\in \{t^*+1,\dots,\infty\}.$

    Since no constraint binds, by \cref{thm:inner}(b), we find that $y(\bfb_{t'})$ must be constant for all $t' > t^*$; let this value be $y_0$, and let $y_1:=y(\bfa)$.
    We will examine two cases.

    \underline{Case 2.1:} $u(0,y_0,t')\leq 0,$ for all $t'>t^*.$

    In this case, the argument from \underline{Case 1} applies. Consider any stopping rule, $d$, and let $t$ be the period where $d(\bfa_t)=\ttstop$ (or $t=\infty$ if the agent never stops). Let $\ms R:=\{(r,s):d(\bfb_{r,s})=\ttstop\}$ denote the set of other partial histories at which $d$ stops. Then, we have:
    \[U_{\tau_d}=\sum_{(r,s)\in\ms R:r>t^*} p(\bfb_r)u(0,y_0,s)+\sum_{(r,s)\in\ms R:r\leq t^*}p(\bfb_{r})u(0,0,s)+U_t\]
    (where the second term arises since $y$ is a standard solution). Then, by the case assumption, the first term is weakly negative, while by \cref{ass:A-payoff-main}(f), the second term is weakly negative. Thus, as before, we obtain $U_{\tau_d}\leq U_{t}\leq U_{t^*},$ so that $\tau_{t^*}$ is optimal for the agent. 

    \underline{Case 2.2:} $u(0,y_0,t')> 0,$ for some $t'>t^*.$

    In this case, by positive monotonicity, $u(0,y_0,t^*)>0.$ Consider any stopping rule, $d$, and let $t$ be the period where $d(\bfa_t)=\ttstop$ (or $t=\infty$ if the agent never stops). We consider two cases, depending on the value of $t.$

    \underline{Case 2.2.1:} $t\leq t^*.$ 
    
    Let $\ms R:=\{(r,s):d(\bfb_{r,s})=\ttstop\}$ denote the set of other partial histories at which $d$ stops; note that for all $(r,s)\in\ms R$, we must have $r\leq t^*.$ Then, we have:
    
    \[U_{\tau_d}=\sum_{(r,s)\in\ms R}p(\bfb_{r})u(0,0,s)+U_t,\]

    where we are using the fact that $y$ is a standard solution. As before, we obtain $U_{\tau_d}\leq U_t\leq U_{t^*}.$

    \underline{Case 2.2.2:} $t > t^*.$

In this case, we claim that, under any history $\bfh\in\mathcal{H}$, we must have $u(\theta,y(\bfh),\tau_d)\leq u(\theta,y(\bfh),\tau_{t^*}).$ We divide histories into three types.

    \textit{Type a:} $\bfh=\bfb_s$, for some $s\leq t^*.$

    In this case, $\tau_{t^*}=\infty,$ so the agent's utility under $\tau_{t^*}$ is 0. If the agent instead stops at time $t'\neq \infty$ under $\tau_d$, he receives $u(0,0,t')<0$, since $y$ is a standard solution.

    \textit{Type b:} $\bfh=\bfb_s$, for some $s> t^*.$

    Under $\tau_{t^*}$, the agent receives $u(0,y_0,t^*)$, which is positive. By the case assumption ($t>t^*$), the agent cannot stop earlier than $t^*$, and so must receive $u(0,y_0,t')$, for some $t'\geq t^*.$ By positive monotonicity, this must be weakly lower than $u(0,y_0,t^*)$.

    \textit{Type c:} $\bfh=\bfa.$

    Under $\tau_{t^*}$, the agent receives $u(1,y_1,t^*)$, which we claim is positive. To see this, examine the FOCs for $y_1$ and $y_0$, \eqref{eq_foc_ya} and \eqref{eq_foc_yb}. We obtain:
    \begin{equation}
        \label{eq_2.2.2c1}
        v_y(1,y_1,t^*)\leq 0
    \end{equation}
    and $v_y(0,y_0,t^*)\leq0,$ where the latter inequality can be strict only if $y_0=0;$ since we need $y_0>0$ for $u(0,y_0,t^*)>0$ to hold, we conclude that:
    \begin{equation}
        \label{eq_2.2.2c2}
        v_y(0,y_0,t^*)=0.
    \end{equation}
    Now, comparing equations \eqref{eq_2.2.2c1} and \eqref{eq_2.2.2c2}, and using \cref{ass:P-payoff-main}(b) and (f) (strict concavity of $v$ in $y$ and monotonicity of $v_y$ in $\theta$), we conclude that $y_1\geq y_0.$ Finally, using \cref{ass:A-payoff-main}(a) and (e) (monotonicity of $u$ in $y$ and $\theta$), we conclude that, since $u(0,y_0,t^*)>0$ and $y_1\geq y_0$, we must have $u(1,y_1,t^*)>0.$
    
     Now, under $\tau_d$, the agent stops at $t$, and so receives $u(1,y_1,t)$, which must be weakly lower than $u(1,y_1,t^*)$ by positive monotonicity.

    Thus, we must have $U_{\tau_d}\leq U_{t^*}.$ Since we concluded this in both cases (2.2.1 and 2.2.2), this holds for all stopping rules, $d$, and so $\tau_{t^*}$ must be optimal for the agent.     
\end{proof}


\subsubsection{\texorpdfstring{Proof of \cref{lem:principal_wants_simple}}{Proof of Lemma \ref{lem:principal_wants_simple}}}\label{proof:lem:principal_wants_simple}



\noindent \textbf{\cref{lem:principal_wants_simple}}. 
\textit{There exists a policy $y^*$ and a simple stopping time $\tau_{t^*}$ such that $(y^*,\tau_{t^*})$ solves $[\ms P_0]$.}
\begin{proof}
Given any $(y,\tau)$ which is optimal for $[\ms P_0]$, and such that $\tau$ is not a simple stopping time, we will produce a $\tilde{y}$ where $(\tilde{y},\tau_0)$ is feasible for $[\ms P_0]$; we will then show that this $(\tilde{y},\tau_0)$ weakly improves on $(y,\tau)$. Since a solution to $[\ms P_0]$ is guaranteed to exist by the Weierstrass theorem,\footnote{The argument is analogous to that of the proof of \cref{lem:standard-solutions-exist}(a).} this will prove the lemma.
\medskip

    Let $d$ be the decision rule corresponding to $\tau$, and let $t$ be the period at which $d(\bfa_t)=\ttstop,$ or $t=\infty$ if there is no such period. Let $\ms R:=\{(r,s):d(\bfb_{r,s})=\ttstop\}$ denote the (non-empty, since $\tau$ is not a simple stopping time) set of other partial histories at which $d$ stops. Let $r\leq t,s\geq r$ be defined by:
    \[(r,s)\in\argmax_{(\hat{r},\hat{s})\in\ms R} v(0,y(\bfb_{\hat{r}}),\hat{s});\]
    i.e., $(r,s)$ gives the principal's favorite instance of the agent stopping after a breakdown under $(y,\tau).$ We show in \cref{OA:simple_claim} that, because $(y,\tau)$ solves $[\ms P_0],$ we must have $v(0,y(\bfb_r),s)\geq0$ and $V_t(y)\geq 0.$

    Now, define the policy rule $\tilde{y}$ as follows. If $t<\infty$, let $\tilde{y}(\bfb_i)=y(\bfb_r),$ for all $i\leq t$ and $\tilde{y}=y$ for all other histories. If $t=\infty$, let $\tilde{y}(\bfb_i)=\tilde{y}(\bfa)=y(\bfb_r)$, for all $i\in T$. As indicated above, we will argue that $(\tilde{y},\tau_0)$ improves on $(y,\tau).$

    First, we must show that $\tau_0$ is optimal for the agent against $\tilde{y},$ so that $(\tilde{y},\tau_0)$ is feasible for the principal. 
    \medskip

    \noindent \textit{Claim.} $\tau_0$ is optimal for the agent against $\tilde{y}.$
    \begin{proof}
        Consider any other stopping rule $\hat{d},$ and let $\hat{t}$ be the period at which $\hat{d}(\bfa_{\hat{t}})=\ttstop,$ or $\hat{t}=\infty$ if there is no such period. Let $\hat{\ms R}:=\{(\hat{r},\hat{s}):\hat{d}(\bfb_{\hat{r},\hat{s}})=\ttstop\}$ denote the set of other partial histories at which $\hat{d}$ stops. Let $\hat{\tau}$ be the corresponding stopping time. 

        It is without loss of optimality to assume that $U_{\hat{t}}(\tilde{y})\geq0$ (if not, $\hat{\tau}$ can be improved by the agent never stopping in any $\bfa_i$) and that $\hat{t}\leq t$ (since stopping at $\bfa_t$ inherits superiority over stopping at any $\bfa_i$, $i>t$, from the optimality of $d$ against $y$).
        
         Using \eqref{eq:u_formula} and \eqref{eq:A_payoff_contrib}, we may write
        \begin{align*}
            U_{0}(\tilde{y})&=\sum_{i=1}^{\hat{t}} p(\bfb_i)u(0,y(\bfb_r),0)+U_{0|\hat{t}}(\tilde y) \\
            U_{\hat{\tau}}(\tilde{y})&=\sum_{(\hat{r},\hat{s})\in \hat{\ms R}} p(\bfb_{\hat{r}})u(0,y(\bfb_r),\hat{s})+U_{\hat{t}}(\tilde{y}),
        \end{align*}
        where we are using the fact that $\hat{t}\leq t$ to deduce that, if there is a breakdown before $\hat{t},$ the policy will be $y(\bfb_r).$
        
        We argue that $U_0(\tilde y) \geq U_{\hat \tau}(\tilde y)$. First, let us match each value of $i$ in the top sum with the same value of $\hat{r}$ in the bottom sum, if it exists, and with 0 otherwise. By optimality of $\tau$ against $y$, we must have that $u(0,y(\bfb_r),s)\geq0$, and so by positive monotonicity, $u(0,y(\bfb_r),0)\geq0.$ Thus, when a term in the top sum is matched with 0, it is larger. Similarly, by positive monotonicity, we have that $u(0,y(\bfb_r),0)\geq u(0,y(\bfb_r),\hat{s})$, for any $\hat{s}$ (if $u(0,y(\bfb_r),\hat{s})<0$, then $u(0,y(\bfb_r),0)$ is clearly greater), and so when a term in the top sum is matched with a term in the bottom sum, it is also larger. 
        
        Moreover, it must be that $U_{0|\hat{t}}(\tilde y) \geq U_{\hat{t}}(\tilde{y})$. If $\hat{t}<\infty$, this follows by APM. If $\hat{t}=\infty$, then $U_{0|\hat{t}}=p(\bfa)u(1,y(\bfb_r),0)\geq 0= U_{\hat{t}}(\tilde{y})$, where the inequality follows because $u(0,y(\bfb_r),s)\geq 0$ since the agent was willing to stop at $\bfb_{r,s}$ under $y$, and using positive monotonicity and \cref{ass:A-payoff-main}(e). This proves the claim.
    \end{proof}

It remains to argue that the principal prefers $(\tilde y, \tau_0)$ to $(y,\tau)$. Using \eqref{eq:v_formula} and \eqref{eq:P_payoff_contrib}, we may write
    \begin{align*}
        V_{0}(\tilde y)&=\sum_{i=1}^t p(\bfb_i)v(0,y(\bfb_r),0)+V_{0|t}(y)\nonumber\\
        V_\tau(y)&=\sum_{(\hat{r},\hat{s})\in\ms R}p(\bfb_{\hat{r}})v(0,y(\bfb_{\hat{r}}),\hat{s})+V_t(y) + \sum_{(\hat r, \hat s) \notin \ms R \text{ for any } \hat s\geq \hat r} p(\bfb_{\hat r})v(0,y(\bfb_{\hat r}),\infty).\label{eq_princ_simp_v}        
    \end{align*}
First, let us compare the top sum with the two sums in the bottom, term by term. Because $v(0,y(\bfb_r),s)\geq 0$, positive monotonicity implies that $v(0,y(\bfb_r),0)\geq 0$, whereas $v(0,y(\bfb_{\hat r}),\infty) \leq 0$ by \cref{ass:P-payoff-main}. Additionally, because we chose $(r,s)$ to be the principal's favorite instance of stopping after a breakdown, and because of positive monotonicity, we must have $v(0,y(\bfb_r),0)\geq v(0,y(\bfb_r),s)\geq v(0,y(\bfb_{\hat{r}}),\hat{s})$ for any $(\hat{r},\hat{s}) \in \ms R$. Hence, the top sum is greater than the sum of the two sums in the bottom. 


Moreover, it must be that $V_{0|t}(y) \geq V_t(y)$. If $t<\infty,$ this follows from APM. If $t=\infty,$ then we have $V_{0|t}(y) = p(\bfa) v(1,y(\bfb_r),0) \geq p(\bfa)v(0,y(\bfb_r),0) \geq 0$, where we are using both statements in \cref{ass:P-payoff-main}(f) for the first inequality, and the second inequality has been observed earlier in this proof. In this case, since $V_{\infty}(y)=0$, we again conclude that $V_{0|t}(y) \geq V_t(y).$
Thus, $(\tilde{y},\tau_0)$ improves on $(y,\tau)$, proving the lemma. \end{proof}

\subsubsection{\texorpdfstring{Proof of \cref{thm:aux=original}}{Proof of Theorem \ref{thm:aux=original}}}\label{proof:thm:aux=original}

\cref{thm:aux=original} follows as a corollary to \cref{lem:p_sol_is_feas,lem:principal_wants_simple}.
\medskip

\noindent \textbf{\cref{thm:aux=original}}. 
\textit{The following statements are true:
\begin{enumerate}
    \item If $(y,\tau_t)$ is a standard solution to $[\ms P]$, then it solves $[\ms P_0].$ 
    \item If $(y,\tau_t)$ solves $[\ms P_0]$, then $(y,\tau_t)$ solves $[\ms P].$ 
\end{enumerate}}

\begin{proof}
To prove (a), suppose $(y,\tau_t)$ is a standard solution to $[\ms P].$ Then, by \cref{lem:p_sol_is_feas}, $(y,\tau_t)$ must be feasible for $[\ms P_0]$.

By \cref{lem:principal_wants_simple}, there must exist a policy and a simple stopping time, $(y^*,\tau_{t^*})$, which solves $[\ms P_0]$. Since $[\ms P]$ relaxes the agent's incentive constraints relative to $[\ms P_0],$ $(y^*,\tau_{t^*})$ must be feasible for $[\ms P]$. Thus, $V_{t}(y)\geq V_{t^*}(y^*)$, and so $(y,\tau_t)$ must solve $[\ms P_0].$
\medskip

To prove (b), note that since $\tau_t$ is a simple stopping time and $[\ms P]$ relaxes the agent's incentive constraints relative to $[\ms P_0]$, $(y,\tau_t)$ must be feasible for $[\ms P].$

    By \cref{lem:standard-solutions-exist}, there must exist a standard solution to $[\ms P];$ call it $(y^*,\tau_{t^*}).$ By \cref{lem:p_sol_is_feas}, $(y^*,\tau_{t^*})$ is feasible for $[\ms P_0].$ Thus, $V_{t}(y)\geq V_{t^*}(y^*)$, and so $(y,\tau_t)$ must solve $[\ms P].$
\end{proof}

\subsection{\texorpdfstring{Discussion of \cref{thm:aux=original} Assumptions}{Discussion of Theorem \ref{thm:aux=original} Assumptions}}\label{app:non-inclusion}

Let $\ms S_0$ and $\ms S$ denote the set of solutions to $[\ms P_0]$ and $[\ms P],$ respectively. We will sketch an argument that $\ms S_0 \not\subseteq \ms S$ and $\ms S \not\subseteq \ms S_0.$

 We have $\ms S_0 \not\subseteq \ms S$ because of potential indifferences, which APM does not rule out. For instance, consider a two-period example where neither the principal nor the agent care about the state; furthermore, neither cares about the agent's investment timing, though the principal substantially prefers investing to never investing. Though inducing the simple stopping time $\tau_0$ is optimal for the principal, it is also optimal to induce the non-simple stopping time in which the agent stops at time 1, no matter what. This is not a solution to $[\ms P]$ since it is non-simple.

Conversely, we also have that $\ms S \not\subseteq \ms S_0.$ The issue is that a non-standard solution to $[\ms P],$ $(y,\tau_t),$ could offer the agent a positive payoff if he stops following a breakdown at some period $s<t$. In $[\ms P],$ the only way for the agent to exploit this would be to stop following no breakdown, prior to $s$; this may not be desirable for him. But, in $[\ms P_0],$ he could stop following the breakdown at $s$ \textit{and} following no breakdown at $t$. 

\subsection{\texorpdfstring{Proofs of Results in \cref{sec:rent}}{Proofs of Results in Section \ref{sec:rent}}}\label{proof:rent-section}

In this subsection, we first state and prove \cref{prop-prime:null-is-lower-bound,prop-prime:no-rent,prop-prime:rent}, which are generalized versions of \cref{prop:null-is-lower-bound,prop:no-rent,prop:rent} that do not rely on \cref{ass:rent}. We then argue that under \cref{ass:rent}, \cref{prop-prime:null-is-lower-bound,prop-prime:no-rent,prop-prime:rent} imply \cref{prop:null-is-lower-bound,prop:no-rent,prop:rent}.

Without \cref{ass:rent}, both the null policy rule and the null payoff must be redefined to depend on the period $t$.
Specifically, the \textit{null policy rule for period $t \in T$} is\footnote{For $t=\infty$, we define  $y^{(\infty)}_{\text{null}}(\bfh) \:= 0$ for all $\bfh$.}
\begin{align*}
    \ynullt (\bfh) \:= \begin{cases}
        0 \caseif \bfh \in \{\bfb_1,\bfb_2,\ldots,\bfb_t\} \\
        \arg\max_y v(0,y,t) \caseif \bfh \in \{\bfb_{t+1},\bfb_{t+2},\ldots,\bfb_T\} \\
        \arg\max_y v(1,y,t) \caseif \bfh = \bfa.
    \end{cases} 
\end{align*}
The \textit{null payoff for period $t$} is 
$\Unullt \:= \max_{s} U_s (\ynullt).$
Note that the maximizing $s$ need not equal $t$.

The following result is a generalization of \cref{prop:null-is-lower-bound} and establishes a lower bound on the agent's equilibrium payoff.
\begin{customprop}{2$'$}\label{prop-prime:null-is-lower-bound}
    If $(y,\tau_t)$ solves $\original$, then $U_t(y) \geq \Unullt$.
\end{customprop}
\begin{proof}
    By \cref{thm:inner=aux,thm:aux=original}, $y$ must solve $\inner$. Then, since $w_y>0$, $v_{yy}< 0$, and $\phi$ is non-increasing, the first-order conditions \eqref{eq_foc_ya} and \eqref{eq_foc_yb} imply $y(\bfh)\geq \ynullt(\bfh)$ for all $\bfh \in \{\bfb_{t+1},\cdots,\bfb_T,\bfa\}$. For $\bfh \in \{\bfb_1,\cdots,\bfb_t\}$, we have $y(\bfh) \geq \ynullt(\bfh) = 0$. Since $u_y>0$, we must then have $U_t(y) \geq \Unullt$.    
\end{proof}

The following proposition generalizes \cref{prop:no-rent}. It states that the lower bound in \cref{prop-prime:null-is-lower-bound} is tight if $\phi$ is constant.

\begin{customprop}{3$'$}[No Rent Without Complementarity]
\label{prop-prime:no-rent}
    Suppose $\phi$ is constant. If $(y,\tau_t)$ solves $[\ms P_0]$, then $U_t(y) = \Unullt$.
\end{customprop}
\begin{proof}
We first prove the following lemma, which is an analogue of \cref{prop-prime:no-rent} for the inner problem $\inner$.


\begin{lem}\label{lem:utility-at-null-peak}
    Suppose $\phi$ is constant, and suppose $y$ is a solution to $[\ms P(t)]$, for some $t \in T \cup \{\infty\}$. Then, we have
    $U_t(y)=\max_{s \geq t} U_s (\ynullt).$
\end{lem}
\begin{proof}
    Let $\hat t$ denote the smallest solution to $\max_{s \geq t} U_s (\ynullt)$. We write the proof assuming $\hat t < \infty$, but the same proof goes through for $\hat t = \infty$ if one adopts the convention $T+1\:=\infty$ and $\bfb_{T+1} \:= \bfa$.
    
    We first argue that $y(\bfh)= \ynullt(\bfh)$ for all histories $\bfh \in \{\bfb_{\hat t+1},\bfb_{\hat t+2},\ldots, \bfb_T, \bfa\}$. We established in the proof of \cref{prop-prime:null-is-lower-bound} that $y(\bfh)\geq \ynullt(\bfh)$ for all $\bfh$, so suppose toward a contradiction that $y(\bfh) > \ynullt(\bfh)$ for some $\bfh \in \{\bfb_{\hat t+1},\bfb_{\hat t+2},\ldots, \bfb_T, \bfa\}$.  Define a new policy rule $\tilde y$ by $\tilde y(\bfh) = \ynullt(\bfh)$ for $\bfh \in \{\bfb_{\hat t+1},\bfb_{\hat t+2},\ldots, \bfb_T, \bfa\}$, and $\tilde y(\bfh) = y(\bfh)$ for all other $\bfh$. 
    We argue that $\tilde y$ still satisfies the constraints of $\inner$, namely, $U_{t}(\tilde y)\geq U_{s}(\tilde y)$ for all $s \geq t$.
    Because $\phi$ is constant, the payoff difference $U_t(\tilde y) - U_s(\tilde y)$ only depends on the policy components $\tilde y(\bfb_{t+1}), \tilde y(\bfb_{t+2}), \ldots, \tilde y(\bfb_{s})$.\footnote{This follows from equation \eqref{eq:payoff-dif}.} As long as $s \leq \hat t$, these components are unchanged from $y$ to $\tilde y$. Hence, the constraints for $s \leq \hat t$ are unaffected and still satisfied under $\tilde y$. For any $s > \hat t$, we have
    $$ U_t(\tilde y) \geq U_{\hat t}(\tilde y) =  U_{\hat t}(\ynullt) \geq U_{s}(\ynullt) = U_s(\tilde y), $$
    where the two equalities hold by the definition of $\tilde y$, and the second inequality holds by the definition of $\hat t$. Therefore, the constraints for $s > \hat t$ also remain satisfied.
    However, the principal's expected payoff is clearly higher under $\tilde y$ than under $y$.
    Thus, for all histories after $\hat t$, the policy must be null.

    We are now ready to show  $U_t(y) = U_{\hat t}(\ynullt)$. If $t=\hat t$, the equality follows from the result shown in the above paragraph, so let us assume $t < \hat t$. Note that we have 
    $$U_t(y)\geq U_{\hat t}(y)\geq U_{\hat t}(\ynullt),$$
    where the first inequality holds because $y$ is feasible under $\inner$, and the second inequality holds because $y(\bfh)\geq \ynullt(\bfh)$ for all $\bfh$.
    Towards a contradiction, suppose $U_t(y)>U_{\hat t}(\ynullt)$.

    Let $t'$ be the largest integer such that $t \leq t' < \hat t$ and $U_{t'}(y)>U_{\hat t}(\ynullt)$. Clearly, $t'$ exists. Also, we must have $y(\bfb_{t'+1})> \ynullt (\bfb_{t'+1}) = \argmax_y v(0,y,t) $, as otherwise \cref{thm:inner}(a) and (b) would imply $y(\bfh) = \ynullt(\bfh)$ for all $\bfh \in \{\bfb_{t'+1},\bfb_{t'+2},\ldots, \bfb_T, \bfa\}$, 
    contradicting $U_{t'}(y) > U_{\hat t}(\ynullt)$.
    Define a new policy rule $y'$ by $y'(\bfb_{t'+1}) = y(\bfb_{t'+1}) - \epsilon$ and $y'(\bfh) = y(\bfh)$ for all $\bfh \neq \bfb_{t'+1}$, where $\epsilon>0$ is small enough that $U_{t}(y')>U_{\hat t}(\ynullt)$ still holds. 
    We argue that $y'$ satisfies all of the constraints, $U_{t}(y')\geq U_{s}(y')$ for $s \geq t$.
    First, the constraints for $t \leq s \leq t'$ are unaffected because $\phi$ is constant.
    Next, the constraints for $t' < s < \hat t$ (if any) are still satisfied because, for these values of $s$, we have 
    $$ U_t(y') > U_{\hat t}(\ynullt) \geq U_{s}(y) = U_s(y'), $$ 
    where the second inequality follows from the definition of $t'$. 
    Finally, for $s \geq \hat t$, we have 
    $$ U_{t}(y')>U_{\hat t}(\ynullt) \geq U_s(\ynullt) = U_s(y'), $$
    where the equality holds because $y'(\bfh) = y(\bfh) = \ynullt(\bfh)$ for $\bfh \in \{\bfb_{\hat t+1},\bfb_{\hat t+2},\ldots, \bfb_T, \bfa$\}.
    Hence, $y'$ still satisfies all constraints. However, the principal's expected payoff is higher under $y'$ than under $y$. This contradicts the optimality of $y$ for $\inner$. Thus, we must have $U_t(y) = U_{\hat t}(\ynullt)$.
\end{proof}


The proof of \cref{prop-prime:no-rent} is now straightforward. 
Since $(y,\tau_t)$ solves $\original$, \cref{thm:aux=original}(b) implies that $(y,\tau_t)$ also solves $\auxiliary$. Then, by \cref{thm:inner=aux}(a), $y$ solves $\inner$, so we have $U_t(y)=\max_{s \geq t} U_s (\ynullt)$ by \cref{lem:utility-at-null-peak}. We may thus write
$$U_t(y) \geq \max_{s} U_s (y) \geq \max_{s} U_s (\ynullt) \geq \max_{s \geq t} U_s (\ynullt) = U_t(y),$$
where the first inequality comes from the feasibility of $(y,\tau_t)$ under $[\ms P_0]$, and the second inequality from the fact that $y(\bfh) \geq \ynullt(\bfh)$ for all $\bfh$. Therefore, all of the inequalities hold as equalities, and we have $U_t(y)=\max_{s} U_s (\ynullt)$. 
\end{proof}


The following proposition generalizes \cref{prop:rent}, providing a sufficient condition for the agent to obtain rent. 
\begin{customprop}{4$'$}[Rent from Complementarity]\label{prop-prime:rent}
Let $(y,\tau_t)$ be a standard solution to  $\original$.
Suppose the following statements hold:
\begin{enumerate}[(i)]
    \item $\phi$ is strictly decreasing.
    \item If $\hat t \in \argmax_s U_s (\ynullt)$, then $\hat t > t$.
    \item  There exists $\hat t \in \argmax_s U_s(\ynullt)$ such that $\hat t \leq T$, and for this $\hat t$, there exists $\bfh \in \{\bfb_{\hat t+1},\bfb_{\hat t+2},\ldots,\bfb_T,\bfa\}$ such that $y(\bfh)>0$.
\end{enumerate}
Then, $U_t(y) > \Unullt$. 
\end{customprop}
\begin{proof}
As argued in the proof of \cref{prop-prime:null-is-lower-bound}, it must be that $y$ solves $\inner$, and $y \geq \ynull^{(t)}$ pointwise. 
Note that condition (ii) implies $y \neq \ynullt$.
Take $\hat t$ and $\bfh$ satisfying condition (iii). By condition
(ii), $\hat t > t$.

\medskip

\noindent \textit{Claim.} $y(\bfh) > \ynullt(\bfh)$.
\begin{proof}
We will prove the claim assuming $\bfh = \bfb_{t'}$ for some $t' \in \{\hat t+1, \hat t+2, \ldots, T\}$; an analogous argument can be made for the remaining case of $\bfh = \bfa$. Consider the Lagrangian for $\inner$. Since $y(\bfb_{t'}) > 0$ , the FOC of the Lagrangian with respect to $y(\bfb_{t'})$ must hold with equality and may be written as
 \begin{equation}
    \label{eq:prop4_1}
    v_y(0,y(\bfb_{t'}),t) = -\left( \sum_{r=1}^{t'-t-1} \lambda_r\left[\phi(t) - \phi(t+r) \right]+\sum_{r=t'-t}^{T-t}\lambda_r\phi(t) + \lambda_\infty \phi(t) \right) w_y(0,y(\bfb_{t'})).
    \end{equation}
We claim that there exists $r  \in \{1,2, \ldots, T-t,\infty\}$  such that $\l_r>0$. If instead $\l_r = 0$ for all $r$, then since $y$ is standard, it must be that $y = \ynullt$, contradicting condition (ii). 
Therefore, at least one $\l_r$ is strictly positive.

Then, the condition that $\phi$ is strictly decreasing, together with \cref{ass:A-payoff-main}(a) and (b), implies that the RHS of \eqref{eq:prop4_1} is strictly negative.
Since $v$ is concave in $y$, and $v_y(0,\ynullt(\bfh),t) = 0$ as long as $\ynullt(\bfh) >0$,\footnote{If $\ynullt(\bfh) = 0$, the claim holds trivially because $y(\bfh)>0$ by (iii).} it must be that $y(\bfh) > \ynullt(\bfh)$.
\end{proof}
We then have
\begin{align*}
    U_t (y) \geq U_{\hat t} (y) > U_{\hat t}(\ynullt) = \max_s U_{s}(\ynull^{(t)}) = \Unullt.
\end{align*}
The first inequality holds because $(y,\tau_t)$ is feasible under $\original$. The second inequality holds because $y \geq \ynull^{(t)}$ pointwise, and because $y(\bfh) > \ynull^{(t)}(\bfh)$ and $\bfh \in \{\bfb_{\hat t+1},\bfb_{\hat t+2},\ldots,\bfb_T,\bfa\}$ by the claim above. We have thus proved $U_t(y) > \Unullt$, as desired.\footnote{As an aside, we can also show at this point that the multipliers used to prove the claim satisfy $\l_r = 0$ for $r \in \{\hat t-t+1,\ldots,T-t,\infty\}$. To see why, note that $\l_r>0$ implies $U_t(y)= U_{t+r}(y)$ by complementary slackness. However, we must have $U_{\hat t}(y) > U_{t+r}(y)$ because $U_{\hat t}(\ynullt) \geq U_{t+r}(\ynullt)$, $\phi$ is strictly decreasing, and $y(\bfh)>\ynullt(\bfh)$ for some $\bfh \in \{\bfb_{\hat t+1}, \ldots, \bfb_{T},\bfa\}$ by the claim above. Since $U_t(y) \geq U_{\hat t}(y)$, we have a contradiction.
}    
\end{proof}

Finally, we argue that when \cref{ass:rent} is satisfied, \cref{prop-prime:null-is-lower-bound,prop-prime:no-rent,prop-prime:rent} can be simplified to \cref{prop:null-is-lower-bound,prop:no-rent,prop:rent}. The argument is a direct consequence of the following lemma, which states that the weak version of condition (ii) in \cref{prop:rent} always holds under \cref{ass:rent}.

\begin{lem}\label{lem:0-null-enough}
    Suppose \cref{ass:rent} holds. If $(y,\tau_t)$ is a standard solution to $\original$, and  $\hat t \in \argmax_s U_s(\ynull)$, then $\hat t \geq t$.
\end{lem}
\begin{proof}
    Suppose towards a contradiction that $\hat t < t$. Then, the principal's expected payoffs satisfy
    \begin{align*}
        V_{\hat t} (\ynull^{(\hat t)}) > V_t(\ynullt) \geq V_t(y).
    \end{align*}
    The first inequality follows from an argument analogous to that in the proof of \cref{thm:inner=aux}(b); the inequality is strict because $v(0,\uy(0),\hat t)>0$ by \cref{ass:rent}. Moreover, since $\tau_{\hat t}$ is a best response for the agent under $\ynull$, it remains a best response under $\ynull^{(\hat t)}$. Therefore, $(\ynull^{(\hat t)}, \tau_{\hat t})$ is a feasible solution to $\original$ and gives the principal a strictly higher payoff than does $(y,\tau_t)$, a contradiction.
\end{proof}
Under \cref{ass:rent}, $\ynullt$ and $\ynull$ coincide on histories $\{\bfb_{t+1},\ldots,\bfa\}$, so \cref{lem:0-null-enough} implies $\argmax_s U_s (\ynull) = \argmax_s U_s (\ynullt)$ and $\Unull = \Unullt$. \cref{prop:null-is-lower-bound,prop:no-rent,prop:rent} follow immediately.

\subsection{\texorpdfstring{Proof of \Cref{prop:state-msrble}}{Proof of Proposition \ref{prop:state-msrble}}}\label{proof: sufficient condition for OTA}

Towards a contradiction, suppose $U_t=U_{t+s}$ for some $s \in \{1,2,\ldots, T-t-1\}$. We then have
\begin{align*}
    0 \geq U_{t+s-1} - U_t = U_{t+s-1} - U_{t+s} = \left(\prod_{i=1}^{t+s-1} p_i\right) (1-p_{t+s})\left[ w(0,y(\bfb_{t+s}))-A(t+s)\right].
\end{align*}
Since $p_i \in (0,1)$, it must be that $w(0,y(\bfb_{t+s}))-A(t+s) \leq 0$.
We then have
\begin{align}\label{eq:prop6-1}
    w(0,y(\bfb_{t+s+1}))\leq w(0,y(\bfb_{t+s}))\leq A(t+s)<A(t+s+1),
\end{align}
where the first inequality holds because $y(\bfb_{t+s})\geq y(\bfb_{t+s+1})$ (by \Cref{thm:inner}), and the last inequality follows because $A$ is strictly increasing. However, \eqref{eq:prop6-1} implies $U_{t+s+1} > U_{t+s} = U_t$, which violates the agent's incentive constraint. 

Then, $y(\bfb_{t+1}) = y(\bfb_{t+2}) = \cdots = y(\bfb_{T})$ follows from \cref{thm:inner}(b). \qed

\subsection{\texorpdfstring{Proof of \Cref{prop: sufficient condition for nestedness}}{Proof of Proposition \ref{prop: sufficient condition for nestedness}}}\label{proof: sufficient condition for nestedness}

\paragraph*{Proof of (a).} For any $s\in\{1,2,\ldots, T\}$ we have
\[ U_{s}(\ynull)-U_{s-1}(\ynull)=\left(\prod_{i=1}^{s-1}p_i\right) (1-p_s)[A(s)-w(0,\uy(0))],\]
so the sign of $U_{s}(\ynull)-U_{s-1}(\ynull)$ is the same as the sign of $A(s) - w(0,\uy(0))$. Since $A(s)$ is strictly decreasing, $U_s(\ynull)$ must be strictly quasi-concave in $s$ for $s \in \{1,2,\ldots,T\}$. Moreover, we have
\[U_T(\ynull) - U_\infty(\ynull)= \left(\prod_{i=1}^{T}p_i\right) [w(1,\uy(1))+g(1,T)]>0,\]
where the inequality follows from condition (iv). Therefore, $U_s(\ynull)$ is strictly quasi-concave on the entire set $\{0,1,\ldots,T,\infty\}$. Condition (v) rules out the possibility of two adjacent maximizers, so $U_s(\ynull)$ is uniquely maximized at some $s = \hat t < \infty$.

\paragraph*{Proof of (b).}

Towards a contradiction, suppose there exists a positive integer $s \leq \hat t - t$ such that the constraint $U_t(y)\geq U_{t+s}(y)$ holds with a strict inequality, and let $s$ be the smallest such positive integer.
Let $q \leq \hat t- t - s$ be the smallest positive integer such that the constraint $U_t(y)\geq U_{t+s+q}(y)$ binds. Such a $q$ exists because we have $U_t(y) = U_{\hat t}(y) = U_{\hat t}(\ynullt)$ from the proof of \cref{prop-prime:no-rent}.

We have
    \begin{align*}
        U_t(y)- U_{t+s}(y)
        = &\, \sum_{r=1}^s
        \left(U_{t+r-1}(y)- U_{t+r}(y)\right) \\
        = &\, \sum_{r=1}^s \left( \prod_{i=1}^{t+r-1}p_i \right) (1-p_{t+r})[w(0,y(\bfb_{t+r}))-A(t+r)] > 0,
    \end{align*}
where the second equality uses equation \eqref{eq:payoff-dif}.
By the definition of $s$, it must be that $U_{t+r-1}(y)- U_{t+r}(y)=0$ for all $r<s$, so we have $w(0,y(\bfb_{t+s}))>A(t+s)$. 
Similarly, by expanding $U_t(y)- U_{t+s+q}(y)$ and then using the definition of $q$, we may conclude $w(0,y(\bfb_{t+s+q}))<A(t+s+q)$. Then, since $A$ is strictly decreasing, we get
    \[w(0,y(\bfb_{t+s}))>A(t+s)>A(t+s+q)>w(0,y(\bfb_{t+s+q})).\]
Therefore, we must have $y(\bfb_{t+s})>y(\bfb_{t+s+q})$. 

However, since $U_t(y)>\max\{U_{t+s}(y),\ldots, U_{t+s+q-1}(y)\}$, \Cref{thm:inner} implies
    \[y(\bfb_{t+s})=y(\bfb_{t+s+1})=y(\bfb_{t+s+2})=\ldots = y(\bfb_{t+s+q}),\]
which is a contradiction.

\paragraph*{Proof of (c).}

Consider any solution $y^{(t)}$ to $\inner$. From condition (iii) and the proof of \cref{prop-prime:no-rent}, we know that $y^{(t)}$ must satisfy 
\begin{equation*}
    \begin{cases}
        y^{(t)}(\bfb_{s})=\uy(0) \quad \forall s\geq \max\{\hat t+1, t+1\}\\
        y^{(t)}(\bfa)=\uy(1).
    \end{cases}
\end{equation*}
Also, if $t< \hat t$, part (b) and equation \eqref{eq:payoff-dif} together imply that $w(0,y^{(t)}(\bfb_s)) = A(s)$ for all $s \in \{t+1, t+2, \ldots, \hat t\}$. This uniquely pins down $y^{(t)}(\bfb_s)$ independently of $t$, since $w_y>0$ by \cref{ass:A-payoff-main}(a). Hence, the solution to $\inner$ is unique for each $t$, and these solutions must be nested.


\subsection{\texorpdfstring{Proof of \Cref{prop: nestedness implies delays}}{Proof of Proposition \ref{prop: nestedness implies delays}}}\label{proof: nestedness implies delays}

Suppose towards a contradiction that $(y,\tau_{t})$ solves $\original$, but there does not exist a solution $(y^{FB},\tau_{t^{FB}})$ to $[\ms P^{FB}]$ with $t^{FB} \leq t$. Then, there must exist a solution $(\hat y,\tau_{\hat t})$ to $[\ms P^{FB}]$ such that $\hat t > t$.\footnote{This is because there must exist a policy rule and a simple stopping time which solves $[\ms P^{FB}]$. The existence follows from \cref{lem:fb-feas,thm:fb-equivalence} in \cref{OA:1st-best}.} 

By condition (iii), $y^{(t)}$ solves $\inner$, and $y^{(\hat t)}$ solves $[\ms P(\hat t)]$. 
Since $(y,\tau_t)$ solves $\original$, by \cref{thm:aux=original}(b) and \cref{thm:inner=aux}(a), we have $W(t) \geq W(\hat t)$, that is,
\begin{align}\label{eq:prop7-1}
    V_t(y^{(t)}) \geq V_{\hat t}(y^{(\hat t)}).
\end{align}
Using condition (ii) and the fact that $y^{(t)}$ nests $y^{(\hat t)}$, we may write
\begin{align*}
    V_t(y^{(t)}) &\, = K_t + p(\bfa) c(1,y^{(t)}(\bfa))  + \sum_{s=t+1}^{T} p(\bfb_s) c(0,y^{(t)}(\bfb_s)), \\
    V_t(y^{(t)}) - V_{\hat t}(y^{(\hat t)}) &\, = \begin{dcases}
        K_t - K_{\hat t}  + \sum_{s=t+1}^{\hat t} p(\bfb_s) c(0,y^{(t)}(\bfb_s)) \caseif  \hat t \leq T \\
         V_t(y^{(t)}) \caseif \hat t = \infty,
    \end{dcases} 
\end{align*}
where 
\begin{align*}
    K_s \:= \begin{cases}
        p(\bfa)f(1,s) + (p(\bfa_s) - p(\bfa))f(0,s) \caseif s \leq T \\
        0 \caseif s = \infty.
    \end{cases}
\end{align*}
Importantly, the term $K_s$ does not depend on the policy rule, so the difference in the principal's expected payoff, $V_t(y^{(t)}) - V_{\hat t}(y^{(\hat t)})$, only depends on the policy rule via the components $y^{(t)}(\bfh)$ for $\bfh \in \{\bfb_{t+1},\ldots, \bfb_{\hat t}\}$.

Define the policy rule $\tilde y$ by\footnote{If $\hat t = \infty$, then we let $\tilde y(\bfh) = y^{(t)}(\bfh)$ for all $\bfh$.}
\begin{equation*}
    \tilde y(\bfh)=
        \begin{cases}
            0 & \text{for } \bfh \in \{ \bfb_1,\bfb_2,\ldots,\bfb_{t}\} \\
            y^{(t)}(\bfh) & \text{for } \bfh \in \{ \bfb_{t+1}, \bfb_{t+2}, \ldots, \bfb_{\hat t}\} \\
            \hat y(\bfh) & \text{for } \bfh \in \{ \bfb_{\hat t+1}, \bfb_{\hat t+2}, \ldots, \bfb_{T}, \bfa\}.
        \end{cases}
\end{equation*}
We claim that $(\tilde y,\tau_t)$ solves $[\ms P^{FB}]$. 
Since $\tilde y$ coincides with $y^{(t)}$ for histories in $\{\bfb_{t+1},\ldots, \bfb_{\hat t}\}$, and coincides with $\hat y$ for the remaining histories, the difference in the principal's expected payoff must satisfy
\begin{align*}
    V_t(\tilde y) - V_{\hat t}(\hat y)  = V_t(y^{(t)}) - V_{\hat t}(y^{(\hat t)}) \geq 0,
\end{align*}
where the inequality follows from \eqref{eq:prop7-1}.
That is, the principal is no worse off under $(\tilde y,\tau_t)$ than under $(\hat y, \tau_{\hat t})$.\footnote{Since $(\hat y,\tau_{\hat t})$ solves $\auxiliaryfb$, the principal's expected payoff under $(\hat y,\tau_{\hat t})$ must be $V_{\hat t}(\hat y)$, even if $(\hat y,\tau_{\hat t})$ is not standard.}

It remains to show that the agent is willing to play $\tau_t$ under $\tilde y$. However, we have
$$U_t(\tilde y) - U_{\hat t}(\hat y) = U_t(y^{(t)}) - U_{\hat t}(y^{(\hat t)}) \geq 0,$$
where the equality follows from an argument analogous to the one above, since $\phi$ is constant so that the agent's payoff is of the form $w(\t,y) + g(\t,t)$; and the inequality holds because $y^{(t)}$ solves $\inner$ and nests $y^{(\hat t)}$. Since $(\hat y, \tau_{\hat t})$ solves $[\ms P^{FB}]$, we have $U_{\hat t}(\hat y) \geq \max_s U_s(0)$, so we must have $U_t(\tilde y) \geq \max_s U_s(0)$.
Therefore, $(\tilde y,\tau_t)$ solves $[\ms P^{FB}]$, which contradicts the assumption made in the beginning of the proof.\qed

\newpage

\section*{\centering  \fontsize{20pt}{24pt}\selectfont Online Appendix}
\addcontentsline{toc}{section}{Online Appendix}
\label[onlineapp]{sec:online-appendix}

\bigskip
\bigskip

\section{Motivating Example: Optimal Investment Timing}\label[onlineapp]{OA:motivating-example}

In discussing our motivating example (\cref{sec:motivating-example}), we only solved for the cost-minimizing policy rule that induces the agent to invest in period 0. This section considers the optimal investment timing -- should the principal induce the agent to invest early, or have him wait for more information? To answer this question, we must specify the principal's benefit from investment. Suppose the principal receives a benefit $V>0$ if the agent invests, regardless of the timing of the investment. For example, the principal may be a government that only cares that the investment occurs before the next election.

Then, the principal's expected payoff from inducing the agent to invest in period 0 is
\begin{align}
    &\, V - [p y^{SB}(1) + 2(1-p) y^{SB}(2)] \cr
    = &\, V - (2-p)I +(4-p)b, \label{eq:example-early}
\end{align}
whereas the principal's expected payoff from inducing the agent to wait until period 1 and invest if and only if state $c=1$ is
\begin{align}
    p(V-I+b). \label{eq:example-late}
\end{align}
The principal induces the agent to invest in period 0 whenever \eqref{eq:example-early} is positive and greater than \eqref{eq:example-late}, which is the case if, for example, the benefit $V$ is large enough.\footnote{We may ignore other strategies for the principal. Inducing investment if and only if $c=2$ is clearly dominated by inducing investment if and only if $c=1$. Also, it is impossible to induce the agent to invest in period 1 regardless of the state, since the agent would rather invest in period 0 to obtain an extra flow benefit $b$.}

Here, the principal has no intrinsic preference for \textit{early} investment, as she obtains the same benefit $V$ regardless of when the agent invests. Nevertheless, the principal provides costly certainty in order to induce early investment. This is because investing in period 0 means investment happens with probability 1, while if the agent waits, he will only invest if $c=1$.


Let us compare the optimal investment timing under first- and second-best. Under first-best, the principal's expected payoff from inducing the agent to invest in period 0 is $V - p y^{FB}(1) = V - I + 2b$. This is strictly greater than the corresponding second-best payoff, \eqref{eq:example-early}. On the other hand, her first-best expected payoff from inducing the agent to wait and invest only in the good state is $p(V - I + b)$, which equals the corresponding second-best payoff \eqref{eq:example-late}. Therefore, if the principal induces late investment under first-best, she also induces late investment under second-best; however, even if the principal induces early investment under first-best, she may induce late investment under second-best. Hence, in this example, moral hazard weakly delays investment, and sometimes strictly so. 

\section{Stronger Monitoring Technology}\label[onlineapp]{OA:monitoring}

This section relates to \cref{rem:monitoring}. We will first show that the principal does not gain from being able to condition the policy rule on whether or not the agent invested (\cref{prop:moni-1}). We will then argue that, once the principal has this stronger monitoring power, we may relax the assumption that the agent only benefits from policy if he invests (\cref{prop:moni-2}).

Let $i \in \{0,1\}$ indicate whether or not the agent invested, where $i=1$ means the agent invested at some $t\in T$, and $i=0$ if the agent never invested. Let $\tty(\bfh,i)$ denote a policy rule that conditions on both the history $\bfh \in \ms H$ and the indicator $i$. We first show that the principal does not benefit from conditioning policy on $i$.

\begin{prop}\label{prop:moni-1}
    Suppose the principal can commit to a policy rule $\tty(\bfh,i)$. There exists an optimal policy rule which satisfies $
    \tty(\bfh,0) = \tty(\bfh,1)$ for all $\bfh$.
\end{prop}
\begin{proof}
    Take any feasible solution $(\tty,\tau)$. We will produce another feasible solution $(\tilde \tty,\tau)$ such that $\tilde \tty$ is constant in $i$, and the principal's expected payoff is weakly higher under $(\tilde \tty,\tau)$ than under $(\tty,\tau)$. The result follows because there exists a solution to $\original$.

    We define $\tilde \tty$ as follows. For $\bfh$ such that $\tau(\bfh) < \infty$ (that is, for histories under which the agent is supposed to have stopped), set $\tilde \tty(\bfh,0) = \tilde \tty(\bfh,1) = \tty(\bfh,1)$. 
    For $\bfh$ such that $\tau(\bfh) = \infty$, set $\tilde \tty(\bfh,0) = \tilde \tty(\bfh,1) = 0$.

    Clearly, the principal's equilibrium payoff is weakly higher under $\tilde \tty$ than under $\tty$. The agent's equilibrium payoff is the same under $\tty$ and $\tilde \tty$. Moreover, the agent's deviation payoffs are weakly lower under $\tilde \tty$, since if he doesn't stop when he should, he still gets 0, whereas if he stops when he should not, he now gets the zero policy. Therefore, $(\tilde \tty, \tau)$ is feasible and gives a weakly higher expected payoff to the principal. 
\end{proof} 

Next, consider the following relaxation of \cref{ass:A-payoff-main}(c):
\[
    \text{Assumption 1(c$'$):  For all $\theta$ and $y$,  $u(\theta,y,\infty) \geq u(\theta,0,\infty)=0$.}
\]
Let us modify the principal's program $\original$ in two ways -- the principal can commit to $\tty(\bfh,i)$, and \cref{ass:A-payoff-main}(c) is replaced with Assumption 1(c$'$). Let $\modori$ denote such a modified program. 
We will make use of the following condition:
\begin{align}\label{eq:moni-1}
        \begin{cases}
            \tty(\bfh,0) = 0 & \text{for all } \bfh\\
            \tty(\bfh,1) = 0 \caseif \tau(\bfh) = \infty.
        \end{cases}
\end{align}
We have the following lemma.
\begin{lem}\label{lem:moni-1}
    In $\modori$, it is without loss of optimality to restrict attention to solutions satisfying \eqref{eq:moni-1}.
\end{lem}
\begin{proof}
    Let us identify $\tau$ with the pair $(t,\ms R)$, where $t$ is the period at which $d(\bfa_t)=\ttstop$ (or $t=\infty$ if there is no such period), and $\ms R:=\{(r,s):d(\bfb_{r,s})=\ttstop\}$.
    Using an argument similar to that in \cref{OA:simple_claim}, one can show that it is without loss of optimality to only consider solutions $(\tty,\tau)$ such that the principal always obtains a non-negative payoff from the continuation histories after the agent stops -- that is, $v(0,\tty(\bfb_r,1),s) \geq 0$ for all $(r,s) \in \ms R$, and $V_t(\tty) \:= \sum_{k=t+1}^T p(\bfb_k)v(0,\tty(\bfb_k,1),t)+p(\bfa)v(1,\tty(\bfa,1),t) \geq 0$.\footnote{If the principal ever obtains a negative payoff, then consider $\hat \tty$ which is equal to $\tty$ except at the following histories: (1) for $\bfh \in \{\bfb_1,\ldots,\bfb_t\}$, 
    set $\hat\tty(\bfh,0) = \hat\tty(\bfh,1) = 0$ if $\tau(\bfh) = \infty$, or if $\tau(\bfh) < \infty$ and $v(0,\tty(\bfh,1),s) < 0$ for $s = \tau(\bfh)$; (2) for $\bfh \in \{\bfb_{t+1},\ldots,\bfb_T,\bfa\}$,
    set $\hat\tty(\bfh,0) = \hat\tty(\bfh,1) = 0$ if $V_t(\tty) < 0$. As in \cref{OA:simple_claim}, one can verify that the principal strictly prefers $\hat \tty$ over $\tty$.}
    Henceforth in this proof, we restrict attention to such solutions.

    Take any feasible solution $(\tty,\tau)$. Let $\tilde \tty$ be the policy rule which satisfies \eqref{eq:moni-1} and otherwise coincides with $\tty$. We will show that there exists a feasible solution $(\tilde \tty, \tilde \tau)$ which gives the principal a weakly higher expected payoff than $(\tty,\tau)$. We are done if $(\tilde \tty,\tau)$ is feasible, since the principal weakly prefers $(\tilde \tty,\tau)$ over $(\tty,\tau)$. So, suppose $(\tilde \tty,\tau)$ is not feasible. This means there exists some $\tilde \tau$ such that, under $\tilde \tty$, $\tilde \tau$ is the agent's best response and gives him a strictly higher expected payoff than $\tau$. 

    We claim that $\tilde\tau$ must have the agent stop at some $\bfa_k$ with $k<t$. This is because, under $\tilde \tty$, the agent cannot profitably deviate from $\tau$ at any other partial history. First, at a partial history $\bfb_{r,r}$ such that $(r,s) \in \ms R$ for some $s\geq r$, the agent should continue to follow $\tau$ and stop at $\bfb_{r,s}$, since the change from $\tty$ to $\tilde \tty$ does not affect his payoff from stopping at $\bfb_{r,s}$ and can only decrease his payoff from not stopping. At a partial history $\bfb_{r,r}$ such that $(r,s) \notin \ms R$ for any $s\geq r$, the agent should again follow $\tau$ and never stop, obtaining 0, since stopping will give him a negative payoff under $\tilde \tty$. At $\bfa_t$, the agent will continue to stop, since only his deviation payoffs have been decreased under $\tilde \tty$. Hence, it must be that $\tilde\tau$ has the agent stop at some $\bfa_k$ with $k<t$.
    
    We now argue that the principal weakly prefers $(\tilde \tty,\tilde\tau)$ over $(\tty,\tau).$ Since the principal weakly prefers $(\tilde \tty,\tau)$ over $(\tty,\tau)$, it is enough to show that she weakly prefers $(\tilde \tty,\tilde\tau)$ over $(\tilde \tty,\tau)$. For histories $\bfb_r$ such that there is no $s \geq r$ with $(r,s) \in \ms R$, the principal is weakly better off under $(\tilde \tty,\tilde\tau)$ by \cref{ass:P-payoff-main}(a) and (d). For histories $\bfb_r$ such that $(r,s) \in \ms R$ for some $s \geq r$, and for histories $\{\bfb_{t+1},\ldots,\bfb_T,\bfa\}$, the principal is weakly better off under $(\tilde \tty,\tilde\tau)$, due to APM and because we are assuming, without loss of optimality, that the principal obtains a non-negative expected payoff conditional on these histories under $(\tty,\tau)$ and thus also under $(\tilde\tty,\tau)$. This proves the lemma.
\end{proof}

Now, fix a payoff function $u(\t,y,t)$ which satisfies \cref{ass:A-payoff-main}, and let $\ttu(\t,y,t)$ be a payoff function that satisfies all the statements in \cref{ass:A-payoff-main} except possibly for (c), satisfies Assumption 1(c$'$), and coincides with $u$ except possibly when $y >0$ and $t=\infty$. We will compare the original program $\original$ with the modified program $\modori$, where $\modori$ has the same payoff functions and parameters as does $\original$, except that $u$ is replaced with $\ttu$ in $\modori$.

Define 
\begin{align*}
    K \:= &\, \{(\tty,\tau) \mid (\tty,\tau) \text{ satisfies } \eqref{eq:moni-1}\}\\
    L \:= &\, \{(y,\tau) \mid y(\bfh) = 0 \text{ if } \tau(\bfh) = \infty\}.
\end{align*}
We can define a bijection $\Psi: K \to L$ such that each $(\tty,\tau)$ is mapped to $\Psi(\tty,\tau) = (y,\tau)$, where $y$ is given by $y(\bfh) = \tty(\bfh,1)$ for all $\bfh$.

If $\Psi(\tty,\tau) = (y,\tau)$, then for every $\bfh$ and for every $t \in \Tui$, we have $\mathtt{u}(\t,\tty(\bfh,i),t) = u(\t,y(\bfh),t)$, where $\t$ is determined by $\bfh$ and $i$ by $t$. That is, the agent's ex post payoff $\ttu$ under $\tty$ always coincides with his ex post payoff $u$ under $y$, regardless of whether the agent follows $\tau$ or not. To see this, note that $\tty(\bfh,1) = y(\bfh)$, so the payoff equality holds whenever $t <\infty$. Moreover, if $t = \infty$, then $\tty(\bfh,0) = 0$, so the equality again holds because $\ttu(\t,0,\infty) = u(\t,y(\bfh),\infty) = 0$.

Next, if $\Psi(\tty,\tau) = (y,\tau)$, then for every $\bfh$, we have $v(\t,\tty(\bfh,i),\tau(\bfh)) = v(\t,y(\bfh),\tau(\bfh))$, where $\t$ is determined by $\bfh$, and $i$ by $\tau(\bfh)$. That is, the principal's \textit{on-path} ex post payoff always coincides under $(\tty,\tau)$ and $(y,\tau)$. To see this, note that because $\tty(\bfh,1) = y(\bfh)$, the equality holds whenever $\tau(\bfh) < \infty$. Moreover, if $\tau(\bfh) = \infty$, then $\tty(\bfh,0) = y(\bfh) = 0$, so the equality again holds.\footnote{The principal's payoff does not necessarily coincide off-path, since if $\tau(h) < \infty$ but $t = \infty$, then we have $\tty(\bfh,0) = 0$, but $y(\bfh)$ need not be 0. This does not matter for our analysis.}

From the above two paragraphs, we conclude that if $\Psi(\tty,\tau) = (y,\tau)$, then $(\tty,\tau)$ and $(y,\tau)$ induce the same expected payoff for the principal, and $(\tty,\tau)$ is feasible under $\modori$ if and only if $(y,\tau)$ is feasible under $\original$. Hence, $(\tty,\tau)$ is constrained-optimal for $\modori$ among solutions in $K$ (which by \cref{lem:moni-1} also means that it is globally optimal for $\modori$) if and only if $(y,\tau)$ is constrained-optimal for $\original$ among solutions in $L$.

As a result, if we find a constrained-optimal solution $(y,\tau)$ to $\original$ among solutions in $L$, then $\Psi^{-1}(y,\tau)$ is an optimal solution to $\modori$ and gives the same payoff to the principal. In particular, $\original$ admits an optimal solution that is standard, and a standard solution is in $L$, so we have the following result.
\begin{prop}\label{prop:moni-2}
    If a standard solution $(y,\tau_t)$ solves $\original$, then the solution $(\tty,\tau_t)$, where $\tty$ is given by $\tty(\bfh,0) = 0$ and $\tty(\bfh,1) = y(\bfh)$, solves $\modori$ and gives the same expected payoff to the principal.
\end{prop}

\section{A Sufficient Condition for APM}
\label[onlineapp]{OA:bgd}

We present a sufficient condition for the agent's utility function $u$ to satisfy APM.

\begin{lem}
    \label{lem:u-is-apm} If $g(\theta,t)/\phi(t)$ is non-increasing in $t$, then the agent's payoff function, $u(\theta,y,t),$ satisfies APM.
\end{lem}

\begin{proof}
    Fix some finite sequence $(\alpha_i,\theta_i,y_i)_i$, with $\alpha_i>0$ for all $i$, and suppose that \begin{equation}
    \label{eq_t_prime_bdg}\sum_i \alpha_iu(\theta_i,y_i,t)\geq 0,\end{equation} for some $t\in T.$ We would like to show that \[\sum_i \alpha_iu(\theta_i,y_i,s)\geq \sum_i \alpha_iu(\theta_i,y_i,t),\]
    for any $s<t.$ 

    Define $W:=\sum_i\alpha_iw(\theta_i,y_i)$ and $G(\hat{t}):=\sum_i\alpha_i g(\theta_i,\hat{t})/\phi(\hat{t}).$ Observe that $G$ is non-increasing, since $g/\phi$ is non-increasing in $t$. Then, rewriting \eqref{eq_t_prime_bdg}, we find:
    \begin{equation*}
        \sum_i\alpha_i[\phi(t)w(\theta_i,y_i)+g(\theta_i,t)] =\phi(t)[W+G(t)]\geq 0.
    \end{equation*}
    Then, for any $s<t$, we have:
    \begin{align*}
		\phi(t)[W+G(t)]&\geq 0 \\ 
	   W+G(t) &\geq 0 &[\phi(t)>0] \cr
		\phi(s)(W+G(s)) &\geq	\phi(t)(W+G(t)) & [\phi,G \text{ decreasing}] 
	\end{align*}
and so: \[\sum_i \alpha_iu(\theta_i,y_i,s)\geq \sum_i \alpha_iu(\theta_i,y_i,t),\]
as desired. 
\end{proof}

\section{Analysis of First-Best Benchmark}\label[onlineapp]{OA:1st-best}
In the first-best problem, the principal commits to a \textit{first-best policy rule} $Y(\bfh,t)$,\footnote{In this appendix, we will always use $Y$ to denote a first-best policy rule ($Y:\ms H\times(T\cup\{\infty\})\to\R_+$) and $y$ to denote a policy rule ($y:\ms H\to\R_+$).} which depends on both the history $\bfh$ and the agent's investment time $t$: 
\begin{align}
    [\ms P_0^{FB}] \hspace{3em}\max_{Y,\tau} &\, \mathbb{E} [v(\t,Y(\bfh,\tau(\bfh)),\tau(\bfh))] \hspace{15em}  \notag\\
    \text{subject to } &\, \E[u(\t,Y(\bfh,\tau(\bfh)),\tau(\bfh))] \geq \E[u(\t,Y(\bfh,\tau'(\bfh)),\tau'(\bfh))] \qquad \forall \tau' \in \ms T.\notag
\end{align}
Note that this problem differs from the problem $[\ms P^{FB}]$ described in the main text; we will call $[\ms P^{FB}]$ the auxiliary first-best problem.

It is natural for the principal's policy to maximally punish the agent's deviations. Furthermore, as we will show, the principal should induce a simple stopping time. We call a pair $(Y,\tau_t)$ with these two properties a \textit{standard contract}.

\begin{defn}[Standard contract]
    A \textbf{standard contract} is a pair $(Y,\tau_t)$ such that $Y(\bfh,t')=0$ if $t'\neq\tau_t(\bfh).$
\end{defn}

Fixing a simple stopping time $\tau_t$, there is a bijection between standard contracts $(Y,\tau_t)$ and policy rules $y$, given by interpreting the policy rule as the policy in the standard contract, conditional on the agent's obedience. That is, the mapping from a standard contract to a policy rule, and its inverse, are given by: 
\[y^{(Y,\tau_t)}(\bfh):=Y(\bfh,\tau_t(\bfh)), \qquad Y^{(y,\tau_t)}(\bfh,t'):=\begin{cases} y(\bfh) & t'=\tau_t(\bfh) \\
0 & t'\neq\tau_t(\bfh).
\end{cases}\] Notice that the agent's expected payoff from playing $\tau_t$ under the standard contract $(Y,\tau_t)$ is equal to $U_t(y^{(Y,\tau_t)}).$

Our first lemma says that a standard contract is optimal for the principal. 
\begin{lem}
\label{lem:fb-standard}
    There exists a standard contract which solves $[\ms P_0^{FB}].$
\end{lem}
\begin{proof}
    We will argue that any $(Y,\tau)$ which solves $[\ms P_0^{FB}]$ but is not a standard contract can be weakly improved by a standard contract. Since a solution to $[\ms P_0^{FB}]$ exists by the Weierstrass theorem,\footnote{The argument for existence is analogous to that of the proof of \cref{lem:standard-solutions-exist}(a).} this will establish the result. 

    There are two ways in which $(Y,\tau)$ could fail to be a standard contract.  First, $\tau$ could fail to be a simple stopping time. Second, $\tau$ could be a simple stopping time, but $Y(\bfh,t')>0$ for some $t'\neq \tau(\bfh)$ (i.e., agent disobedience is not fully punished). In the second case, the principal can weakly improve by using $(\widetilde{Y},\tau)$, where $\widetilde{Y}(\bfh,t')=0$, for all $t'\neq \tau(\bfh)$, and otherwise $\widetilde{Y}=Y$. (This is because fully punishing deviations maintains the agent's incentives, and the principal's payoff is not affected on-path.)

    In the first case, let $d$ be the stopping rule induced by $\tau$, and let $t$ be the period where $d(\bfa_{t})=\ttstop$ (or $t=\infty$ if the agent never stops). Let $\ms R:=\{(\hat{r},\hat{s}):d(\bfb_{\hat{r},\hat{s}})=\ttstop\}$ denote the set of other partial histories at which $d$ stops. Let $(r,s)$ be the principal's favorite instance of the agent stopping after a breakdown under $Y$ -- i.e., $(r,s)\in\argmax_{(\hat{r},\hat{s})\in\ms R} v(0,Y(\bfb_{\hat{r}},\hat{s}),\hat{s}).$
    
     We claim that the principal can improve by using the standard contract $(\widetilde{Y},\tau_0)$, where $\widetilde{Y}$ is given by:
    \[\widetilde{Y}(\bfh,0)=\begin{cases} 
Y(\bfb_{r},s) & \bfh\in\{\bfb_1,\dots,\bfb_t\}, \\
Y(\bfh,t) & \text{otherwise,}  
\end{cases} \qquad \widetilde{Y}(\bfh,\hat{t})=0 \text{ for all }\bfh\text{ and }\hat{t}\neq 0,\]
if $t<\infty$; and $\widetilde{Y}(\bfh,0)=Y(\bfb_r,s)$ for all $\bfh$, and $\widetilde{Y}(\bfh,\hat{t})=0$ for all other $\hat{t}$, if $t=\infty.$ That is, the principal induces the agent to stop at 0 by promising him, if he is obedient, $Y(\bfb_r,s)$ when a breakdown occurs by time $t$ and otherwise the obedient policy under $Y$; if he is not obedient, he receives 0. The remainder of the argument closely parallels the proof of \cref{lem:principal_wants_simple}, and is omitted.
\end{proof}

Now, recall the principal's auxiliary first-best problem described in the main text, 
\begin{align}
    [\ms P^{FB}] \hspace{13em}\max_{y,\,\tau_t} &\, \mathbb{E} [v(\t,y(\bfh),\tau_t)] \hspace{15em}  \notag\\
    \text{subject to } &\, U_t(y) \geq \max_s U_s(\mathbf{0}). \notag
\end{align}
We aim to relate $[\ms P^{FB}]$ and $[\ms P_0^{FB}]$; the following lemma will be useful.

\begin{lem}\label{lem:fb-feas}
    Suppose $(Y,\tau_t)$ is a standard contract which is feasible for $[\ms P_0^{FB}]$. Then, $(y^{(Y,\tau_t)},\tau_t)$ is feasible for $[\ms P^{FB}]$. Suppose $(y,\tau_t)$ is feasible for $[\ms P^{FB}]$. Then, $(Y^{(y,\tau_t)},\tau_t)$ is feasible for $[\ms P_0^{FB}]$. 
\end{lem}
\begin{proof}
    \textit{    Suppose $(Y,\tau_t)$ is a standard contract which is feasible for $[\ms P_0^{FB}]$. Then, $(y^{(Y,\tau_t)},\tau_t)$ is feasible for $[\ms P^{FB}]$. }

    Let $y:=y^{(Y,\tau_t)}.$ Because $(Y,\tau_t)$ is feasible for $[\ms P_0^{FB}]$, we have $U_t(y)\geq U_{t'}(\mathbf{0})$ for all $t'$ (the agent cannot improve by deviating to a different simple stopping strategy). Thus, $(y,\tau_t)$ is feasible for $[\ms P^{FB}]$.
    \medskip
    
    \textit{Suppose $(y,\tau_t)$ is feasible for $[\ms P^{FB}]$. Then, $(Y^{(y,\tau_t)},\tau_t)$ is feasible for $[\ms P_0^{FB}]$.}
    
    Let $Y:=Y^{(y,\tau_t)}.$ Consider any alternative stopping rule, $d$, to $\tau_t,$ and let $\hat{t}$ be the period where $d(\bfa_{\hat{t}})=\ttstop$ (or $\hat{t}=\infty$ if the agent never stops). Let $\ms R:=\{(r,s):d(\bfb_{r,s})=\ttstop\}$ denote the set of other partial histories at which $d$ stops. Then, the agent's utility from following $d$ under $Y$ is given by:
    \[U_{\tau_d}=\sum_{(r,s)\in\ms R} p(\bfb_r)u(0,0,s)+\begin{cases}U_{\hat{t}}(\mathbf{0}) & \hat{t}\neq t\\
    U_t(y) & \hat{t}=t\end{cases}.\]
    The agent's utility from following $\tau_t$ is $U_t(y).$ Then, in either case, we have that $U_t(y)\geq U_{\tau_d}$, since the first term is negative by \cref{ass:A-payoff-main}(f), and $U_t(y)\geq U_{\hat{t}}(\mathbf{0})$ by the feasibility of $y$ for $[\ms P^{FB}].$ Thus, $(Y,\tau_t)$ is feasible for $[\ms P_0^{FB}]$.
\end{proof}

Now, we can prove \cref{thm:fb-equivalence}, which shows an equivalence between $[\ms P^{FB}]$ and $[\ms P_0^{FB}]$.
\begin{thm}
\label{thm:fb-equivalence}
    Suppose $(y,\tau_t)$ solves $[\ms P^{FB}].$ Then, $(Y^{(y,\tau_t)},\tau_t)$ solves $[\ms P_0^{FB}]$. Suppose $(Y,\tau_t)$ is a standard solution to $[\ms P_0^{FB}].$ Then, $(y^{(Y,\tau_t)},\tau_t)$ solves $[\ms P^{FB}]$.
\end{thm}
\begin{proof}
        \textit{Suppose $(y,\tau_t)$ solves $[\ms P^{FB}].$ Then, $(Y^{(y,\tau_t)},\tau_t)$ solves $[\ms P_0^{FB}]$.}

        Let $Y:=Y^{(y,\tau_t)}.$ First, by \cref{lem:fb-feas}, $(Y,\tau_t)$ is feasible for $[\ms P_0^{FB}]$. 
        
        Next, by \cref{lem:fb-standard}, a standard solution to $[\ms P_0^{FB}]$ exists; call it $(Y^*,\tau_{t^*}).$ Again, by \cref{lem:fb-feas}, $(y^{(Y^*,\tau_{t^*})},\tau_{t^*})$ is feasible for $[\ms P^{FB}]$. Then, by optimality of $(y,\tau_t)$ for $[\ms P^{FB}]$, 
        $(y,\tau_t)$ is at least as good as $(y^{(Y^*,\tau_{t^*})},\tau_{t^*})$ for $[\ms P^{FB}]$.
        Thus, $(Y,\tau_t)$ is at least as good as $(Y^*,\tau_{t^*})$ for $[\ms P_0^{FB}]$, and hence optimal.
        \medskip

        \textit{Suppose $(Y,\tau_t)$ is a standard solution to $[\ms P_0^{FB}].$ Then, $(y^{(Y,\tau_t)},\tau_t)$ solves $[\ms P^{FB}]$.}

    Let $y:=y^{(Y,\tau_t)}$. First, by \cref{lem:fb-feas}, $(y,\tau_t)$ is feasible for $[\ms P^{FB}]$. 
    
    Next, by the Weierstrass theorem, a solution to $[\ms P^{FB}]$ exists;\footnote{The argument is analogous to that of the proof of \cref{lem:standard-solutions-exist}(a).} call it $(y^*,\tau_{t^*}).$ Again, by \cref{lem:fb-feas}, $(Y^{(y^*,\tau_{t^*})},\tau_{t^*})$ is feasible for $[\ms P_0^{FB}]$. Then, by optimality of $(Y,\tau_t)$ for $[\ms P_0^{FB}]$, 
    $(Y,\tau_t)$ must be at least as good as $(Y^{(y^*,\tau_{t^*})},\tau_{t^*})$ for $[\ms P_0^{FB}]$.
    Thus, $(y,\tau_t)$ is optimal for $[\ms P^{FB}].$
\end{proof}
Similarly to \cref{thm:aux=original}, \cref{thm:fb-equivalence} is useful in two respects. First, it allows us to find (standard) solutions to $[\ms P_0^{FB}]$ by solving the more tractable auxiliary problem, $[\ms P^{FB}].$ Second, it shows that any properties we establish of solutions to the auxiliary problem also hold of standard solutions to $[\ms P^{FB}_0].$

\subsection{\texorpdfstring{Proof of \cref{prop:benchmark-solution}}{Proof of Proposition \ref{prop:benchmark-solution}}}\label[onlineapp]{proof:benchmark-solution}
    Fix time $t \in T$, and consider the \textit{inner first-best problem}:
\begin{align}
    [\ms P^{FB}(t)] \hspace{13em}\max_{y} &\, \mathbb{E} [v(\t,y(\bfh),\tau_t)] \hspace{15em}  \notag\\
    \text{subject to } &\, U_t(y) \geq \max_s U_s(\mathbf{0}). \notag
\end{align}
Consider the Lagrangian of $[\ms P^{FB}(t)]$. Suppose $y(\bfh)>0$ for each history $\bfh \in \{\bfb_{t+1},\bfb_{t+2},\ldots,\bfb_T,\bfa\}$. Then, the first-order condition of the Lagrangian with respect to $y(\bfh)$ is
    \begin{align}\label{eq:prop-2-3}
        v_y\left(\theta, y(\bfh),t\right)+\lambda u_y\left(\theta, y(\bfh),t\right)=0,
    \end{align}
    where $\lambda$ is the Lagrange multiplier on the  agent's incentive constraint. The LHS of \eqref{eq:prop-2-3} is strictly decreasing in $y(\bfh)$, so for each $\bfh$, there exists a unique $y(\bfh)$ which solves \eqref{eq:prop-2-3}. Since $\t=0$ whenever $\bfh \in \{\bfb_{t+1},\bfb_{t+2},\ldots,\bfb_T\}$, the first part of \eqref{eq:prop-2-1} follows. The second part of \eqref{eq:prop-2-1} follows from  rearranging \eqref{eq:prop-2-3}.

\section{\texorpdfstring{Proof of \cref{lem:principal_wants_simple} Claim}{Proof of Lemma \ref{lem:principal_wants_simple} Claim}}\label[onlineapp]{OA:simple_claim}

We claim that because $(y,\tau)$ solves $[\ms P_0],$ we must have $v(0,y(\bfb_r),s)\geq0$ and $V_t(y)\geq 0.$
\begin{proof}
    Suppose not. Then, the principal can strictly improve by changing her policy to $\hat{y}$, which is equal to $y$ except at the following histories: (1) $\hat{y}(\bfb_{t'})=0$ for all $t'\leq t$ where the agent never stopped under $\tau$ in $\bfb_{t'}$; (2) $\hat{y}(\bfb_{t'})=0$ for all $t'\leq t$ where the principal received negative payoff under $(y,\tau)$ in $\bfb_{t'}$ (i.e., $v(0,y(\bfb_{t'}),\tau(\bfb_{t'}))<0$); (3) if $V_t(y)<0$, then $\hat{y}(\bfa)=0$; and (4) if $V_t(y)<0$, then $\hat{y}(\bfb_{t'})=0$ for all $t'>t$.

    The principal recommends that the agent not delay stopping after any breakdown (i.e., not stop at $\bfb_{r',s'}$ for some $s'>r'$); this is incentive-compatible by APM (i.e., delaying stopping is weakly dominated by stopping immediately following the breakdown). The principal also recommends that, if $V_t(y)\geq 0,$ the agent not stop at some $\bfa_{t'}$ with $t'>t$. This is incentive-compatible because, in this case, $\hat{y}=y$ for all histories $\bfh\in\{\bfb_{t+1},\dots,\bfb_{T},\bfa\}$ and so, under $\hat{y}$, stopping at $\bfa_t$ inherits superiority over stopping at $\bfa_{t'}$ from its superiority under $y$.

    We claim that, whatever stopping time $\hat{\tau}$ the agent uses subject to following these two recommendations, the principal strictly benefits relative to $(y,\tau).$ In particular, in any history $\bfh\in\{\bfb_1,\dots,\bfb_t\}$, either (a) the agent stopped under $\tau$ and the principal received non-negative payoff from this; now under $\hat{\tau}$ the agent stops weakly earlier, so the principal weakly improves by APM; or (b) the agent did not stop under $\tau$, or stopped but the principal received negative payoff; now whatever the agent does under $\hat\tau$, the principal receives non-negative payoff since policy is 0. 

    To analyze the remaining histories, we consider two cases. \medskip

    \noindent \textit{Case 1.} $V_t(y)\geq 0$, $v(0,y(\bfb_r),s)<0.$
    
    In this case, since $V_t(y)\geq 0$, in any history $h\in\{\bfb_{t+1},\dots,\bfb_{T},\bfa\},$ the agent stops at $\bfa_t$ under both $y$ and $\hat{y}$, and so the principal receives equal payoff and in particular does no worse under $\hat{y}$.
    But, the histories described in (b) above where the principal received negative payoff under $y$ and hence strictly improved under $\hat{y}$ occur with positive probability, and so the principal strictly improves under $(\hat{y},\hat{\tau}).$\medskip

        \noindent \textit{Case 2.} $V_t(y)< 0$.
        
    In this case, as in (b) above, in any history $h\in\{\bfb_{t+1},\dots,\bfb_{T},\bfa\},$ the principal receives non-negative payoff. But, by the case supposition, she received negative expected payoff over all these histories under $(y,\tau)$. Thus, the principal strictly improves under $(\hat{y},\hat{\tau})$.
\end{proof}


\section{A Simple Example with Agent Rent}\label[onlineapp]{OA:rent}

We provide a simple example which serves two purposes. First, it shows that the agent can obtain positive rent even if his limited liability constraint ($y\geq 0$) is slack, and the principal's payoff $v$ is decreasing in $y$. Second, it shows that condition (iii) in \cref{prop:rent} is not necessary for the agent to obtain positive rent.

The example is similar to the motivating example in \Cref{sec:motivating-example}, with three key differences.
First, the state space is $c\in\{1,10\}$; second, the cost of the policy to the principal is $cy^2$, now quadratic; finally, policy and early investment are complements for the agent:
\begin{align*}
        \text{Agent's payoff} &\,= \begin{cases}
            ay - I \caseif \text{he invests in period 0} \\
            y - I  \caseif \text{he invests in period 1} \\
            0 \caseif \text{he never invests,}
        \end{cases}
\end{align*}
where $a>1$ scales the policy's benefit to the agent if he invests in period 0.
The principal's cost-minimization problem of inducing investment in period 0 is:
\begin{align*}
    \min_{y} &\, py(1)^2+10(1-p)y(10)^2 \\
    \text{subject to } &\,  pay(1)+(1-p)ay(10)-I\geq 0 \tag{IR} \\
    &\, pay(1)+(1-p)ay(10)-I\geq p(y(1)-I). \tag{NW}
\end{align*}

Consider the following parameters: $ I=2,\, a=2,\, p=\frac13$. The optimal policy is $y(1)=\frac{20}9,\, y(10)=\frac49$. 
Importantly, the agent's equilibrium payoff is
\[pay(1)+(1-p)ay(10)-I=\frac{2}{27} > 0,\]
so the IR constraint is slack. That is, the agent receives positive rent.

Note that, although the principal's payoff is decreasing in $y$, limited liability of the agent does not bind -- $y$ is strictly positive under both states. Also note that, under the null policy rule which is $y\equiv 0$, the agent will never invest, so condition (iii) of \cref{prop:rent} is violated.

\section{Moral Hazard Can Hasten Investment}\label[onlineapp]{OA:hastening}

\cref{subsec:investment-timing} provided a sufficient condition for moral hazard to delay investment. Here, we provide an example in which moral hazard strictly hastens investment -- the second-best investment timing is strictly earlier than the first-best investment timing.

Let $T=2$ and $p_1=p_2=1/2$. The payoffs are
    \begin{align*}
        u(\t,y,t) &\, = y - 12 + 8\theta + 4.9\ind{t=0} + 3\ind{t=1} \\
        v(\t,y,t) &\, = -y^2 + 26.9 \ind{t=0 \text{ or } t=1}.
    \end{align*}
    That is, for the agent, investment costs $12$, gives a benefit of 8 if $\t=1$, and also gives flow benefits of 1.9 in $t=0$ and 3 in $t=1$. The principal faces a quadratic cost of policy and obtains a benefit of 26.9 if the agent invests early enough, in $t=0$ or $t=1$. One can easily verify that the first-best investment timing is $t^{FB}=1$, while the second-best investment timing is $t^{SB} = 0$. 

    \begin{table}[h]
        \centering
        \subfloat[First-best]{\begin{tabular}{|l|l|l|l|l|}
        \hline
        Stopping time & $y(\bfb_1)$   & $y(\bfb_2)$   & $y(\bfa)$ & P's payoff \\ \hline
        $t=0$         & 5.1 & 5.1 & 5.1 & 0.89   \\ \hline
        $t=1$         & 0    & 5    & 5 & 0.95   \\ \hline
        \end{tabular}}
        \vspace{0.2cm}
        
        \subfloat[Second-best]{\begin{tabular}{|l|l|l|l|l|}
        \hline
        Stopping time & $y(\bfb_1)$   & $y(\bfb_2)$   & $y(\bfa)$ & P's payoff \\ \hline
        $t=0$         & 5.47 & 5.47 & 4 & 0.49   \\ \hline
        $t=1$         & 0    & 6    & 4 & 0.45   \\ \hline
        \end{tabular}}

        \caption{Hastening}
        \label{tab:hastening}
    \end{table}

    To understand this result intuitively, see \cref{tab:hastening}, which shows the optimal policy for the $t=0$ and $t=1$ inner problems in both the first-best and second-best cases. (The principal's benefit from early stopping is large enough that these are the only stopping times she should consider.) Notice that, since the principal's payoff is state-independent, the first-best policy is constant across histories. In contrast, the second-best policy is state-measurable, but has $y(0)>y(1).$ 

    To select between inducing stopping in $t=0$ and $t=1$, the principal must weigh a number of considerations. Some  considerations -- namely, that the agent requires more aggregate rewards to stop at 0, but having him stop at 1 risks an early breakdown occurring and hence the principal losing her early-stopping benefit -- are quantitatively identical between the first-best and second-best cases. The key consideration which differs relates to the principal's ability to \textit{smooth} her policy. Holding fixed the aggregate policy level the agent faces, the principal would prefer to smooth this policy over time (due to the convexity of $c$). Since inducing the agent to stop at 0 gives the principal more periods to smooth policy over, this consideration pushes the principal to prefer stopping at 0.
    
    This smoothing effect is more significant in the second-best setting. In the second-best setting, if the principal targets stopping in $t=1$, she must set $y(0)=6$ (a relatively high value) as a result of the agent's no waiting constraint. Additionally, whether the principal targets stopping in $t=0$ or $t=1$, she sets $y(\bfa)$ to the same level. This means that, in the second-best case, by moving from $t=1$ to $t=0,$ a high policy ($y(0)=6$) is smoothed from 1 to 2 histories; while in the first-best case, a lower policy ($y=5$) is smoothed from 2 to 3 histories.\footnote{Notice that these policies are both smoothed and (in aggregate) raised when moving to $t=0;$ this raising is why we see an overall increase from $y=5$ to $y=5.1$ in the first-best case.} 
    As a result, smoothing is more valuable in the second-best case, leading to hastening.

\end{document}